\documentclass[11pt]{article}
\usepackage{amsmath,amssymb,amsthm,epsfig}
\usepackage{enumitem}

\newtheorem{theorem}{Theorem}
\newtheorem{lemma}[theorem]{Lemma}
\newtheorem{corollary}[theorem]{Corollary}

\newtheorem{observation}[theorem]{Observation}

\newtheorem{claim}[theorem]{Claim}
\newtheorem{remark}[theorem]{Remark}

\theoremstyle{definition}
\newtheorem{definition}[theorem]{Definition}

\newtheorem{algspec}[theorem]{Algorithm} 

\newenvironment{compactitemize}{\begin{itemize}[noitemsep,topsep=0pt,parsep=0pt,partopsep=0pt]} {\end{itemize}}
\newenvironment{compactenumerate}{\begin{enumerate}[noitemsep,topsep=0pt,parsep=0pt,partopsep=0pt]} {\end{enumerate}}
\renewenvironment{proof}{{\bf Proof.}}{\hfill \proofbox \vskip0.2cm }
\newcommand{\proofbox}{\hbox{\vbox{\hrule\hbox{\vrule\phantom{\vrule height 8.6pt
                         width 6pt depth 0pt}\vrule}\hrule}\quad}}

\renewcommand{\epsilon}{\varepsilon}

\newcommand{\comment}[1]{}

\newcommand{\sncycle}{decomposition node cycle}

\newcommand{\separator}{separator}
\newcommand{\tric}{T}
\newcommand{\tricname}{root graph}
\newcommand{\Wlog}{without loss of generality}
\newcommand{\WLOG}{Without loss of generality}

\newcommand{\FAdjacent}{\textsc{Adjacent}}
\newcommand{\FConn}{\textsc{Conn}}
\newcommand{\FRedConn}{\textsc{RedConn}}

\newcommand{\FPart}[1]{\textsc{#1-Part}}
\newcommand{\FRed}{\textsc{Red}}
\newcommand{\FNonRed}{\textsc{NonRed}}

\newcommand{\FIsRedEdge}{\textsc{IsRedEdge}}
\newcommand{\FPlanar}[1]{\textsc{Planar}}
\newcommand{\FPlanarRed}{\textsc{PlanarRed}}
\newcommand{\FNonRedConn}{\textsc{NonRedConn}}
\newcommand{\FFirstComp}{\textsc{FirstComp}}
\newcommand{\FSecondComp}{\textsc{SecondComp}}
\newcommand{\FFirstCompConn}{\textsc{FirstComp}}
\newcommand{\FSecondCompConn}{\textsc{SecondComp}}
\newcommand{\FSeparate}{\textsc{Separate}}
\newcommand{\FClosest}{\textsc{Closest}}
\newcommand{\FFerocious}{\textsc{Ferocious}}
\newcommand{\FIsFerocious}{\textsc{IsFerocious}}

\mathcode`l="8000
\begingroup
\makeatletter
\lccode`\~=`\l
\DeclareMathSymbol{\lsb@l}{\mathalpha}{letters}{`l}
\lowercase{\gdef~{\ifnum\the\mathgroup=\m@ne \ell \else \lsb@l \fi}}%
\endgroup

\title{Connectivity Preserving Iterative Compaction and Finding 2 Disjoint Rooted Paths in Linear Time}
\author{Ken-ichi Kawarabayashi, Zhentao Li, Bruce Reed}
\begin{document}
\maketitle
\begin{abstract}
In this paper we show how to combine two algorithmic techniques to obtain linear time algorithms for various optimization problems on graphs, and present a subroutine which will be useful in doing so.

The first technique is iterative shrinking. In the first phase of an iterative shrinking algorithm, we construct a sequence of graphs of decreasing size $G_1,\ldots,G_\ell$ where $G_1$ is the initial input, $G_\ell$ is a graph on which the problem is easy(often because it is small), and $G_i$ is obtained from $G_{i+1}$ via some shrinking algorithm. In the second phase we work through the sequence in reverse, repeatedly constructing a solution for a graph from the solution for its successor. In an iterative compaction algorithm, we insist that the shrinking algorithm is actually a compaction algorithm, i.e. for some constant $\delta>0$, for every consecutive pair, $G_i$, $G_{i+1}$ of the sequence we have $|V(G_{i+1})|+|E(G_{i+1}| \le (1-\delta)(|V(G_i)|+|E(G_i)|)$.

Another approach to solving optimization problems is to exploit the structural properties implied by the connectivity of the input graph. Thus, both isomorphism testing and planar embedding are easier for 3-connected planar graphs because these graphs have a unique embedding. This approach can be used on graphs which are not highly connected by decomposing an input graph into its highly connected pieces, solving subproblems on these specially structured pieces and then combining their solutions. This usually involves building, explicitly or implicitly, a tree decomposition of bounded adhesion and working with the pieces into which it splits the input.

We combine these two techniques by developing compaction algorithms which when applied to the highly connected pieces preserve their connectivity properties. The structural properties this connectivity implies can be helpful both in finding further compactions in later iterations and when we are manipulating solutions in the second phase of an iterative compaction algorithm which uses such a compaction algorithm as a subroutine. In particular, we show that for any $c>0$ there is a $d>0$ such that for any graph $G$ we can obtain a minor $J$ of $G$, which is three connected if $G$ is, satisfying one of: (i) $J=G-F$ for a set $F$ of edges such that for every edge $xy$ of $F$, $x$ and $y$ are joined by $c$ paths of $G-F$, (ii) $J=G-S$ for a stable set $S$ such that every vertex $v$ of $S$ has degree at most $d$ and every pair of its neighbours are joined by $c$ paths of $G-S$, (iii) $J=G/ N$ for an induced matching $N$ such that every vertex in $V(N)$ has degree at most $d$, or (iv) $J$ is obtained by contracting exactly two edges in each of a set $T$ of disjoint triangles, the vertices of which all have degree 3.

This is the first compaction algorithm for all graphs. Previous algorithm may return the fact that the input graph is not in the class of interest instead of the next graph $G_{i+1}$.

To illustrate how this compaction algorithm can be used as a subroutine, we present a linear time algorithm that given four vertices $\{s_1,s_2,t_1,t_2\}$ of a graph $G$, either finds a pair of disjoint paths $P_1$ and $P_2$ of $G$ such that $P_i$ has endpoints $s_i$ and $t_i$, or returns a planar embedding of an auxiliary graph which shows that no such pair exists. This is the first linear time algorithms for solving this problem.
\end{abstract}

\section{The Introduction}


In this paper, we show how to  combine two algorithmic techniques to obtain linear time algorithms for
various optimization problems on graphs, and present a subroutine which will be  useful in doing so.

{\it Iterative shrinking}  algorithms have two stages. We first construct a sequence of smaller and smaller inputs  until we obtain one on which the problem we are
interested in is easy to solve (often because this final input is very small).  We then work through our sequence of inputs in reverse, obtaining a solution to each problem by exploiting the solution we have already obtained to the next input in the sequence.
As we discuss more fully below, this is a very common paradigm in algorithm design.

An  {\it iterative compaction} algorithm\footnote{this term was coined  by Haralambides and Makedon \cite{compaction} to describe a particular instance of this paradigm, we extend it to general approach for the first time.}, is an iterative shrinking algorithm such that  for   some constant $\delta>0$, for every
consecutive pair, $G_i$, $G_{i+1}$ of the sequence we have $|V(G_{i+1})|+|E(G_{i+1}| \le (1-\delta)(|V(G_i)|+|E(G_i)|)$.

Many iterative compression algorithms\footnote{An \emph{iterative compression} algorithm is an algorithm which takes a solution of objective value $k$ to a minimization problem and get a solution of value $k-1$ or states no such solution exists} are also iterative shrinking algorithms.
For example, in \cite{rsv} which is often credited with introducing iterative compression, Reed, Smith and Vetta present an iterative shrinking algorithm for finding a set $S$ of at most $k$ vertices in an input graph $G$ such that $G-S$ is bipartite  or determining that no such set exists.
The algorithm presented there orders the vertices of $G$ as $v_1,\ldots,v_n$ and considers
the sequence of graphs $G_1,\ldots,G_n$ where $G_i$ is the subgraph of $G$ induced by $\{v_1,\ldots,v_{n+1-i}\}$. Thus, for $i<n$,  $G_i$  is obtained from
$G_{i+1}$ by adding the vertex $v_{n+1-i}$. Now, if the algorithm has determined that  $G_{i+1}$ cannot be made bipartite by deleting a set $S$  of at most $k$ vertices then neither can $G_i$ and the algorithm  simply records this fact. Otherwise, we have a set $S$ of at most $k$ vertices of $G_{i+1}$   such that $G_{i+1}-S$ is bipartite. But now we have a set  $S' =S+v_{n+1-i}$ of at most $k+1$ vertices  such that $G_i-S'$ is bipartite and this makes it much easier to look for the desired  set $S$ in $G_i$.

\paragraph{Iterative compaction for computing treewidth.} An earlier, more involved, and more interesting application of iterative shrinking (in fact, iterative compaction) is found in Bodlaender's linear time algorithm \cite{bod} which, for any fixed $k$, given  an input graph $G$ either determines that $G$ has  tree width exceeding  $k$ or finds a tree decomposition of  width at most $k$ for $G$\footnote{We omit definitions  related to tree width in this introductory chat}. The algorithm for  $k$ uses as a subroutine  a linear time algorithm which,  for some $c_k,d_k, \epsilon_k>0$, given an input graph $G$ either (i) finds a minor $F$ of $G$ with more than $k|V(F)|$ edges, (ii)   finds a matching of size $\epsilon_k |V(G)|$ in $G$ such that every vertex in the matching has degree at most $d_k$, or (iii) finds a stable set $S$  of size at least $\epsilon_k  |V(G)|$ in $G$  such that each vertex of this stable set has degree at most $d_k$ and has the same neighbourhood as more than  $c_k$ vertices of $G-S$.

Bodlaender's algorithm first applies this subroutine. If the subroutine  obtains output (i) then the algorithm  notes that this implies that $G$ has tree width exceeding $k$ simply returns this fact and stops.

If the subroutine obtains output (iii), then the algorithm  calls itself recursively on $G-S$. If $G-S$ has tree width exceeding  $k$ then so does $G$, so the algorithm simply returns this fact. Otherwise,
the recursive call returns a tree decomposition of $G-S$ of width at most $k$. Bodlaender showed that it is a straightforward matter to obtain a tree decomposition of $G$ of width $k$ from any such tree decomposition of $G-S$\footnote{For every $v$ in $S$ there will be a node $t=t(v)$ of the tree decomposition such that $W_t$ contains all the neighbours of $v$. For each such $v$ we will add  a leaf $s(v)$ to the tree incident to $t(v)$ and set $W_{s(v)}=N(v) \cup v$.}.

If the subroutine obtains output (ii) it contracts
each edge of the matching into a single vertex. It then calls itself on the resultant graph $G^*$. If $G^*$ has tree width exceeding $k$ then so does $G$ and Bodlaender's algorithm returns this fact.  Otherwise, the algorithm returns a tree decomposition of width $k$ for $G^*$. Now, it turns out that since each vertex of $G^*-G$ corresponds to the two endpoints of an edge of $G$, this tree decomposition immediately yields a tree
decomposition of $G$ of width at most $2k$. This tree decomposition of $G$ can be exploited by a dynamic programming algorithm which
determines if $G$ has a tree decomposition of width at most $k$, and if so constructs it.

\subsection{Connectivity properties}

Another approach to solving optimization problems is to exploit the structural properties implied by the connectivity of the input graph. Thus, both isomorphism testing and planar embedding are easier for 3-connected planar graphs because these graphs have a unique embedding. This approach can be used on graphs which are not highly connected by decomposing an input graph into its highly connected pieces, solving subproblems on these specially structured pieces and then combining their solutions. This usually involves building, explicitly or implicitly, a tree decomposition of bounded adhesion and working with the pieces into which it splits the input.

\subsection{Our result}

In this paper, we combine the above two techniques by developing compaction algorithms which when applied to the highly connected pieces preserve their connectivity properties.

Our algorithm was motivated by, and has the same flavour as Bodlaender's subroutine, but has two important advantages, firstly it can be applied to all graphs, 
not just those of bounded tree width, and secondly it maintains 3-connectivity and other 
connectivity properties in each step. In order to achieve these goals, we need a significantly
more complicated procedure and proof. 

To handle dense graphs in linear time,  given such a graph the algorithm applies an 
algorithm of Frank, Ibaraki, and Nagamochi~\cite{fin} to find a set of  edges to be deleted , such that doing so  does not  affect the 3-connectivity of the graph. Specifically, we delete a  large set $F$ of edges  such that the endpoints of every edge of $F$ are joined by $c$ paths in $G-F$ for some $c>>3$. This means   that the set of cutsets of size at most $c-1$ in $G$ and $G-F$ are the same. 

To handle sparse, unstructured graphs,  our  algorithm, if it returns a stable set $S$ of size at least $\epsilon|V|$  to be deleted, simply insists that for every vertex $v$ of $S$, every two neighbours of  $v$  are joined by $c$ disjoint paths of $G-S$, rather than that there are $c$ vertices of $G-S$ which are adjacent to all of the neighbours of $v$.  To see that this modification is necessary, consider a random bipartite graph
where one side $A$ has  $n^{4/5}$ vertices, and each of the vertices on the other side 
chooses 3 neighbours in $A$ uniformly at random. Almost surely, this  graph will be 3-connected,
but no two vertices will have the same neighbourhood. Furthermore, deleting any edge destroys 
3-connectivity and there is no matching of size exceeding $|A|$ which is $o(|V|)$. 

Given a sparse graph, our algorithm either finds a  large stable set of this type,
or  large matching the vertices of which have bounded degree. If  , our algorithm finds  such a large matching,  it then need to do a significant amount of  extra work to massage  this  matching  to ensure that if $G$ is 3-connected then so is the  minor we output. It turns out 
  this may require us to contract a set   of triangles into vertices (by contracting two of the 3 edges of each triangle) rather than the edges of a matching.To see that this is necessary, consider the 
graph obtained from $K_{n-3,3}$ by replacing each vertex  $v$ of degree three, by a triangle and a matching between the triangle and the neighbours of $v$. Again deleting any edge destroys
$3$-connectivity, as does deleting any stable set, or contracting the edges of any matching with more than 6 edges. Contracting each triangle down to a vertex, however, keeps the graph 3-connected. 

We will prove that for every $c$ there are $\Delta,n_0$ and $ \delta>0$  such that  an algorithm with the following specifications exists:

\begin{algspec}\label{spec:compactor}
\textsc{Compactor}
~\\
\emph{Input:} A 3-connected graph $H$ with at least $n_0$ vertices and a set $C$ of at most 5 vertices of $H$.
~\\
\emph{Output:} One of the following:
\begin{enumerate}
\item[(i)]  a set $F$ of  at least $\delta(|V(H)|+|E(H)|)$ edges disjoint from $C$  such that for every edge $xy$ of $F$, $x$ and $y$ are joined by $c$ paths of $H-F$,
\item[(ii)]  a stable set $S$ disjoint from $C$ of at least $\delta(|V(H)|+|E(H)|)$  vertices such that  every  vertex $v$ of $S$ has degree at most $\Delta$ and every pair of its neighbours is joined by $c$ paths of $H-S$,
\item[(iii)] a matching $N$ disjoint from $C$  with at least $\delta(|V(H)|+|E(H)|)$ edges such that every vertex of $N$ has degree at most $\Delta$ and contracting   the edges of
$N$ yields a 3-connected graph, or
\item[(iv)] a set $T$ of at least $\delta(|V(H)|+|E(H)|)$ disjoint  triangles disjoint from $C$  such that every vertex of each triangle has degree three  and contracting each triangle to a vertex yields  a 3-connected graph.
\end{enumerate}
\emph{Running time:} $O(|E(G)|+|V(G)|)$.
\end{algspec}

 We will also  show how  our compaction procedure   can be used as a subroutine in a linear time  iterative compaction algorithm for  the 2 Disjoint Rooted Paths problem (2-DRP).  In a forthcoming  paper we use this algorithm for 2-DRP as a subroutine in an algorithm which for any fixed $\ell$,  finds a separator of size at most $c_l\sqrt{n} $ in a graph with no $K_{\ell}$ minor.

\subsection{2-DRP}

\paragraph{2-DRP problem definition.} An instance of  2-DRP  consists of a graph $G$, and four specified  vertices  $\{s_1,s_2,t_1,t_2\}$  of $G$, called terminals.  We are asked to determine if there are two vertex disjoint paths $P_1$ and $P_2$ of $G$   such that $P_i$ has endpoints $s_i$ and $t_i$. Our algorithm  either finds the desired paths  or returns a planar  embedding of an auxiliary graph which shows that no two such paths exist.

 Our auxiliary  graph is derived from the graph $G'$ obtained from $G$ by adding a new vertex  $v^*$, edges from $v^*$ to each terminal, and the edges
$s_1s_2,s_2t_1,t_1t_2$ and $t_2s_1$ (see Fig. \ref{fig:auxgraph}).  We  note that  the desired two paths exist in $G$ precisely if $G'$ contains a subdivision of $K_5$ \footnote{A subdivision of a graph $H$ is obtained from $H$ by replacing an edge of $H$ by a path.}
whose set of centers is  $C=\{v^*,s_1,s_2,t_1, t_2\}$ (as such a subdivision contains the two desired disjoint paths, and given the paths we can obtain a subdivision whose edge set is the edges of the paths and the  edges in $E(G')-E(G)$).

\begin{figure}
\begin{center}
\epsfig{file=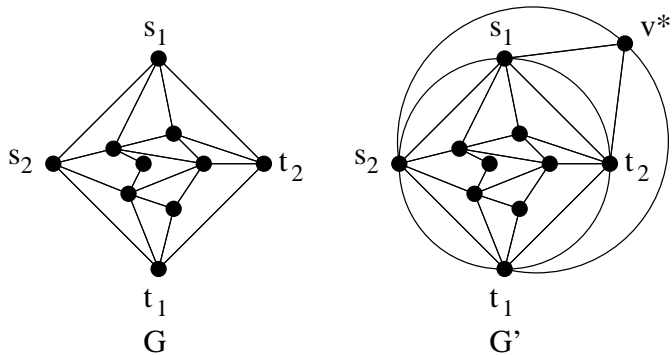}
\end{center}
\caption{An instance of 2-DRP and the corresponding graph $G'$}
\label{fig:auxgraph}
\end{figure}

\paragraph{Planarity and connectivity reductions.} So, if $G'$ is planar the desired two paths do not exist.  As discussed below,  if $G'$ is 4-connected then the desired two paths exist precisely if $G'$ is planar.  On the other hand, if $G'$ has small cutsets then  the desired paths may not exist even if $G'$ is non-planar.  Indeed the argument of the above paragraph  shows  the desired  paths cannot exist if the component of $G'$   containing $v^*$ is planar, even if all the other components are huge cliques (and therefore non-planar). In the same vein, since the   edges of a  $K_5$ subdivision form  a 2-connected graph, we see that if the block $B$ of $G'$ containing  $E(G')-E(G)$ is planar then the desired two paths do not exist. We can  reduce from $G'$ down to $B$ by repeatedly deleting  vertices separated from $v^*$ by a cutset of size 0 or 1. The auxiliary graph we are interested in, is obtained by  reducing $B$ further using    cutsets of size 2 and 3.

\paragraph{Reduction on 2-cuts} We consider first reducing on cutsets of size 2.\footnote{In this paper, when we say ``cut'', we always mean a vertex cut. So if we say ``$k$-cut'', this means a vertex cut of order $k$.} If  $X$ is a cutset of size 2 in $B$, then at most one component of $B-X$ intersects $C$.  Letting $U$ be any other component of $B-X$, we see that if the desired  subdivision exists and intersects $U$ then this intersection is
 exactly the interior of a path between the vertices of $X$.  Since $B$ is 2-connected, there is always a path with interior in  $U$ linking the two vertices of $X$. It follows that the desired subdivision exists in $B$ if and only if it exists in the graph obtained from $B$ by deleting $U$ and adding an edge
 between the vertices of $X$ if they are not adjacent in $G'$.

Iteratively applying this replacement until no further replacement are possible  we arrive at the unique triconnected component  $\tric$ of $G'$ containing the vertices of $C$. We call any graph $\tric$ that can be obtained this way (from an initial instance of 2-DRP) a \emph{\tricname} with \emph{root vertices} $C$.  Every component $U$ of $B-V(\tric)$ has edges to precisely two vertices $x_u$ and $y_u$ of $\tric$, which are joined by an edge of $\tric$. Furthermore, the edges of $T$ are precisely those edges of $B$ with both endpoints in $\tric$, along with an edge  between $x_U$ and $y_U$ for every component $U$ of $B-V(\tric)$. The remarks of the last paragraph show that the desired $K_5$ subdivision exists in $B$ precisely if it exists in $\tric$.

Hopcroft and Tarjan \cite{2conn} have shown that we can find $\tric$ in linear time\footnote{We discuss their work and prove that $\tric$ is unique below.}. It is a straightforward matter, given $\tric$ and the desired $K_5$ subdivision  within it, to replace the edges of $E(\tric)-E(G)$ in the subdivision by paths with their interiors in the components of $G-V(\tric)$ and hence obtain the desired $K_5$ subdivision in $G'$, in linear time.  Thus, we need only determine if the desired $K_5$ subdivision exists in $\tric$.

\paragraph{Reduction on 3-cuts} The other key to doing so is to  perform  reductions using the cutsets of size 3 in $\tric$.

\begin{definition}
If $H$ is a 3-connected graph containing the vertices of $C$ and the edges of $G'$ joining them, then a \emph{reduction} of $H$ is a 3-connected graph $F$ such that
\begin{enumerate}
\item[(i)] $V(F)$ is a subset of $V(H)$ containing $C$,
\item[(ii)] for every component $U$ of $H-V(F)$, the set $S(U)$ of vertices in $F$ with a neighbour in $U$ forms a triangle of $F$, and
\item[(iii)] $E(F)$ consists of the edges of $E(H)$ with both endpoints in $V(F)$ and for every component $U$ of $H-V(F)$, an edge between
every pair of vertices of $S(U)$ not adjacent in $H$.
\end{enumerate}

Given a component $U$ of $H-V(F)$ we say that $U$ \emph{attaches} at $S(U)$ and $S(U)$ is the \emph{\separator} of $U$ with respect to $(H,F)$ (as it separates $U$ from $F-S(U)$).
\end{definition}

\begin{remark}
We note that we do not need to impose the condition that $F$ is 3-connected here, as it is implied by (i), (ii), and (iii).
More strongly, any cutset $X$ of $F$ is a cutset of $G$ because for any component $U$ of $G-V(F)$ since $S(U)$ is a clique we have that $U$ has neighbours in only one component of $F-X$.
\end{remark}

\begin{observation}
For any reduction $F$ of a root graph $\tric$, the desired $K_5$ subdivision exists in $\tric$ precisely if a subdivision with the same centers exist in $F$.
\end{observation}

\begin{proof} If the desired subdivision  exists in $\tric$, then we can obtain  a $K_5$ subdivision in  $F$  with the same centers by replacing any path of the subdivision with its interior in  a component $U$ of $\tric-V(F)$ by  an edge between two vertices of $S(U)$. Conversely, if the desired subdivision  exists in $F$, then  $F-v^*$ contains two vertex disjoint paths $Q_1$ and $Q_2$ such that $Q_i$ has endpoints $s_i$ and $t_i$.   If we take the $Q_i$ as short as possible then they are induced and  hence for any $U$ they use at most one of the three edges joining the vertices of $S(U)$. So, we can find two vertex disjoint paths of $\tric-v^*$ with the same endpoints , and hence the desired  $K_5$ subdivision in  $\tric$,  by replacing any edge of $E(F)-E(\tric)$ joining two vertices in some $S(U)$ by a  path with its interior in  $U$.
\end{proof}

\begin{theorem}\cite{Sh, se, Th}
\label{seymourthm}
The desired two paths exist in $G$ or equivalently there is a subdivision of $K_5$ in $\tric$
 whose set of centers  is $C$  precisely if there is no planar reduction of $\tric$\footnote{By our observation,
 if the desired subdivision  exists in $G'$, then every reduction of $\tric$ is non-planar. The other direction is more difficult.}.
 \end{theorem}


 Our algorithm finds either the desired $K_5$ subdivision  in $\tric$, or a planar reduction of $\tric$.
 In doing so, it works with a sequence of 3-connected minors of $\tric$, each of which ``contains'' the
 vertices of $C$ and the edges of $G'$ between them.

\paragraph{Our Algorithm} To be more precise, we  construct and consider a sequence of graphs  $G_1=\tric,G_2,G_3,\ldots,G_\ell$ such that
\begin{enumerate}
\item[(a)] each $G_i$ is a  3-connected minor of $G_{i-1}$,
\item[(b)] the five vertices of $C$ have been contracted into different vertices of $G_i$ and none of the edges of $\tric$ with both endpoints in $C$ have been deleted,
\item[(c)] for some $\delta>0$ and all $i$ from $1$ to $\ell-1$, $|E(G_i)|+|V(G_i)|< (1-\delta)(|V(G_{i-1}|+|E(G_{i-1})|) $, and
\item[(d)]  $G_\ell$ has bounded size.
\end{enumerate}

We refer to  3-connected minors of $T$ satisfying (b) as    {\it compactions} of $\tric$. To simplify
exposition we abuse notation slightly and  for every  vertex $x$ of $C$, we also use $x$ to denote the vertex of  the compaction into which $C$  has been contracted,

After constructing this sequence, our algorithm finds and outputs, for each $i$ from $\ell$ down to 1,  either
a subdivision of $K_5$ in $G_i$ with set of centers   $C$\footnote{or equivalently two disjoint paths  $P_1$ and $P_2$ of $G_i-v^*$ such that
$P_i$ has endpoints $s_i$ and $t_i$},    or  a planar reduction $H_i$
of $G_i$.  If $i$ is $\ell$ this can be done in constant time  because $G_\ell$ has bounded size. For $i<\ell$, we need to exploit the output  we have obtained
for  $G_{i+1}$ when considering $G_i$. In order to be able to do so,  if we output a reduction $H_i$ for $G_i$ we insist that it  is of a certain special type   and  we record  some extra information for each component $U$ of $G_i-H_i$. We  provide  further details on the exact nature of the output below.

We will show that for each $i<\ell$ we can construct  $G_{i+1}$ in time which is linear in $|E(G_i)|+|V(G_i)|$ and
can obtain the output for $G_i$ from the output for $G_{i+1}$ in time which is linear in $|E(G_i)|+|V(G_i)|$. By property  (c)  of  the graphs in the sequence we construct above,  it follows that our algorithm runs in linear time in total.

\paragraph{Paper structure} It remains to describe our algorithm for compacting  $G_i$ down to  $G_{i+1}$, and our algorithm
for obtaining the output corresponding to $G_i$ by uncompacting  the output corresponding to $G_{i + 1}$.
We provide specifications and outlines  of the compaction algorithm in Section \ref{compactionoutline}, details can be found in Section \ref{compactiondetails}.
We provide specifications and outlines  of the algorithms which uncompacts the outputs  in Section
\ref{uncompactionoutline}, details can be found in Section \ref{uncompactiondetails}.
Before doing any of this, we present some necessary background. In Section \ref{preserving3conn} we examine certain properties ensuring that contracting a matching or a set of  triangles preserves 3-connectivity.
Once we've described our compaction algorithm, in Section \ref{reductionssec}, we will  investigate more closely the structure of the reductions we consider, and the extra information we need to output in order for this dynamic programming approach to work.
 In the rest of this introductory section, we survey previous work related to
solving 2-DRP.

\subsection{Previous Work on 2-DRP}

Theorem \ref{seymourthm} tells us that if we are given a reduction $F$ of a \tricname{} $\tric$ such that $F$  has no proper reduction (i.e. a reduction of $F$ which is not $F$
itself),
then to test if the desired $K_5$ subdivision exists in $F$ and  equivalently $\tric$, we need only
test if $F$ is planar, which can be done in linear time.
The classical approach to solving the 2-DRP problem is to
construct such  a reduction $F$ of $\tric$  and apply  a planarity testing algorithm to it.

It is easy to test if a given reduction $F$ of $\tric$ has a proper reduction in polynomial time. We are simply asking if
there is a vertex $v$ of $F$ which can be separated from $C$ by a cutset $X$ of size three\footnote{We remark that this is essentially equivalent to
testing if $F$ is 4-connected. If $F$ is a 3-cut of $F$ which does not separate any vertex $v$ from $C$ then it separates two vertices of $C$
and hence must be $v^*,s_i,t_i$ for some $i$ and separate $s_{3-i}$ from $t_{3-i}$. But then every component $U$  of $F-C$ can only see one of $s_{3-i}$ and $t_{3-i}$. Hence $F-C$ is non-empty or we  could reduce by  "cutting off"  such a $U$. So not permitting a reduction  is equivalent to 4-connectivity unless
our reduction $F$ has vertex set $C$.}.
To see this note first that  if $J$ is a proper reduction of $F$ then every vertex $v$ of
$V(F)-V(J)$ can be separated from $C$ by a 3-cut of $J$.
On the other hand if $v$ can be separated from $C$ by a 3-cut $X$ of $J$
then, letting $U$ be the component of $F-X$ containing $v$, the graph $J$ obtained from $F-V(U)$
by adding edges so that $X$ is a clique is a proper reduction of $F$; $J$ is 3-connected because $F$ was and we added edges so that
$X$ is a clique.

 It is easy to see that a reduction of a reduction of $\tric$ is also a reduction of $\tric$\footnote{Again this is because we put a clique on the neighbourhood of every component cut off by the first reduction.}. Thus, iteratively applying the above test, and creating a proper
reduction of the current reduction if one exists, we create a sequence of smaller and smaller reductions of $\tric$. This process must eventually terminate with a reduction which has no reduction to which we can apply our planarity testing algorithm.


 Jung\cite{jung} gave the first algorithm in the 4-connected case. There, no 3-cut exists so the input graph itself is a reduction with no reductions on which we may test planarity.

 If we successively find reductions of reductions as described here, we can determine if the desired two paths exist in $O(mn)$ time. Indeed, we need only test if each vertex can be separated  from $C$ by a 3-cut once, as our reductions  cannot introduce new small cuts. Each of these tests, and consequent reduction if any,  takes $O(m)$ time.
 Having tested and possibly reduced on each vertex, Jung's result tells us that we 
 can determine if the paths exist using a planarity testing algorithm. If we actually want to find the paths when they exist, then for each edge $e$ in turn, we can determine if the paths exist when $e$ is deleted and if so delete it. When we are done, the only edges remaining form the desired two paths. This takes $O(m^2n)$ time.

 Using their theorem, Thomassen, Seymour, and Shiloach independently gave the first polynomial time algorithm for 2-DRP. Tholey \cite{tholey} improved this to an $O(m\alpha(m,n))$ algorithm by using a more sophisticated approach for testing for 3-cuts. Later, Tholey \cite{tholeynew} further improved his algorithm to run in $O(m+n\alpha(m,n))$ time.

As an application of our iterative compaction algorithm, we describe the first $O(m+n)$ algorithm for 2-DRP.

\section{Preserving 3-Connectivity}
\label{preserving3conn}
A matching $N$ in a 3-connected graph $J$ is {\em 3-connectivity preserving}  if the graph $J/N$
obtained by contracting its edges is 3-connected. A set ${\cal T}$ of triangles is {\em 3-connectivity preserving} if the graph obtained by contracting two  edges of each triangle in ${\cal T}$ is 3-connected. In this section we set out some conditions ensuring that a matching or set of disjoint triangles
is  3-connectivity preserving.

The case of triangles is easy to handle.

\begin{theorem}\label{tri3con}
  Suppose $H$ is 3-connected and $\cal{T}$ is a set of disjoint triangles in $H$ where each vertex of a triangle has degree 3.
  Then ${\cal T}$ is 3-connectivity preserving.
\end{theorem}

\begin{proof}

We need to show that $H$ contains no cutset which is  the union of two triangles  in ${\cal T}$ or of a triangle of ${\cal T}$ and possibly some other vertex.

  Since triangles in $\cal{T}$ only have degree 3 vertices, every triangle $T$ is incident to exactly three edges not in $T$.
  So, since $H$ is 3-connected, no triangle $T \in \cal{T}$ is a cutset of $H$ as otherwise, some component is adjacent to at most one vertex of the triangle  and that vertex is a cutvertex in $H$, a contradiction.

Next suppose we have a 4 cut $X$ in $H$ consisting of one of the triangles in ${\cal T}$ and some other vertex $z$. Then again there is a component  of $G-X$ seeing only one vertex of the triangle  and that vertex together with $z$  forms a 2-cut of $H$, a contradiction.

 Finally suppose we have a 6-cut $X$ in $H$ consisting of two of our triangles. There are only six  edges from these  triangles to $H-X$, so $H-X$ has 2 components each adjacent to 3 vertices in $X$: 1 in one triangle and 2 in the other. (Otherwise each  of the triangles would be a cutset which we've already shown is impossible).  For each component take the neighbour which is alone in its triangle and these two vertices form a 2-cut of $H$, a contradiction.
\end{proof}

Certain sets of edges are also easy to handle. We say that an edge $e$   is {\em sweet} if for one of its endpoints $x$ either the
neighbours of $x$ not in $e$ form a clique or are a $P_3$ whose midpoint sees no vertex outside $N(x)$. We call $x$ a \emph{sweet end} of $xy$.
The reason for this choice of notation is the following.

\begin{lemma}
\label{sweetie}
Any matching $N$  in a 3-connected graph $J$  with at least 6 vertices all of whose edges are sweet is 3-connectivity preserving.
\end{lemma}

\begin{proof}
It is enough to show that there is no cutset of $J$ consisting of two edges of the matching or one edge of the
matching and another vertex.   Suppose the contrary is true and choose a minimal  cutset $X$ of this type.

$X$ is a minimal cutset of $J$ as otherwise $X$ has size 4 by 3-connectivity of $J$ and deleting a vertex leaves a cutset of the same type.
We examine the sweet end $x$ of an edge $xy$ of $N$ contained in $X$. By minimality of $X$, $x$ is adjacent to all components of $J-X$ (in fact, both endpoints of $xy$ are adjacent to all components). Since there are at least two components of $J-X$, $N(x)-y$ is not a clique and so $J[N(x)-y]$ is a $P_3$ with midpoint $z$. Since $z$ is adjacent to exactly $N(x)-y$ in $J-y$, $z\in X$. Since $z$ has no other neighbours except possibly $y$, $J-X$ has exactly two components and $x$ has exactly one neighbour in each. Again since $z$ is adjacent to exactly $N(x)-y$ in $J-y$, $X = \{x,y,z\}$ (otherwise, the second edge of $N$ in $X$ contains $z$ and a neighbour of $z$ not in $N(x)-y$).

Since $J$ has at least 6 vertices, one of the two components of $J-X$ has at least two vertices $u,v$ with, say $u$ the only vertex of that component adjacent to $x$ (and $z$). But now $J-u-y$ separates $v$ from $x$ (and $z$), a contradiction to $J$ being 3-connected.
\end{proof}

Given a 3-connected graph $J$, a cutset $X$ and a component $U$ of $J-X$ we say that an edge $xy$ of $U$  is {\it well-behaved} with respect to $[X,U]$ if
(i) for any clique $Z$ of size at most 2 in  $V(J)-X-U$, the vertices of $X$ lie in the same component of $J-x-y-Z$, and
(ii) for any vertex $z$ of $X \cup U$, $J-x-y-z$ is connected.

\begin{lemma}
\label{wellbehaved}
Suppose that we are given cutsets $X_1,\ldots,X_l$  of a 3-connected graph $J$ and a component $U_i$
of $G-X_i$ for each $i$ such that $U_i$ is disjoint from $X_j \cup U_j$ for all $j \neq i$. Then, if $e_i$ is well-behaved for
$[X_i,U_i]$, we have that  $F=\{e_1,\ldots,e_l\}$ is a 3-connectivity preserving matching.
\end{lemma}

\begin{proof}
Clearly $\{e_1,\ldots,e_l\}$ is a matching. Suppose that it is not 3-connectivity preserving. Then there are either two matching
edges or a matching edge and a vertex which form a minimal cutset  $W$ of $G = J/F$.

\WLOG\ we can assume that $e_1$ is in $W$.
Since $W$ is minimal, each vertex in $W$ has neighbours in all components of $G-W$.

So, if $W-e_1$ is disjoint from $X_1 \cup U_1$ then there is a path from $e_1$ to $W-e_1$ inside each component of $G-W$ and each of these paths must pass through $X_1$. By (i) of the definition of well-behaved, $X_1$ is in a single component of $G-W$, a contradiction as there are at least two components of $G-W$.

So, $W-e_1$ intersects $X_1 \cup U_1$ and so cannot be an edge of any $U_j$ for $j \neq 1$.
So it is not any $e_j$ and must be a vertex $z$ of $X_1 \cup U_1$.
Hence, $G-V(e_1)-z$ is disconnected, contradicting (ii) in the definition of well-behaved.
\end{proof}

\section{The Compaction Algorithm}
\label{compactionoutline}

In this section, we present an algorithm with the following specifications for  any  fixed $c$
and parameters $\delta$ and $\Delta$  which we define in terms of $c$ below:

\begin{algspec}\label{spec:itercompactor}
\textsc{Iterative Compactor}
~\\
\emph{Input:} A \tricname{} $\tric$.
~\\
\emph{Output:} A sequence of compactions $G_1=\tric,\ldots,G_l$ of $\tric$ so that $|V(G_l)|<n_0$ and for each $i$
$G_{i+1}$ is obtained from $G_i$ by one of:
\begin{enumerate}
\item[(a)] the deletion of a set $F$ of at least $\delta(|V(G_i)+|E(G_i)|)$ edges disjoint from $C$  such that for every edge $xy$ of $F$, $x$ and $y$ are joined by $c$ paths of $G_i-F$,
\item[(b)] the deletion of a set $S$ of  at least $\delta(|V(G_i)+|E(G_i)|)$ vertices disjoint from $C$ such that  every  vertex $v$ of $S$ has degree at most $\Delta$ and every pair of its neighbours  are joined by $c$ paths of $G_i-S$,
\item[(c)] the contraction of a set $N$ of edges not incident to $C$ such that
 $N$ is a matching of size at least $\delta(|V(G_i)+|E(G_i)|)$ in $G_i$ and
every vertex of $N$ has degree at most $\Delta$ and contracting the edges of
$N$ yields a 3-connected graph, or
\item[(d)] the contraction of two edge from each of a set $\mathcal{T}$ of at least $\delta(|V(G_i)+|E(G_i)|)$  disjoint triangles such that every triangle in $\mathcal{T}$ is disjoint from $C$ and every vertex of each triangle has degree three.
\end{enumerate}
\emph{Running time:} $O(|V|+|E|)$
~\\~\\
\emph{Description:}
Having found $G_1,\ldots,G_i$, our first  step is to  check if $|V(G_l)|<n_0$. If so, our compaction algorithm has done its job and it terminates.
Otherwise we  apply our compactor algorithm to $G_i$ (with the same parameters as specified here). According to the output it returns, we delete a set $F$ of edges, delete a set $S$ of vertices, contract the edges of a matching $N$, or two edges from each triangle in a set $\mathcal{T}$ of disjoint triangles to obtain $G_{i+1}$. We record the set of contractions or deletions we
performed to obtain the new compaction and go on to the next iteration.

It remains to describe our compactor algorithm.
\end{algspec}

We can assume $c \ge 10$ as  an algorithm for $c=j$ also works for $c<j$.

 For a given $c \ge 10$, we choose $d>1000$ big enough in terms of $c$,  and $\delta>0$ small enough in terms of $c$ and $d$  so that they satisfy various inequalities
given throughout the proof. We set $\epsilon$ to be $(2c+1)\delta$.    We set $\Delta= \frac{1}{\epsilon}+6$ and ensure that we have chosen
$\delta<\frac{1}{(2c+1)d}$ so $\Delta >d$.    We define $n_0$, the  bound on the number of vertices of $G_i$ below which we fall out of the iterative
procedure and solve the problem on $G_i$  by brute force as a function of $\epsilon$ and $d$.  We do so implicitly by insisting that it
is large enough so as to satisfy certain inequalities given below.

We determine if $G_i$ has average degree exceeding $4c$. If so, we  find a
 set $F'_i$ of   half of the edges of $G_i$  which satisfy the property of (a)  using an approach due to Frank, Ibaraki, and Nagamochi. Specifically, we apply the main result of \cite{fin} to obtain $F_i'$.
We set $F_i$ to be those edges of $F_i'$ not incident to a vertex of $C$. Since $c \ge 10$, $|F_i| \ge \frac{|E(G_i)|}{4} \ge \frac{|V_i|+|E_i|}{5}$.
We set  $G_{i+1}$ to be $G_i-F_i$, and go on to the next iteration.

 If $G_i$ has more than $n_0$ vertices and average degree less than $4c$ then
we   find output as in (ii), (iii), or (iv) of the compactor algorithm using an algorithm with the following specifications.

\begin{algspec}
\textsc{Low Density Compactor:}
~\\
\emph{Input:} A 3-connected graph $H$ with at least $n_0$ vertices and at most $2c|V(H)|$ edges   and a set $C=\{v^*,s_1,s_2,t_1,t_2\}$.
~\\
\emph{Output:}  one of the outputs (ii)-(iv) of \textsc{Compactor}.
~\\
\emph{Running Time:} $O(|E(H)|+|V(H)|)$.
\end{algspec}

We note that since $\epsilon=(2c+1)\delta$ and $|E(H)| \le 2c|V(H)|$, $\epsilon |V(H)| \ge \delta(|V(H)+|E(H)|)$.

Low Density Compactor  exploits the fact that our bound on the number of edges of $H$  implies that there are at most $2c|V(H)| \over d$
vertices  of degree exceeding $d$ in $H$. Thus, for large $d$, the set $\textsc{Big}_d$ of  vertices of $H$ that have degree at least $d$ contains
only a small proportion of the vertices.
The algorithm greedily   constructs a maximal matching $M$  within the subgraph of $H$ induced by the vertices of degree at most $d$ in $H$ except for those in $C$,
by considering these vertices in non-decreasing order of degree (in $H$) and picking an edge to an unmatched neighbour of each vertex in turn if possible, choosing one with the
lowest possible degree (in $H$).

If $M$ has size less than $\gamma  |V(H)|$ for some $\gamma$ which is much larger than $\epsilon$ but still very small  then the union $B$ of $\textsc{Big}_d,C$ and  $V(M)$ is a small set of vertices such that $V(H)-B$  is a stable set. In this case, LDC shows that we can find a set $S$ of vertices
within $V(H)-B$ satisfying (ii) of the Compactor output.

If  $M$  has  size $\gamma|V(H)|$   then Low Density Compactor will attempt
to massage this matching to obtain a  3-connectivity preserving matching $N$  of size at least  $\epsilon |V(H)|$ all of whose vertices have degree  at most $\Delta$. If LDC  succeeds in constructing such an $M$ then it  returns $M$. Essentially the only reason it can fail to output $M$ is that it finds a large set
of triangles consisting of an edge of $M$ and a third vertex such that all three of these vertices
have degree 3. If this occurs, then contracting all these triangles to vertices leaves a 3-connected graph
so we return this set $T$ of triangles.

\paragraph{Obtaining output (iii) or (iv) in the large matching case}
In massaging $M$ if $M$ is large, LDC first greedily obtains an induced submatching $M'$ of $M$ of size at least $|M| \over 2d-1$ by exploiting the fact that there are at most $2d-2$ edges of $H$ from an  edge $e$ of $M$ to the edges of $M-e$. Then  it greedily chooses a submatching $M^*$ of $M'$ of size at least $|M'|\over 24d$ such that no vertex of degree at most  $12$ is joined to vertices in two distinct edges of $M^*$. It then contracts each of the edges of $M^*$ to obtain a new graph $H^*$.

If $H^*$ has no cutset of size less than 3, then we can set $N=M^*$ and terminate. The key to massaging our matching when
$H^*$ is not 3-connected is to examine the cutsets of size at most 2 in  $H^*$,  and how they interact.

\paragraph{2-cuts in $H/M^*$}
If $v$ is a cutvertex  of $H^*$ then it is either  a cutvertex of $H$ or was obtained by contracting a matching  edge whose endpoints form a cutset in $H$.
So, since $H$ is 3-connected, $H^*$ is 2-connected.  In the same vein,  any 2-cut of $H^*$ corresponds either to (i)  two edges of  $M^*$ whose
endpoints form a 4-cut of $H$,  or  (ii) a three cut of $H$ which contains both endpoints of some edge of $M^*$ and a vertex in no edge of $M$.

\paragraph{2-cuts decomposition trees}
In order to examine how these 2-cuts  interact, it is helpful to discuss various types of tree decompositions  for $H^*$ which decompose it using
2-cuts, and which LDC constructs. These trees  are analogs of the block tree which is the unique tree showing how the cutvertices of $G$
decompose it into its maximal 2-connected components. We digress briefly to introduce their basic properties and then explain how
LDC will use them.

\paragraph{2-cuts decomposition uniqueness}
When we try to build such a tree for a 2-connected graph by considering the 2-cuts within it, complications arise.
Firstly, when we decompose on a 2-cut $\{x,y\}$ we will want to add the edge $xy$ to each piece of the  decomposition  to record the fact that these vertices are
connected by a path through the rest of the graph. However, even if we do this, we cannot construct a unique tree which decomposes a graph  into three-connected
pieces. To see this, it is enough to consider a long cycle. Every pair of non-adjacent vertices of the cycle correspond to a 2-cut and every triangulation of the
cycle corresponds to a decomposition of the cycle into three connected pieces(triangles) which can be represented by a tree--  the vertices are the triangles and chords of the
triangulation and a chord is adjacent to a triangle if it is an edge of the triangle. We will need to settle either for a unique tree decomposition or for 3-connected pieces,
we cannot have both.

It turns out that cycles are actually the only reason that we cannot get both a unique tree and 3-connected pieces. By a {\it strong 2-cut}
we mean a 2-cut whose vertices are joined by three internally vertex disjoint paths ( these paths need not be distinct if the vertices are joined by an edge, all three paths
could just be the edge; of course if the vertices in the cut are not adjacent vertex-disjointness implies distinctness). Every 2-connected graph $J$  has a unique strong 2-cut decomposition tree (we call the abstract tree of this decomposition ''strong 2-cuts tree'').
This tree has a bipartition where one side is the 2 cuts of $J$ and the other is a set of graphs. Each graph is either a cycle or is 3-connected, and is
joined to the set of strong 2-cuts whose endpoints it contains. The tree can be recursively defined as follows:

\begin{enumerate}
\item[(A)] if $J$ has no strong 2-cut then the tree has a single (graph) node and the node corresponding to this graph is $T$,

\item[(B)] if $x$ and $y$ form a strong 2-cut of $J$ then for each component $U$ of $J-x-y$, we let $J_U$ be the graph obtained from the subgraph of
$J$ induced by $V(U) \cup \{x,y\}$ by adding the edge $xy$ if it is not present. We construct a strong 2-block tree for each $J_U$ and then add a (cut)
node corresponding to  $\{x,y\}$ which, for each component $U$ of $J-x-y$,  has an edge to the unique (graph) node of the strong 2-cut tree for $J_U$
which contains the edge $xy$.
\end{enumerate}

\paragraph{2-cuts decomposition construction}
Hopcroft and Tarjan \cite{2conn} showed that this tree can be constructed in $O(|E(G)|)$ time, is unique, and every graph node is either 3-connected or a cycle.
These pieces are called the triconnected components of $J$.
Furthermore, the 2-cuts of $J$  which are not strong are precisely the pairs of non-adjacent vertices in those triconnected components which are cycles.

We can decompose further by choosing a triangulation of each triconnected component which is a cycle, adding the chords of this triangulation as 2-cut nodes,
and its triangles as graph nodes. More precisely, for each cycle triconnected component which is not a triangle, we choose a triangulation of it.
We construct a tree which has  nodes corresponding to the triangles and chords of this triangulation. We delete  the node of  the strong 2-cut
decomposition tree which corresponded to each triconnected component which is a cycle, and for each strong 2-cut which was joined by an arc of the tree to the
corresponding graph node, we add an edge to the node of the tree we have constructed which corresponds to the unique triangle of the triangulation
containing the cycle edge corresponding to the 2-cut (we call the abstract tree of this decomposition ''2-cuts tree'').

We note that since every cut node of the tree we construct has degree at least 2, 
a priori  the leaves of our new tree correspond  to either (i) triconnected components which contain the vertices of only one 2-cut and are separated from the rest of $J$ by this 2-cut, or  (ii) triangles coming from a triconnected component which is a cycle, such that the triangle contains one chord and two cyclic edges whose endpoints do not from a 2-cut. In the latter case, the cut actually cuts of a vertex of degree 2 in $J$.

\paragraph{Small 2-cuts tree}
Now, it is not hard to see that the number of nodes of a 2-cut tree for  a graph $J$
 is at least an eighth  of the number of its vertices which lie in 2-cuts.
So, if the number of nodes of  our 2-cut  decomposition tree for $H^*$   is less than $|M^*| \over 16$ then
there will be at least $|M^*| \over 2$  edges
corresponding to vertices of $H^*$ which are not in any 2-cut. Since $M^*$ is an induced matching and $H^*$ is 3-connected this set of edges
is a 3-connectivity preserving matching. So we let $N$ be this set and terminate.

\paragraph{Large 2-cuts tree}
Otherwise, our 2-cut  decomposition tree for $H^*$ has many nodes,
and like any large tree either contains many leaves or many disjoint  long paths all of whose nodes have degree 2 in the tree.

We could look for our matching by considering either these leaves or paths, and in an earlier version of our paper we did. However, it makes things easier if we consider
leaves instead of paths or if we  can focus on  cuts in $H$ which contain only one vertex coming from a contracted matching
edge rather than two. It turns out that with a little bit more work we can insist on one or the other.

\paragraph{Getting leaves or paths with properties}
Specifically, if we are careful about how we triangulate the cycles of the strong 2-cut tree to
obtain our 2-cut tree, we can show that either there are many leaves, or there is a submatching consisting of  at least $\frac{|M^*|}{30}$ edges of $M^*$ such that no two form a four cut of $H$. If we obtain the latter, we contract  just the submatching  edges and consider the resultant tree. All the two cuts in the resultant graph correspond to a 3-cut of $H$  consisting of one  vertex obtained by contracting a matching  edge and another which is a vertex of $H$.
Now we apply  the argument that shows that we can either obtain the desired matching or many leaves or many long paths.

It remains to sketch how we find $N$ given either many leaves of a 2-cut tree, or many long paths whose
vertices have degree 2 in the tree.

\paragraph{Many leaves case}
Assume first that we find many leaves in one of the 2-cut trees.
Each  such a leaf  $l$  corresponds  either to a triconnected component  of $H^*$ cut off from the rest of $H^*$ by a
2-cut of $H^*$ corresponding to a cutset $X_l$ of $H$ of size at most 4 or to a vertex of $H^*$ cut off by its neighbourhoods which consists of an edge of the matching and a third vertex. . By the \emph{interior} $U_l$  of the leaf, we mean  either the (non-empty) set of vertices of the triconnected component which are not in the cutset, or the vertex which is cut off. These interiors are disjoint so if we find many more than
$\epsilon n$ leaves, we can find  many more than $\epsilon n$ whose interiors have size at most
$\epsilon^{-1}$ and hence such that the vertices of the interior have degree at most $\Delta$.

We would like to find for each of these ``small'' leaves, either   a triangle of degree 3 vertices intersecting the interior,  or an edge  joining two vertices of the interior which is either sweet or well behaved with respect to $[X_l,U_l]$.
If we could do so, we could return  either a large set of triangles, a large  matching of sweet edges, or a large matching of well-behaved edges.
Because of the way we have chosen $M^*$, we will  be able to do  something quite close to this. The only slight complication, is that the sweet edge
may actually contain a vertex of $X_l$ of degree at most $d$. Thus it may intersect $d$ other  sweet edges. However, provided we have many more than
$d \epsilon n$ edges this is not a problem as we can simply chose a  disjoint subset of these sweet edges  to be our matching.

\paragraph{Many paths case}
So, we can assume that both  2-cut trees we constructed   have  few leaves and hence the second has  many disjoint long paths all of whose nodes have degree 2 in the tree.   In this case, we find the desired matching or set of triangles  by looking in the interiors of the subgraphs of $G$ corresponding to these long paths for sweet edges, well-behaved
edges or triangles.

Details are provided in Section \ref{compactiondetails}.

\section{The Compaction Details}
\label{compactiondetails}

As discussed in Section \ref{compactionoutline}, Compactor applies one of
 the algorithm presented in  \cite{fin} or
 Low Density Compactor so it is enough to give the details of  the latter.

As already noted, we construct a matching $M$ in the set of vertices of degree at most $d$ (other than those in $C$),
 by considering the vertices in nondecreasing order of degree in $H$  and picking an edge to an unmatched edge if possible,
 choosing one of lowest possible degree in $H$.It is easy to do this in linear time. 

Now, letting $\textsc{Big}_d$  be the set of  vertices of degree  at least $d$ in $H$ we see that
$X=\textsc{Big}_d \cup C \cup V(M)$ is a cover  of  the edges  of
$H$.
Since  $|\textsc{Big}_d|$    is  at most $\frac{4c|V|}{d}$,
provided $n_0>\frac{5d}{2c}$ we obtain that either  $M$ has at least  $\frac{2c|V|}{d}$  edges or
$X$ has size at most $\frac{10c|V|}{d}$.

We consider these two possibilities separately.

\subsection{Finding a highly connected subgraph from a small cover}\label{sect:highlyconnected}

If $X$ has size at most $\frac{10c|V|}{d}$ we apply the following algorithm, which
may be of independent interest and then set  $S$ to be
the vertices of $H$ which are not in the output subgraph $J$. We note that if $d$ is large  enough in terms
of $c$  then $S$ has at least  $\frac{|V(H)|}{2}$  vertices.

\begin{algspec}
\textsc{Embedding a Small Cover in a Highly Connected  Small Subgraph .}
~\\
\emph{Input:} An integer c, a  graph $F$ and a partition of $F$ into a cover $X$ and a stable set $S$.
~\\
\emph{Output:} A  subgraph $J$ (which will be 3-connected if $F$ is)  with $X \subseteq V(J)$ such that
 $|V(J)| \le 2c!|X|$ and for any vertex $v$  of $F-J$, every two neighbours of $v$ in $J$
 are joined by $c$ paths of $J$.
~\\
\emph{Running Time:} $O(c \Delta_S(|E(G)|+|V(G)|))$ where $\Delta_S$ is the maximum degree
of a vertex of $S$.
~\\~\\
\emph{Description:}
\paragraph{Ensuring 1-connectivity}
We first find a  subgraph $J_1$ of $F$ containing $X$ with at most $2|X|$ vertices such that for every
component $K$ of $F$, $J_1 \cap K$ is connected. To do so, we consider a  (first) auxiliary graph $F_1$ with vertex set $X$ which contains   all the edges of $F$ joining two vertices of $X$ as well as an edge between every two   neighbours of each vertex of  $S$. We label each edge of $F_1-E(F)$ by a vertex of $S$ adjacent to both its endpoints.  For any component  $K$ of $F$. Since $K$ is connected and $S$ is a stable set, $F_1 \cap K$ is connected. We find a spanning forest $T_1$  for $F_1$.   Now we consider the subgraph $J_1$  of  $G$ induced by the union of $X$ and the vertices of $S$ labelling the  elements of $E(T_1)-E(G)$.  $J_1$ has at most $2|X|$ vertices and
for every component $K$ of $F$, the intersection of  $J_1$ and $K$
 is connected.

\paragraph{Ensuring 2-connectivity}
We next find a subgraph $J_2$ of $F$ containing $X$ with at most $4|X|$ vertices such that for every block $B$ of $F$, $J_2 \cap B$ is 2-connected.
 To do so, we consider a second  auxiliary graph $F_2$ with vertex set $V(J_1)$  which contains   all the edges of $J_1$ as well as an edge between every two   neighbours of each vertex of $S-V(J_1)$. We label each edge of $F_2-E(J_1)$ by a vertex of $S-V(J_1)$  adjacent to both its endpoints. Since $S$ is a stable set,  for every block $B$ of $F$,
the subgraph of $F_2$ induced by $B$ is 2-connected.  As shown by Frank et al.,  we can, in linear time, find  a spanning  subgraph $F_2^*$ of $F_2$ with at most $2|V(F_2)|$  edges such that  the vertex sets of the blocks of  $F_2$ and $F^*_2$ are the same.   We let $J_2$ be the subgraph spanned by $V(J_1)$ and the vertices of $S$ labelling those edges of $F_2^*$  which are not in $J_1$. We claim that  for every 2-connected component $B$ of $F$, $J_2 \cap B$ is 2-connected.
Assume to the contrary that $x$ and $y$ in $B \cap J_2$ are separated
by some cutvertex $v$ of $J_2$.
Since  each vertex in  $S \cap B$  is adjacent to two vertices of $B$ and  $S$ is stable,  we can assume that $x$ and $y$ are
in $V(B) \cap J_1$. Since $F_2^* \cap V(B)$ is 2-connected, it follows that there must be an edge  $wz$ of $F_2^*-v$ such that $w$ and $z$ are in different
components of $J_2-v$.    Clearly, this edge must be labelled $v$. So $v$ is  not in $V(J_1)$. But now,  since $J_1 \cap K$ is  connected
for every component $K$ of $F$,
there is a path of $J_1$ and hence $J_2-v$ from
$x$ to $y$, which is a contradiction.  So $J_2$ is 2-connected.

\paragraph{Ensuring 3-connectivity}
We next find a subgraph $J_3$ of $F$ containing $X$ with at most $12|X|$ vertices such that  every two vertices of $X$ which are joined by 3 internally disjoint paths of
$F$ are joined by 3 internally vertex disjoint paths of $J_3$.    To do so, we consider a third  auxiliary graph $F_3$ with vertex set $V(J_2)$  which contains   all the edges of $J_2$ as well as an edge between every two   neighbours of each vertex of $S-V(J_2)$. We label each edge of $F_3-E(J_2)$ by a vertex of $S-V(J_2)$  adjacent to both its endpoints. Since  $S$ is a stable set,  every two vertices of $X$ joined by three internally vertex disjoint paths of  $F$ are joined by three  internally vertex disjoint paths of $F_3$. As shown by Frank et al., we can, in linear time, find  a spanning subgraph $F_3^*$ of $F_3$ with at most $3|V(F_3)|$ edges such that every two vertices joined by three internally vertex disjoint paths of $F_3$ are joined by three internally vertex disjoint paths of $F_3^*$. We let $J_3$ be the subgraph spanned by $V(J_2)$ and the vertices of $S$ labelling the edges of $F_3-E(J_2)$.We claim that every two vertices $x$ and $y$ of $X$ which are joined by three paths of $F$ are joined by three paths of $J_3$.
Assume to the contrary that this is false for some $x$ and $y$ and  let  $u$ and $v$ be   a cutset of $J_3$ which separates them.
Since $F_3^*$ is 3-connected, there is an edge $xy$ of $F_3^*-u-v$ whose endpoints are in different components of $J_3-u-v$.  \WLOG, this edge  is labelled
$v$. So $v$ is  not in $V(J_2)$. But now,  since  for every block $B$ of $F$, $J_2[V(B)]$ is a block of
$J_2$,   there is a path of $J_2-u$ and hence $J_3-u-v$  from $x$ to $y$,
which is a contradiction. So,   every two vertices of $X$ which are joined by three internally disjoint paths of
$F$ are joined by three internally vertex disjoint paths of $J_3$.We remark that if $F$ is 3-connected, then since  $S$ is stable, and every vertex of $S$ sees at least three vertices of $X$, this implies that $J_3$, and every  induced subgraph of $F$ containing $J_3$ will also be 3-connected

stable, i

\paragraph{Ensuring $c$-connectivity}
We continue in this fashion for $c$ iterations ensuring that for any two vertices of $S$ joined by  $i$ vertex disjoint paths
of $F$ are also joined by $i$ vertex disjoint paths of $J_i$. Further, the size of the
auxiliary graph increases by a factor of at most $i$ in   iteration $i$ for $i>1$.

Now, consider any vertex $v$  of $F-J_c$. Suppose for a contradiction that for  two of its neighbours
$u$ and $v$  the maximum size of a set of internally vertex disjoint paths from $u$ to $v$ in $F_c$  was at most
$i$ for  some $i<c$. Then $u$ and $v$ are joined by $i+1$ vertex disjoint paths of $F$ because we can
add the path $uvw$ to this set. But, since $J_{i+1} \subseteq J_c$ this contradicts the fact that  every two vertices joined by $i+1$ internally vertex disjoint paths of $F$ are joined by $i+1$ internally
vertex disjoint paths  in $J_{i+1}$.

Furthermore,
If $F$ is 3-connected then for any 2-cut $Y$ in $J_3$, some component
of $J_3 -Y$ contains only vertices of $S$. Since $S$ is stable, this component must be a vertex $v$,
Since $F$ is 3-connected and $S$ is stable, $v$ has three neighbours in $X$ and hence in $V(J_3)$ which is a
contradiction.

So, we can return $J=J_c$. This completes our description of the algorithm.
\end{algspec}

We can construct $F_i$ in $O(c \Delta_S|E(F)|+|V(F)|)$ time and it has at most this many edges.
Since there are $c$ iterations, and our treatment of an auxiliary graph takes time  linear in its size,
the  claimed running time bound for the algorithm is correct.

\subsection{Matchings}

The next few subsections detail what we do in case we find a matching $M$ in $H-\textsc{Big}_d$ with at least
$\frac{2c|V|}{d}$ edges.  As mentioned in Section \ref{compactionoutline} we first obtain  a submatching
$M^*$ with at least  $\frac{c|V|}{24d^3}$ edges which is induced and such that no vertex of degree at most  12 is adjacent to two edges of $M^*$.

Then  we contract $M^*$  and build  a special 2-block tree of the resulting graph $H^*$. We then  either find a  large  submatching  $M^+$  of $M^*$  no two edges of which form a 4-cut in $H$  or a large number of  leaves  of the tree.

If we find a large number of leaves, we will find our matching or set of triangles by considering the
interiors of these leaves separately.
If we obtain a matching $M^+$, we contract it and consider a strong 2-cut tree and a  special 2-cut tree
for the resultant graph.
We find either a  large submatching $N$ whose contraction preserves 3-connectivity or a large number  of leaves in the special 2-cut tree, or a large number of long paths whose vertices have degree 2. in the strong 2-cut tree.  In either of the latter cases, we look
for our matching or set of triangles in the interiors of the corresponding subgraphs of $H$.

Forthwith the details.

\subsection{2-cut tree properties}\label{sect:2blocktree}

In this section, we describe and prove properties of the 2-cut tree  needed in the next section.
Some of our results discuss a set $S$ of vertices, when we apply them these will be the vertices
obtained by contracting the edges of a matching.

Recall that a 2-connected graph may have different 2-cut trees.
For example, each triangulation of the 6-cycle gives a different
2-cut tree. We can obtain all 2-cut trees of a graph $J$ from its
strong 2-cut tree  by triangulating each cycle which corresponds to a node of the tree.  In other words, the 2-cut tree only changes depending on
the triangulation chosen.

There is only one triangulation for a cycle of length three so we focus on
 nodes containing cycles of length at least  four
in this section.

We refer to a cycle  of length at least four corresponding to a node of  the strong 2-cut tree  as a \emph{\sncycle}.

We note also that the number of leaves in the 2-cut tree changes
depending on this choice of triangulation. For example if $J$ is a cycle of length $2k$,
the resultant triangulation could result in a tree which is a path and hence has two leaves( if all the chords
come out of the same vertex) or  as in the case of the special decompositions we
consider,  a tree with $k$ leaves.

\begin{definition}
  An \emph{intermediate tree} is obtained from the strong 2-cut tree by triangulating each \sncycle\ $C$ of length at least six 
   corresponding to graph nodes within the strong 2-cut tree as follows. Let
  $C=c_1,\ldots,c_k$. Split $C$ along the $\lfloor k/2 \rfloor$ cuts
  $c_1c_3,c_3c_5,c_5c_7,\ldots$ (ending with either $c_{k-2}c_k$ if
  $k$ is odd or $c_{k-1}c_1$ if $k$ is even).

  A 2-cut tree is {\em special} if it is obtained from an intermediate tree by triangulating all cycle nodes (i.e., both remaining ``interior'' cycles and \sncycle s of length four or five) in any way.
\end{definition}

We can build a special 2-cut tree for $J$ in linear time by first
building the strong 2-cut tree  and then triangulating the \sncycle s as described
above. We note that in the strong 2-cut tree, the interior cycle corresponding
to a \sncycle\ of length $k>4$ has degree $\lfloor \frac{k}{2} \rfloor$.

Our interest in these trees is explained by the following lemma:

\begin{lemma}
\label{specialty}
If the sum of the number of nodes in the union of all the  decomposition node cycles of length at least 6
   of the strong 2-cut tree,  and the number of  strong-2 cuts which either  decompose $J$  into at least  three pieces, or are incident  in the strong 2-cut tree   to   a  non-cycle  node  incident to more than  two strong cuts nodes  is at least $7l$
then every special 2-cut tree for $J$ has at least $l$ leaves.
\end{lemma}

\begin{proof}
Every leaf of the intermediate tree contains a leaf of each special tree so it is enough to show the intermediate tree has at least  $l$ leaves. Otherwise,
 it has at most $l$ vertices of degree three and hence the strong 2-cut tree contains at most $l$ nodes which are strong 2-cuts which decompose the graph into  three pieces or
are cycles of length at least 6, or are non-cycle nodes incident to more than 2 strong cuts.

So the sum of the number of  strong 2-cuts which  either decompose $G$ into three pieces,  or are incident to a non-cycle node incident to more that two 2-cuts and
$|C| - (|C| \mod 2)$ over all decomposition node cycles $C$ of length at least six,
is at least $6l$.

It follows that the total degree of the  nodes of the intermediate tree which are either
\begin{compactitemize}
\item
 strong 2-cuts which decompose $G$ into at least three pieces,
\item
 non-cycle nodes incident to more than two 2-cuts or
\item
  correspond to the interior cycle of a decomposition node cycle of length at least six,
\end{compactitemize}
is at least $3l$.
But since there are only $l$ nodes of these three types, we again obtain that there are $l$ leaves.
\end{proof}

This result implies:

\begin{lemma}
\label{specialty2}
For any set $S$ of nodes of $J$ either every special 2-cut tree for $J$ has at least $\frac{|S|}{15}$
nodes or we can find in linear time  a set of $\frac{|S|}{15}$ nodes of $S$ no two of which form a strong 2-cut or lie in
a decomposition node cycle.
\end{lemma}

\begin{proof}
By Lemma \ref{specialty} we can assume that the union of
\begin{compactitemize}
\item
all the  decomposition node cycles of length at least 6 of the strong 2-cut tree and
\item
the strong 2-cuts which either
\begin{compactitemize}
\item
decompose $J$ into at least three pieces or
\item
are incident to a non cycle node incident to more than two 2-cuts
\end{compactitemize}
\end{compactitemize}
contain at most $\frac{7|S|}{15}$ vertices of $S$.

We delete the nodes corresponding to these decomposition node cycles and strong 2-cuts from the tree
and delete the vertices   they contain from $S$  to obtain a set $S'$ . We note that the non cycle nodes which are incident to more than two strong 2-cuts are now isolated vertices of the resultant forest
and contain no vertices of $S'$ which are in strong 2-cuts. We delete these nodes from the
tree although we do not delete the vertices they contain from $S'$.

Every vertex of $S'$ appears in the nodes of a  (possibly empty) tree of the forest we have obtained and no node of the tree
contains more that 5 vertices. Standard results about chordal graphs or graphs of bounded tree
width allow us to 5 colour the vertices of $S'$ in linear time
 so that no two vertices in the same colour class
lie in a node together. The largest colour class has  more than  $\frac{|S|}{15}$ vertices and has the property in (ii).
\end{proof}

Lemma \ref{specialty} also implies.

\begin{lemma}\label{2cutstructcombined}
  Let $J$ be a 2-connected graph  and suppose that the cut trees for $G$ have size $s$ then for     every special  2-cut tree $T$  for  $J$, either 
  \begin{enumerate}
  \item[(i)]
  $T$ has at least $\frac{s}{2000}$ leaves, or
  \item[(ii)]
    the strong 2-cut tree for $J$ contains at least $\frac{s}{2000}$ disjoint paths, each with 141 nodes, all of which begin and    end with a cutting node and contain only vertices which have degree    2 in the tree.
  \end{enumerate}
Furthermore,  we can  either find such a set of paths or determine that (i) holds in linear time.
\end{lemma}

\begin{proof}
We can assume $T$  has  fewer than $\frac{s}{2000}$ leaves or (i) holds and we are done. 
By Lemma \ref{specialty}, the number of nodes in the decomposition node cycles of $G$ of length at least six  is at most $\frac{7s}{2000}$, and hence their triangulations account for at most $\frac{7s}{2000}$ of the nodes of
$T$.  Every other decomposition node cycle of the strong 2-cut tree corresponds to at most 3 nodes
of $T$.

 So if we construct the strong 2-cut tree and  delete the decomposition cycles nodes  of size  at least six from it, either this leaves a forest with
at least $\frac{1993s}{2000}$ nodes or we can conclude that (i) holds and we stop. At most $\frac{s}{2000}$ of these nodes have degree exceeding 2  and  at most $\frac{s}{2000}$ are leaves or we obtain a contradiction to our assumption.  Deleting them  leaves a set of  at most $\frac{9s}{2000}$ paths  all of whose vertices have degree 2 in the strong 2-cut tree which contain in total  at least  $\frac{1991s}{2000}$  vertices or (i) holds and we stop. Deleting at most  141 vertices from each of these paths
so that they all have length 0 mod 142,  and then partitioning the remainder
we can obtain a set of more than $\frac{s}{2000}$ paths of length 142, each of which contains a path of length 141 beginning and ending with a cutting node. So, we can return (ii).
\end{proof}

\subsection{Looking Up at The Leaves or Down at The Paths}

The results of the last section are used to show that  we can find  either  a 3-connectivity preserving  matching $N$ or a 2-cut tree with many leaves, or a strong 2-cut
tree with lots of disjoint paths of length 141 all of whose vertices have degree 2 and whose endpoints are 2-cuts.  In the latter two cases we want to find
a 3-connectivity preserving matching or set of triangles.

We describe in detail how to do so, in this section. In the next subsection we put these pieces together to complete our description of Low Degree Compactor.

We look at each of these leaves or paths separately, and find a sweet or well-behaved edge or a triangle in each leaf or path. Theorem \ref{tri3con}, Lemma \ref{sweetie} and Lemma \ref{wellbehaved} tells us this will be sufficient for obtaining a 3-connectivity preserving matching or set of triangles

\begin{lemma}
\label{leafedge}
Let $H$ be a 3-connected graph and $M^*$ a matching in $H$ with at least $\frac{c|V(G)|}{24d^3}$ edges which is induced and such that no vertex of degree at most 12 is adjacent to two edges of $M^*$.

Suppose that $X$ is a minimal cutset consisting either of two edges of $M^*$ or an edge of $M^*$ and a vertex and $U$ is a component of $H-X$. Then there is either
\begin{compactenumerate}
\item
a triangle whose vertices have degree three intersecting  $U$, or
\item
an edge which is well-behaved for $[X,U]$, or
\item
a sweet edge which is either contained in $U$ or has one endpoint in $U$,  and the other a vertex of degree at most $d$ in $X$.
\end{compactenumerate}
Furthermore, we can find this triangle or edge simply by examining $X \cup U$ and the degree of each vertex of $X$.
\end{lemma}

\begin{proof}
We consider a counterexample $[X,U]$ chosen so as to minimize $|U|$.

\paragraph{$X$ is two edges}
Suppose first that $X$ consists of two edges of $M^*$. Then $U$ is not a single vertex $u$ because then it would have degree at most 4, and the
3-connectivity of $H$ would imply that it has edges to both these edges of $M^*$ which contradicts our choice of $M^*$.

\paragraph{Separating vertex subcase}
Now, if there are not two vertex disjoint paths of $H[X \cup U]$ with interiors in $U$ joining the vertices of one edge to another then there is a vertex $z$ of $U$
separating  these two edges and a component $U'$ of $U-z$ adjacent to $z$ and the vertices $x,y$ of one of the two matching edges.

By the minimality of $U$, by examining $U'$ and $X' = \{x,y,z\}$ and the degree of the vertices in $X'$ we can find one of a triangle whose vertices have degree three intersecting  $U'$, an edge which is well-behaved for
$[X',U']$, or a sweet edge which is either contained in $U'$ or has one endpoint  in $U'$, and the other a vertex of degree at most $d$ in $X'$.
In the case of a triangle or a sweet edge, since $U'$ is contained in $U$, and the only vertex of $X'-X$ is contained in $U$, this triangle or sweet edge satisfies 1. or 3. for $[X, U]$, a contraction to $[X, U]$ being a counter-example.

Thus, we may assume there is a well-behaved edge $ab$ of $U'$ for $[X',U']$. We now show $ab$ is well-behaved for $[X,U]$ (which is a contradiction to $[X,U]$ being a counter-example).

To see that (i) in the definition of well-behaved  holds  with respect to $[X,U]$ we note that  because it holds for $[X',U']$ for any clique  $Z$ of size at most 2
in $H-X-U$, the vertices of $X'$  lie in the same component of $H-Z-a-b$. Since $X'$ separates $X$ from $a,b$ it follows that if two vertices of $X$ lie in different
components of $H-Z-a-b$ then they lie in different components of $H-Z$, which would contradict the 3-connectivity of $H$.
To see that (ii) holds in the definition of well-behaved,we note that because (ii) holds for $[X',U']$ we need only show that if $z$ is in $X \cup U-X'-U'$ then
$H-a-b-z$ is connected. Now, since (i) holds for $[X',U']$, X' lies in the same component of $H-a-b-z$. But since $a,b,z$ is a minimal cutset of $H$, both components
of $H-a-b-z$ must intersect both $N(a)$ and $N(z)$. Since $X'$ separates $N(a)$ from $N(z)$, these components also intersect $X'$, which is a contradiction.

\paragraph{Two vertex disjoint paths subcase}
So we can assume that there are two vertex disjoint paths of $U$ joining the two edges of $M^*$ in $X$.
Since $M^*$ is induced both of these have an internal vertex. If both  have only one internal vertex, then by our choice of $M^*$, these vertices
has degree at least 13, and hence each is joined to a vertex off of the two paths. Otherwise, there is an edge on the interior of one of the paths.
In either case, there is an edge $ab$ of $U$ such that there is a path of $U-a-b$ joining our two matching edges.
Thus, (i) in the definition of well-behaved holds for $ab$ (with respect to $[X, U]$).

If (ii) does not hold, then there is a vertex $c$ of $X \cup U$ such that $H-a-b-c$ is not connected.
Since $H$ is 3-connected, every component of $H-X-U$ sees all but at most one of the at least 3 vertices of $X-c$,
so $H-U-c$ is connected. Hence there must be a component $K$ of $H-a-b-c$ completely contained in $U$.

Pick a minimal $K$ (and $a,b,c$) in the following way.
We can determine if there is such a $c$ by examining $X \cup U$, and if not return the well behaved edge $ab$.
Otherwise, by examining $X \cup U$, we can find   an edge $ab$ of $U$  such that there is a path of $U-a-b$  between the two edges of $X$ and a vertex $c$
such that there is a component $K$ of $H-a-b-c$ completely contained in $U$, chosen so that $|V(K)|$ is minimal.

If $K$ is a single vertex $v$ then it has degree 3, and $vc$ is a sweet edge (with sweet end $v$), a contradiction to our choice of $[X, U]$ as this is 3.
So, by the minimality of $K$, there are two paths joining $c$ to $a,b$ that are disjoint except at $c$ and have non-empty interiors in $K$.  Suppose there is an edge $de$ of $K$ not intersecting both these paths. If $de$ is well-behaved, this is a contradiction to our choice of $[X,U]$.
So we know there must be a vertex $f$ of $U$ such that $H-d-e-f$ is not connected. Clearly $d,e,f$ is a minimal cutset of $H$.
If $f$ is not in $K \cup \{a,b,c\}$ then $a,b,c$ lie in the same component of $H-d-e-f$ and separate $f$ from $d,e$ which yields a contradiction.
But, now, every component of $H-K-a-b-c$ sees all of $a,b,c-f$ so there is a component of $H-d-e-f$ completely contained in $K-d-e-f$ which is a contradiction.

\paragraph{$X$ is an edge and a vertex}
So we can assume that $X$ consists of an edge $xy$ of $M^*$ and a vertex $z$.
Now, if $U$ is just a vertex $u$  then $u$  must have degree 3 and $zu$ is a sweet edge. So,  if $z$
has degree at most $d$ we can return it. Otherwise,  by the  order in which we considered the vertices when we chose $M$,
 $x$ and $y$ also have degree 3  and we have found our desired triangle and are done.

Otherwise, we attempt to find two  paths from
$x,y$ to $z$ disjoint except at $z$. If we cannot find two such disjoint paths,we can find a $z'$ and a component $U'$ of
$H-x-y-z'$ completely contained in $U$. Mimicking the argument above, we see that this contradicts the minimality of $U$.
So, we have two such disjoint paths.  If there is no edge $ab$ of $U$ such that there is a path from $z$ to $x$ or $y$  with interior in $U-a-b$
then both these paths have only one internal  vertex and these are the only two vertices of  $U$. Hence they must be adjacent. But now the edges
of these paths not containing $z$ are sweet and we are done. So by examining $U$ we can find  an edge $ab$ of $U$  such that there is a path from $z$ to $x$ or $y$  with interior in $U-a-b$.

We would be done if $ab$ is well behaved, so there is a vertex $c$ of $U+x+y+z$ and a component $K$ of $H-a-b-c$ completely contained in $U$.
Mimicking an argument above we consider an edge $ab$ of $U$  such that there is a path of $U-a-b$  between the two edges of $X$ and a vertex $c$
such that there is a component $K$ of $H-a-b-c$ completely contained in $U$, chosen so that $|V(K)|$ is minimal.

If $K$ is not a single vertex then  by the minimality of $K$, there are two   paths disjoint except at $c$ with non-empty interiors in $K$
joining $c$ to $a,b$.  Suppose there is an edge $de$ of $K$ not intersecting  both these two paths. We are done if $de$ is well-behaved for $[X,U]$.
so we know there must be a vertex $f$ of $U$ such that $H-d-e-f$ is not connected. Clearly  $d,e,f$ is a minimal cutset of $H$.
If $f$ is not in $K \cup \{a,b,c\}$ then  $a,b,c$ lie in the same component of $H-d-e-f$ and separate $f$ from $d,e$ which yields a contradiction.
But, now, every component of $H-K-a-b-c$ sees all of $a,b,c-f$ so there is a component of $H-d-e-f$ completely contained in $K-d-e-f$ which is a contradiction.
So, in this case, both paths have two edges and  $K$ is just the edge joining 
their midpoints.  But now the edges of these paths not containing $c$ are sweet and we are done.

If $K$ is  a single vertex $v$ which  has degree 3,  $vc$ is a sweet edge  Thus if $c$ is in $U$ or has degree at most $d$ we are done,
so it is $z$.   There are two paths of $H-c$ from $a,b$ to $x,y$  because $H$ is 3-connected. They obviously do not pass through $v$.
We are done if $av$ is well behaved for $[X,U]$. So there must be a vertex $e$ such that $a,v,e$ cuts off some part $U_1$ of the interior of $U$ from the rest of $H$. Since $H$ is 3-connected $U_1$ contains a neighbour of $v$ which can only be $b$. Similarly there is an $f$ such that $v,b,f$ cuts off some part $U_2$ of the interior of $U$ including $a$ from the rest of $H$.

 We consider the components of $H-a-b-v-e-f$. One such component $U_3$ contains $c$. so by the 3-connectivity of $H$ it contains all of $H-U-e-f$. There can be no edge from $U_3$ to $a$ because $b,v,f$ separates $a$ from this part of the graph. There can be no edge from $U_3$ to $b$ because $a,v,e$ separates $b$ from this part of the graph. So, since $H$ is 3-connected, it must have an edge from $U_3$ to $e$ and $U_3$ to $f$. Thus, there is no edge from $a$ to $e$ since $f,b,v$ separates $a$ from $U_3$. In the same vein,  if   there were a neighbour of $a$ other than $f,v,b$, then the component $K$ of $H-a-b-v-e-f$ it is in could not be $U_3$ and could not have edges to $e$. Since it is not $U_3$, $K$ does not have edges to $v$ either. So by 3-connectivity it must have edges to $b$ and $f$, contradicting the fact that $a,v,e$ separates $b$ from $U_3$. So no such neighbour exists. Symmetrically $b$ sees only $f,v$ and $a$. So we get a degree 3 triangles $a,b,v$, and are done.

\end{proof}

\begin{lemma}
\label{preprepathedge}
Suppose the following holds:
\begin{enumerate}
\item
$U$ is  the union of some  components of $H-X$
\item
 $X$  consists of the union of $X_1=\{x_1,y_1,z_1\}$ and $X_2=\{x_2,y_2,z_2\}$ for two (not necessarily distinct)  edges $x_1y_1$ and $x_2y_2$ of $M^*$  and two (not necessarily distinct)
 vertices $z_1$ and $z_2$ (but  either edges are distinct or vertices are).
\item   every component of $H-X-U$ either has an edge to all of  $X_1$ and none 
of $X_2-X_1$   or all of $X_2$ and none of $X_1-X_2$   
and  components with both these neighbourhoods exist,
\item  there is an edge $ab$ of $U $ such that every component of $H[X \cup U]-a-b$ intersects both $X_1$ and $X_2$,  and  there are two vertex disjoint paths   of $H-a-b$ from $X_1$ to $X_2$ with their interiors in $U$,
\end{enumerate}
 Then  there is an edge which is well-behaved for $[X,U]$ or there is a
sweet edge which is either contained in $U$ or has one endpoint  in $U$,  and  the other a vertex  in $X$ which is not $z_1$ if $z_1=z_2$.
Furthermore, we can find such an edge by examining just $X \cup U$ and the degree of the vertices of $X$.
\end{lemma}

\begin{proof}

If $Z$ is a clique of $H-X-U$ of size at most 2  then it is disjoint from some component $C$ of $H-X-U$. \WLOG\ this component has edges to $X_1$.
So, there is a connected component of $H-Z-a-b$ containing $X_1$ and hence, since  every component  $H[X \cup U-a-b]$ intersects both $X_1$ and $X_2$, all of  $X$.
Since $H$ is 3-connected it follows that $H-Z-a-b$ is connected. Now, $ab$ is not well-behaved or we would be done.
Thus, there is a vertex $c$ in $U \cup X$  such that  $H-a-b-c$
is disconnected. Since by hypothesis  there is a path of $H-a-b-c$  from $X_1$ to $X_2$, and every component of $H-U-X$ sees all of $X_1$ or $X_2$,
all of $H-U-c$ lies in one component of $H-a-b-c$. So there is a component of $H-a-b-c$ completely contained in $U$.

So either we can return the well-behaved edge $ab$, or by examining $H[X \cup U]$, we can choose 
 an edge  $ab$  and vertex $c$  of  $U \cup X$  such that  $X$ lies in one component of $H[X \cup U-a-b]$ there are two paths  from $X_1$ to $X_2$ with their interiors in $U-a-b$
 and there is a component  $K$ of $H-a-b-c$ completely contained in $U$,  with  $|V(K)|$ minimal over all possible choices of $ab,c$.

If $K$ is not a single vertex then by the minimality of $K$, there are two   paths disjoint except at $c$ with non-empty interiors in $K$
joining $c$ to $a,b$.  Suppose there is an edge $de$ of $K$ not intersecting  both these two paths. We are done if $de$ is well-behaved for $[X,U]$,
so we know there must be a vertex $f$ of $U \cup X$ such that $H-d-e-f$ is not connected. Clearly  $d,e,f$ is a minimal cutset of $H$.
If $f$ is not in $K \cup \{a,b,c\}$ then  $a,b,c$ lie in the same component of $H-d-e-f$ and separate $f$ from $d,e$ which yields a contradiction.
But, now, every component of $H-K-a-b-c$ sees all of $a,b,c-f$ so there is a component of $H-d-e-f$ completely contained in $K-d-e-f$ which  contradicts our choice of $a,b,c$.So, in this case both paths have two edges and $K$ is the edge 
joining their midpoints. In this case, the edges of this path which do not use $c$ are 
sweet and we are done. 

So $K$ is a single vertex $v$ which  has degree 3, and $vc$ is a sweet edge. Thus, we are already done unless $c=z_1=z_2$.
 There are two paths of $H-c$ from $a,b$ to $X-c$   because $H$ is 3-connected. They obviously do not pass through $v$.
We are done if $av$ is well behaved for $[X,U]$. So there must be a vertex $e$ such that $a,v,e$ cuts off some part $U_1$ of $U$ from the rest of $H$. Since $H$ is 3-connected $U_1$ contains a neighbour of $v$ which can only be $b$. Similarly there is an $f$ such that $v,b,f$ cuts off some part $U_2$ of $U$ including $a$ from the rest of $H$.

 We consider the components of $H-a-b-v-e-f$. One such component $U_3$ contains $c=z_1=z_2$, so    our condition on the components of $H-X-U$ imply
$U_3$  contains all of $H-U-e-f$. There can be no edge from $U_3$ to $a$ because $b,v,f$ separates $a$ from this part of the graph. There can be no edge from $U_3$ to $b$ because $a,v,e$ separates $b$ from this part of the graph. So, since $H$ is 3-connected, there  must be an edge from $U_3$ to $e$ and $U_3$ to $f$. Thus, there is no edge from $a$ to $e$ since $f,b,v$ separates $a$ from $U_3$. In the same vein,  if   there were a neighbour of $a$ other than $f,v,b$, then the component $K$ of $H-a-b-v-e-f$ it is in could not be $U_3$ and could not have edges to $e$. Since it is not $U_3$, $K$ does not have edges to $v$ either. So by 3-connectivity it must have edges to $b$ and $f$, contradicting the fact that $a,v,e$ separates $b$ from $U_3$. So no such neighbour exists. Symmetrically $b$ sees only $f,v$ and $a$. So we get a degree 3 triangles $a,b,v$, and are done. .

\end{proof}

\begin{corollary}
\label{prepathedge1}
Suppose the following holds:
\begin{enumerate}
\item
$U$ is  the union of some  components of $H-X$
\item
 $X$  consists of the union of $X_1=\{x_1,y_1,z_1\}$  and $X_2=\{x_2,y_2,z_2\}$ for two (not necessarily distinct) edges $x_1y_1$ and $x_2y_2$ of $M^*$  and  two (not necessarily distinct) vertices  $z_1$  and $z_2$  not in $M^*$.
\item   every component of $H-X-U$ either has an edge to all of $X_1$  and none of $X_2-X_1$ 
   or to all of $X_2$ and none of $X_1-X_2$
and  components with both these neighbourhoods exist,
\item
there are three vertex disjoint paths from $X_1$ to $X_2$ such that one of these paths $P$ links $x_1$ to $x_2$ and has an internal edge $ab$.
\end{enumerate}
Then  there is an edge which is well-behaved for $[X,U]$ or there is a
sweet edge which is either contained in $U$ or has one endpoint  in $U$,  and  the other a vertex  in $X$ which is not $z_1$ if $z_1=z_2$.
Furthermore, we can find such an edge by examining just $X \cup U$ and the degree of the vertices of $X$.
\end{corollary}

\begin{proof}
We let $ab$ be an internal edge of $P$. Since $H$ is 3-connected every component of $H[X \cup U]-a-b$ intersects $X$ and hence,
because of the existence of $P_2$ and $P_3$, both $X_1$ and $X_2$.
We  apply Lemma \ref{preprepathedge}.
\end{proof}

\begin{corollary}
\label{prepathedge2}
Suppose the following holds:
\begin{enumerate}
\item
$U$ is  the union of some  components of $H-X$
\item
 $X$  consists of the union of   an edge $xy$  of $M^*$  and  two distinct  vertices  $z_1$  and $z_6$,
 \item   every component of $H-X-U$ either has an edge to all of $X_1=\{x,y,z_1\}$ but not $z_6$, 
   or to all of $X_2=\{x,y,z_6\}$ but not $z_1$.  
and  components with both these neighbourhoods exist,
\item  there is a path $P$ of $H-x-y$ containing 6 vertices $z_1,..,z_6$ such that for each $i,~ \{z_i,x,y\}$
separates $\{z_j,j<i\}$ from $\{z_j|j>i\}$.
\end{enumerate}
Then  there is an edge which is well-behaved for $[X,U]$ or there is a
sweet edge which is either contained in $U$ or has one endpoint  in $U$,  and  the other a vertex    in $X$.
Furthermore, we can find such an edge by examining just $X \cup U$ and the degree of the vertices of $X$.
\end{corollary}

\begin{proof}
Since $\{z_1,z_3,x,y \}$ separates $z_2$ from $z_4,z_6$ and hence $H-X-U$, and $H$ is 3-connected,
there is a path of $H-z_1-z_3-z_4$ from $z_2$ to one of $x$ or $y$ whose interior is in $U$.
Symmetrically, there is a    path from $z_5$ to one of $x$ or $y$ whose interior is in $U-z_3-z_4$.
We apply Lemma \ref{preprepathedge} to $z_3z_4$.
\end{proof}

\begin{corollary}
\label{pathedge}
Suppose the following holds:
\begin{enumerate}
\item
$W$ is  the union of some  components of $H-Y$
\item
 $Y$  consists of the union of $Y_1=\{x_1,y_1,z_1\}$ and $Y_{71}=\{x_{71},y_{71},z_{71}\}$ for   two (not necessarily distinct)  edge $x_1y_1$  and $x_{71}y_{71}$ of  $M^*$  and  two (not necessarily  distinct  vertices  $z_1$  and $z_{71}$,
\item   every component of $H-Y-W$ either has an edge to all of $Y_1$ and none of $Y_{71}-Y_1$   or to all of $Y_{71}$ and none of $Y_1-Y_{71}$ and  components with both these neighbourhoods exist,
\item  there are  distinct $Y_1,Y_2,..Y_{71}$ such that $Y_i=x_i,y_i,z_i$ for some edge $x_iy_i$ of $M^*$ and some vertex $z_i$ of $H$
and separates $\cup_{j<i} Y_j$ from $\cup_{j>i} Y_j$,
\end{enumerate}
Then  there is an edge which is well-behaved for $[Y,W]$ or there is a
sweet edge which is either contained in $W$ or has one endpoint  in $W$,  and  the other a vertex  of degree $d$  in $Y$.
Furthermore, we can find such an edge by examining just $Y \cup W$ and the degree of the vertices of $Y$.
\end{corollary}

\begin{proof}

Suppose first that $z_6$=$z_1$.
Let $X_1=Y_1$ and $X_2=Y_6$. Let $U$ be the components of $H$ which have edges to both $X_1$ and $X_2$.
Because the $Y_i$ are cutsets, we know that $U \subset W$.

There are 3 vertex disjoint paths $P_1=z_1,P_2,P_3$ of $H[X \cup U']$ from $X_1$ to $X_2$.
Since each $P_j$ intersects each $Y_i$ we know that there total length is at least $8$.
So one of $P_2$ or $P_3$ has an internal edge $ab$. We apply Corollary \ref{prepathedge1}.
Since $z_1=z_2$ and every vertex of $X-z_1$ is either a vertex of  $W$ or a vertex of $Y$ which is in a matching edge and hence
has degree $d$ we are done.

So we can assume that $z_6 \neq z_1$ and $z_{66} \neq z_{71}$. Thus, for all $5<i<67$, $z_i$ is in $W$.
There are 3 vertex disjoint paths $P_1,P_2,P_3$ from $Y_6$ to $Y_{66}$. Suppose first that
one of these paths, $P$, contains 4 distinct vertices  which are an $x_i$ or $y_i$.
Let $X_1$ be $Y_i$ for the lowest indexed $i>5$ such that $P \cap Y_i \neq z_i$.
Let $X_2$ be $Y_i$ for the highest indexed $i<67$ such that $P \cap Y_i \neq z_i$.
Apply Corollary \ref{prepathedge1}. We are done.

It follows that there are at most 11 distinct vertices in $\cup_{5<i<67} \{x_i,y_i\}$ and hence at most
$9$ values of $i$  between $6$ and $66$ for which $\{x_i,y_i\} \neq \{x_{i-1},y_{i-1}\}$. Thus,there must be an $i$ such that
$\{y_i,x_i\}=\{y_{i+5},x_{i+5}\}$. We set $X_1=Y_i$ and $X_2=Y_{i+5}$. Applying Corollary \ref{prepathedge2}, we are done.

\end{proof}

\subsection{Finding a large 3-connectivity preserving matching or a large set of degree 3 triangles}

In this section, we describe how to massage the  matching $M^*$ of size $\frac{cn}{24d^3}$ within the vertices of degree at most $d$  constructed by Low Density Compactor  into either output (iii) for Compactor, i.e.    a 3-connectivity preserving matching of size at most $\epsilon n$  among vertices of degree at most $\Delta$, or output (iv) for Compactor, i.e.   a  set of  $\epsilon n$ disjoint triangles of degree 3 vertices. This completes our description of LDC, Compactor and Iterative Compactor.

To begin we contract the edges of $M^*$  to obtain a graph $H^*$. We let $S$ be the set of vertices to which the edges of $M^*$ were contracted. We construct a special 2-cut tree for $H^*$.
If it contains at least $\frac{|M^*|}{15}$ leaves then  since we are free to make $\delta$ and hence $\epsilon$ as small as we want in terms of $c$ and $d$,
there are  at least $4d\epsilon n$ leaf  nodes and at least $3d \epsilon n$ of these have   interiors with fewer than
 $\epsilon^{-1}$ vertices.  We note this implies that the  vertices of their interior have degree at most $\Delta$. We apply Lemma \ref{leafedge} to find  for each leaf either a degree 3 triangle intersecting its interior, a well-behaved edge of its interior, or a sweet edge with one endpoint in its interior.  This takes a constant time per leaf, as we are looking at a constant sized subgraph.

If there are at least $\epsilon n$  degree three triangles then  since they are distinct, the   3-connectivity of $H$ implies they are disjoint, so we can return this set.
If there are at least $\epsilon n$  interiors containing well-behaved edges  then Lemma \ref{wellbehaved} tells us that the matching consisting of these edges is 3-connectivity preserving, so we return it.

If there are at least $d\epsilon n$  interiors  for which there is  a sweet edge which is either contained in the interior or joins an edge of the interior
to a vertex of degree at most $d$ then we choose a matching from this set by choosing for each vertex not in an interior which is in some of these sweet edges, one of the
at most $d$ such edges it is in. Lemma \ref{sweetie} tells us that this matching  is 3-connectivity preserving, so we return it.

 Otherwise by Lemma \ref{specialty2} we can find a subset $S'$ of $S$ of size at least $\frac{|S|}{15}$ no two members of which lie in a strong 2-cut or in a decomposition node cycle.
Thus, no two of these vertices form a 2-cut of $H^*$. Letting $M^+$ be the edges of $M^*$
which were contracted to the vertices of $S'$, it follows that there are no two edges of $M^+$ whose
four endpoints form a cutset of $H$.

We contract the edges of $M^+$ in $H$ to obtain a new graph $H^+$.
We continue to use $S'$ to represent the vertices to which the edges of $M^+$
are contracted.
We build the strong 2-cut tree and a special 2-cut tree for $H^+$.

If the special  2-cut tree has at least $\frac{|S'|}{4000}$ leaves then there are  at least $3d\epsilon n$ leaf  nodes whose interiors have  at most $\epsilon^{-1} $ vertices.So again,  applying  Lemma  \ref{leafedge} we are done.

If the strong 2-cut tree has at least $\frac{|S'|}{4000}$ disjoint paths of length 141, each beginning
and ending at a cutnode and containing only degree 2 vertices,  then  there are at least ${4d\epsilon n}$ of these paths whose interior contains at most ${\epsilon^{ -1}}$ vertices. We  apply Corollary   \ref{pathedge} to each path to obtain one of:     a degree 3 triangle  intersecting its interior, a well-behaved edge of its interior, or a sweet edge which is either contained in its interior or joins a vertex of its interior to a vertex of degree at most $d$. Proceeding as before, we can find and return one of our desired outputs, in this case.

Otherwise, by Lemma \ref{2cutstructcombined}, the special 2-cut tree, and hence also the strong 2-cut tree
has at most $\frac{|S'|}{2}$ nodes. Now since no two vertices of $S'$ lie in a 2-cut together, there is a most 1 vertex of $S'$ in each node of the strong 2-cut decomposition which is a cycle. So, the number of vertices of
$S'$ in 2-cuts is at most the number of nodes of the strong 2-cut decomposition. Hence there is a set
$S_{last}$ of at least $\frac{|S'|}{2}$ vertices of $S'$ which lie in no 2-cut of $H^+$.

We let $N$ be the edges which were contracted  to obtain $S_{last}$. If there were a 2-cut in $H/N$, since it does not correspond to a 2-cut of $H^+$, it must consist of a 3 cut of $H$ which contains an edge of $N$ and a vertex of some $e \in M^+ - N$ and one of its components consists of the other endpoint $x$  of $E$.
But since $M$ is induced and $x$ has degree at least 3 (because $H$ is 3-connected), this is impossible.
So, we can return $N$.

\section{Reductions and Certificates}
 \label{reductionssec}

 Since the reductions of the \tricname{} $\tric$ we consider are 3-connected, any planar reduction has a unique embedding.
 By a \emph{strong planar reduction} we mean a reduction $F$ of a \tricname{} where every \separator\ is a face of this embedding. We note that the recursive procedure discussed above actually produces
 a strong planar embedding if the two paths do not exist since if some \separator\ is a triangle which is not a face then it is a
 3-cut which separates $C$ from some vertex $v$ of the reduction contradicting the fact that we have obtained a reduction which has no reductions.

 So, we obtain:

\begin{lemma}
 Either the desired 2 paths of an instance of 2-DRP exists or the \tricname{} of that instance has a strong planar reduction.
 \end{lemma}

\subsection{Overview}

Both the classical approach, and our new algorithm attempt to construct strong planar reductions  (or the   the desired two disjoint paths).

 The classical approach attempts to construct reductions which cut off as much as possible,
 thereby obtaining a strong planar reduction if one exists. Our algorithm proceeds differently. It attempts to construct strong planar reductions which cut off as little as possible. The reason that we do this is that it turns out that
 the minimal pieces which we cut off in $G_i$ are uncontracted to pieces which will also be cut off in any strong planar reduction
 of $G_{i-1}$. This is the crucial fact which allows us  to  uncompact up the sequence to obtain a solution to  the  triconnected component $T$ of the auxiliary graph of the original 2-DRP instance.

 The following lemmas point out, some pieces which we know must be cut off in any reduction of a compaction of $\tric$.

\begin{lemma}
\label{K33lem}
 Suppose that we have a $K_{3,3}$ subdivision in a compaction  $F$  of a \tricname{} $\tric$. Then for any  strong planar reduction $J$ of $F$ there must be
 a set $X$ inducing a triangular face of $J$  such that five of the six centers of the subdivision are contained in  the union of $X$ and those   components of $F-V(J)$
 which attach at $X$.
\end{lemma}

\begin{proof}
For each set $X$ of vertices bounding a triangular face of the unique embedding of  $J$ let $U_X$ be the union of the components of $F-V(J)$  which
 attach at X. We claim that there is some $X$ such that $U_X$ contains at least 2 centers of the subdivision. Otherwise, for each  $X$ inducing a
 triangle  bounding a face, the part of the subdivision contained within $U_X$ consists either of  (i) the interior of a path between two vertices of $X$, (ii)  the interior of two paths from one vertex of $X$ to the other two,  or  (iii) three paths from the center of the subdivision contained in $U_X$ to the vertices of $X$. We can draw this part of the
 subdivision within the face bounded by $X$. But this, along with our embedding of $J$ yields a planar drawing of the subdivision, which is
 impossible.  This proves our claim.

 Let $X$ be such that $U_X$ contains two centers of the subdivision. Suppose  that $U_X$ does not  contain a center on each side of the subdivision.
 Then, for the two centers that $U_X$ does contain, $X$ must intersect all six paths from these centers to the centers on the other side of the subdivision. Now, the only way this can happen is if $X$ is the three centers on the other side of the subdivision, in which case the lemma is true.
  So, we can assume that this is not the  case and hence  there are two centers in $U_X$ which are  on opposite sides
 of the bipartition.

 Now, there are four disjoint paths of the subdivision joining neighbours of these  centers to the four other centers so it follows that $U_X$ must actually
 contain two centers from one side of the subdivision and one from the other.  There are five paths of the subdivision  corresponding to edges of $K_{3,3}$ from  neighbours of these centers to the other centers all of which are contained in  $U_X$ or intersect $X$. Since a vertex is only in two of these paths if it is one of the other
 centers of the subdivision, it follows that either there are  5 centers of the subdivision in $X \cup U_X$ or there are four centers of the
 subdivision in $U_X$. We can assume the latter. Now, if the other two centers of the subdivision are in $V-X-U_X$ then the  (disjoint) interiors of each of the
 at least 4 disjoint paths of the subdivision  between  these other two centers and the four centers in $U_X$ must intersect $X$ which is impossible.
\end{proof}

\begin{corollary}
\label{K33cor}
  Suppose that we have a $K_{3,3}$ subdivision in a reduction $F$ of a \tricname{} $\tric$  with one center $v^*$.
 Let $A$ and $B$ be two disjoint connected subgraphs  each of which  contains one of the centers of the subdivision
 on the same side of the $K_{3,3}$ as $v^*$ along with the interior of the three paths of the subdivision which have this center as an endpoint, but
 contains none of the rest of the subdivision. Then, in any strong planar reduction $J$ of $F$ there must be a triangular face with vertex set $X$ such that both $A$ and $B$ are contained in a component whose separator is $X$.
 \end{corollary}

\begin{proof}
By Lemma \ref{K33lem}, there is a triangular face $X$ of $J$  such that letting $U$ be the union of the components of $F-V(J)$ which "attach"
 at $X$, we have that  $U \cup X$ contains 5 centers of the subdivision.  We know $v^*$ is not in $U$ because $U$ is disjoint from $C$. Furthermore, $v^*$ is not in  $X$ because it sees nothing not in  $C$ and hence sees none of $U$.  So the other 5 centers of the subdivision must be in $U \cup X$. Thus each of the paths of the
 subdivision corresponding to edges of the $K_{3,3}$ which have $v^*$ as an endpoint must intersect $X$. It follows that $A$ and $B$ are disjoint from $X$. Since they are connected
 and intersect $U$ it follows that they are contained in components of $U$.
\end{proof}

\begin{definition}
 A  {\it ferociously strong planar reduction} $J$ of $H$, is  a strong planar reduction that satisfies  the following property:

 Every \separator\ $X$ with respect to $(J,H)$ is either the separator of
 two different components or the vertices of the unique component with \separator\ $X$ can be partitioned into two so that each part induces a connected subgraph with edges to each vertex of $X$.
\end{definition}

An immediate corollary of Corollary \ref{K33cor} is the following.

\begin{corollary}\label{maximalplanar}
  If $J$ is a ferociously strong planar reduction of $H$ then any strong planar reduction of $H$ has its vertex set contained in  $V(J)$.
\end{corollary}

I.e. if $H$ has a ferociously strong planar reduction then this is the strong planar reduction which cuts off as little as possible of $H$.

 To complete this section we show below that  if a $K_5$ subdivision with set of centers $C$  does not exist in a compaction $H$ then $H$ has a ferociously strong planar reduction.

For each compaction it considers, our algorithm finds and returns either
the desired $K_5$ subdivision or a ferociously strong planar subdivision, along with for each \separator\ $X$ such that only one component $U_X$ attaches at $X$, two subgraphs of $U_X$ partitioning its vertex set both of which are connected and are joined by an edge to each vertex of $X$.


\begin{lemma}\label{posornegcert}
Every reduction $F$ containing all root vertices $C$ of a \tricname{} $\tric$  either
has a $K_5$ subdivision with set of centers $C$ or
has a ferociously strong planar subdivision.
\end{lemma}

Our proof requires the following technical claim whose proof is deferred to the end of this section.

\begin{claim}\label{k33subdivisionclaim}
  Suppose that $F$ is a reduction of a \tricname{} $\tric$ with set of root vertices $C$ and $X$ is a cutset of size 3 separating 5 of the 6 centers of some  $K_{3,3}$ subdivision of $F$ from $C$, chosen so as to maximize the number of vertices in the component of $F-X$ which intersects $C$. Then either there are two components of $F-X$ disjoint from $C$ which attach at $X$, or there is a partition of the vertices of the  unique  such component  of $F-X$ attaching at $F$ into  two  connected subgraphs, each of which is joined to all three vertices of $X$ by an edge.
\end{claim}

\begin{proof}(of Lemma \ref{posornegcert})
Consider a minimal counterexample $F$ to the statement.
Since $F$ is 3-connected and not $K_5$ either it is  planar or it has a $K_{3,3}$ subdivision \cite{asanok33}. If $F$ is planar then $F$ is a ferociously strong planar reduction of itself.
Otherwise, $F$ has a $K_{3,3}$ subdivision and we find a 3-cut $X$ separating 5 of the 6 centers of this subdivision from $C$ and subject to this maximizing the number of vertices contained in the component of $F-X$ intersecting $C$.

We let $F^*$ be the reduction of $F$ obtained by deleting the components of $F-X$ disjoint from $C$ and adding edges so $X$ is a clique. Now, $F$ has no $K_{5}$ subdivision with set of centers
$C$ which implies that $F^*$  has no such subdivision and is not $K_5$.   Thus, the minimality of $F$ implies that $F^*$ has
a ferociously strong planar reduction $J$.

If there is a vertex $x$  of $X$ which is not in $J$, then it is contained in some component $U$ of $F^*-V(J)$.
We add the vertices of $F-F^*$ to $U$ to obtain a  strong reduction of $F$. If there are two components
of $F^*-S(U)$ which are disjoint from $C$ then this is a ferociously strong reduction. Otherwise, since $J$
was ferociously strong,  $U$ can be partitioned into $A$ and $B$ as in the definition of ferociously strong.
One of these two connected graphs contains $x$, adding the vertices of $F-F^*$ to it, we see that our new
reduction of $F$ is ferociously strong.

If there is no such vertex then $X$ forms a triangle of $J$. This must be a face because $F^*-X$ is connected
and hence so is $J-X$. We note that in this case $J$ is also a strong reduction of $F$. The lemma follows by Claim \ref{k33subdivisionclaim}.
\end{proof}

\begin{proof}(of Claim \ref{k33subdivisionclaim})
We prove this claim in two steps. For a choice of $X$ and $K_{3,3}$
subdivision as in the statement of the claim, we can assume $S(X)$ is
connected or the claim follows by using its components.

Now, let $C$ and $D$ be the two sides of the subdivision.

\begin{lemma}\label{innerlemma1}
For any such $X$ and $K_{3,3}$ subdivision, $S(X)
\cup X$ contains a set $Y$ and two disjoint connected subgraphs $A$ and $B$
 such that
\begin{enumerate}
\item   there are edges from each of $A$ and $B$ to each vertex of $Y$,
\item   there are three vertex disjoint paths of $G-A-B$ from $X$ to $Y$
\item   five of the centers of the subdivision live in $A \cup B \cup Y$.
\end{enumerate}
\end{lemma}

\begin{lemma}\label{innerlemma2}
  In a structure as in Lemma \ref{innerlemma1} which maximizes the size of $A \cup B$, $Y$ is $X$.
\end{lemma}

The claim follows immediately from these two lemmas because, since $S(X)$ is connected, if we choose $A,B$ as in Lemma \ref{innerlemma2} then not only is $Y=X$ but $V(A) \cup V(B)$ must span $S(X)$.

\begin{proof}(of Lemma \ref{innerlemma1})
  We can assume all 6  of the centers of the subdivision
lie in $X \cup S(X)$ or the subdivision itself yields   such a structure.

We consider a set $Q$ of 3 paths from $X$ to the six centers of the
$K_{3,3}$ minimizing the edges not in the subdivision which are  used.
This implies that for any path $P$  of the subdivision which corresponds
to an edge between centers $x$ and $y$ and intersects a path of $Q$,
either both $x$ and $y$ are endpoints of a member of $Q$, or only one is
and the intersection of $P$ with this element of $Q$ is exactly a common
subpath.  Two of the three centers
of the subdivision which are endpoints of elements of $Q$ must be on the
same side of the $K_{3,3}$ subdivision (with partite set $(C', D')$), \Wlog{} $C'$. If the third center endpoint is in $D'$, we can
reroute/extend the path of $Q$ containing it using the path of the
subdivision between this third center endpoint and the third center of $C'$.
By the above remark it stays disjoint from the other elements of $Q$.

Now consider  the six edges of the subdivision from the 2 centers $d_1$ and
$d_2$ in $D'$ which were not endpoints of a path in $Q$ originally. These six
 paths form a subgraph $H$ consisting of  3 vertex disjoint paths from
$d_1$ to $d_2$ each containing a center in $C'$. By our remarks above, each
path of $Q$ when traversed from $X$, first intersects $H$  in a vertex on
the path containing the center in $C'$ which is its endpoint. We let these
three intersection points be $Y$, and there, deletion splits $H$ into $A$ and
$B$.
\end{proof}

\begin{proof}(of Lemma \ref{innerlemma2})
Since $X$ was the closest cutter separating off 5
centers of this subdivision, $Y$ does not separate off $X$ from $A \cup
B$. So, \Wlog, there is a path $P$ from $A$ to $X$ in $X \cup S(X)-Y$ which is internally disjoint from $A \cup B$. We
let $z$ be the first vertex on this path when traversed from $A$ which is
on one of the three vertex disjoint paths from $X$ to $Y$, call this path
$R$. We add $z$ to $Y$ and delete the endpoint of $R$ in $Y$ from $Y$. We
add $P-z$ to $A$ and the component of $R-z$ intersecting $Y$ to $B$.
\end{proof}

\end{proof}

\section{The Uncompaction Procedure}
\label{uncompactionoutline}

Our  uncompaction algorithm  returns for each $i$ from $l$ down to $1$, in turn,  either (i)  vertex disjoint paths $P_1$ and $P_2$  of
$G_i -v^*$ such that $P_i$ joins $s_i$ to $t_i$ or (ii)   a ferociously strong planar reduction $H_i$   of $G_i$, together with, for each component $K$ of
$G_i-V(H_i)$  such that there is no other component of $G_i-V(H_i)$ with the same (three) attachment vertices in $V(H_i)$, a partition of $V(K)$ into two connected
subgraphs each of which has an edge to each attachment vertex of $K$ in $H_i$.

We store all the information needed when returning (ii)  in one coloured minor of $G_i$. The minor is obtained by contracting
 every component $K$ of $G_i-V(H_i)$ as follows.  If  there is another component with the same attachment vertices we contract $K$ to a vertex. Otherwise we contract
each of  the two subgraphs of $K$
in the partition of $K$ of the last paragraph to a vertex and delete the edge between them if it exists.  We colour the vertices of $V(H_i)$, green. For every set $X$
of three vertices of $V(H_I)$ which are the attachment vertices for some component of $G_i-V(H_i)$  there are at least two vertices of $F_i-V(H_i)$ which have neighbourhood $X$.
We colour one red and the rest  yellow. By a {\it twin set} in $F_{I+1}$ we mean a red vertex and all of its yellow twins.

We remark that $H_i$ can be recovered from our coloured $F_i$  by contracting each red  vertex to one of its neighbours and each yellow vertex to a neighbour to which
a red  twin of it has not been contracted. We also note that $F_i$ is obtained from the planar subgraph  $D_i$ consisting of the red and green vertices by adding
some  positive number of  yellow  twins for each vertex in the stable set of red vertices. Furthermore, each of these red vertices has degree 3. For any subgraph $D'_i$  of $D_i$,
there is a corresponding subgraph $F'_i$ of $F_i$ obtained in a similar fashion.

Uncompacting requires us to find a solution for  $G_i$, given a solution for $G_{i+1}$.
If our solution for $G_{i+1}$ is a pair of paths it is a simple matter to uncontract them to obtain a solution for $G_i$ which is a pair of paths.
Otherwise, we want to use $F_{i+1}$ in looking for our solution in $G_i$.

 Key  to doing so is the observation that Corollary \ref{K33cor}  implies that we  need to find a ferociously strong reduction of the minor $G'_i$ of $G_i$ where we perform the contractions into the yellow and red vertices of $F_{i+1}$.  To see this we note that for any red vertex   and yellow twin of it, there is a $K_{3,3}$ subdivision of $G_i$ with one center being $v^*$, another in the  subgraph $A$ contracted to the red vertex, a second in the subgraph $B$ contracted to the yellow vertex, and the remaining three in the subgraphs contracted to the neighbours of the red vertex. Applying Corollary \ref{K33cor} to  this subdivision, $A$ and $B$ we determine that any ferociously strong planar reduction  of
$G_i$ is also a ferociously strong planar reduction of $G'_i$.  We note that the same argument shows it is enough to find a reduction of the graph obtained from $G'_i$ by deleting any edges joining red vertices to yellow.

 We attempt to find a ferociously strong planar reduction  $H_i$
for $G'_i$ and build the corresponding coloured minor  $F_i$. If we  fail we will find a pair of paths we can return. The key to finding $H_i$  is to
find the 3-cuts it uses.

There are a number of facts that help us in finding these 3-cuts. Perhaps the most important is that we are looking for them in a relatively simple graph.
This may not appear to be the case when we add back a set $F$ of edges or a set $S$ of vertices, as the resultant graph could contain any graph as a subgraph. However, because
$c>3$, our property on the edges and vertices we deleted ensures that the 3-cuts of $G_i$ correspond to the 3-cuts in $G_{i+1}$.
This implies that the 3-cuts of $G'_i$ essentially\footnote{the statement  is not precisely true because of the red and yellow vertices but they pose no real problem as they are cut off from the rest of the graph by cutsets of size $3<c$} correspond to the cuts in $F_{i+1}$  so we can look for the cuts in the latter which is obtained from a  planar graph
by  duplicating  a stable set of degree 3 vertices.
When we uncontract a matching or a set of triangles, we will exploit the fact that we are looking at an {\it  expansion} of  such an extension of a    planar graph, i.e.
a graph from which we  can obtain  such an extension  by contracting connected subgraphs of size at most 3.

The reason this helps is that it allows us to look locally. Any 3-cut  in a 3-connected  planar graph must consist of three vertices every two of which lie together on a face.
In the same vein, any 3-cut in the 3-connected expansion of a planar graph must consist of three vertices coming from a set of at most three vertices of the planar graph every pair of
which lie together on a face.

So, we can look for our 3-cuts locally, specifically in subgraphs  of the planar graph (we will not need to consider the yellow vertices) whose face-vertex incidence graphs have bounded diameter( or their expansions).
This is fortuitous because such planar subgraphs   have bounded tree width (as do their expansions as expanding can at most triple the
tree width) which makes them easy to handle.

Now, in building a ferociously strong planar reduction, a priori we are required  to cut off 5 out of 6 centers of every
$K_{3,3}$ subdivision. However, we shall see that  we can focus on a set of $O(|V(G)|)$ subdivisions which
can be handled locally. This also makes our task much simpler.

We close this overview of the uncompaction algorithm with an example which illustrates why we can restrict our attention to a specific set of  ``local'' $K_{3,3}$
subdivisions.
First however, we explain  how to uncontract triangles, which are easy to handle,
and   make a simple  observation  about the connectivity of $D_{i+1}$.
We present the rest of the details in Section \ref{uncompactiondetails}.

We note that $H_{i+1}$ is 3-connected because it is obtained from $G_{i+1}$ by pruning some components of $G-X$ for some cutsets $Y$ (of size 3)
and adding a clique on  any such $Y$. Thus, any 2-cut  $X$  of $D_{i+1}$  must be laminar with one of these cliques and hence contain a red vertex $v$.
Now, since there are 3 paths of $D_{i+1}$ from $v^*$ to $v$, the component $U$ of $D_{i+1}$ containing $v^*$ contains two of the neighbours of $v$.
So replacing $v$ by its third neighbour $w$ yields another 2-cut unless $Y$ has only two components one of which is $w$. This implies
the following:

\begin{observation}
\label{Di3conobs}
$D_{i+1}$ is a subdivision of a 3-connected planar graph  and  so has a unique embedding.
Furthermore,  the only
2-cuts of $D_{i+1}$ are either (a) two red vertices separating a path with one or two green internal vertices joining them from the rest of the graph, or (b)
a red and a green vertex separating a common  green neighbour from the rest of the graph.
\end{observation}

We will need to pay special attention to the degree 2 green vertices cut off by 2-cuts but they will not cause any real
problems.

\subsection{Uncontracting  A Set of  Triangles}

If $G_{i+1}$ was obtained from $G_i$ by contracting a set of triangles  whose vertices have degree 3,
 our job is straightforward. Any triangle which is contracted into the subgraph corresponding  to a yellow or red vertex simply stays in this vertex we do not uncontract it.
If a triangle was contracted to a green vertex adjacent to a yellow and red vertex, then there is an edge from the triangle to
each of a red vertex, a yellow vertex and a green vertex. We contract the vertex of the triangle
which sees a red vertex into the red vertex, the vertex which sees a yellow vertex into the yellow vertex and thereby obtain the same coloured minor as before.
We now simply uncontract all the remaining triangles into a green triangle,  and compute the unique planar embedding of the resultant graph consisting of the red
and green vertices\footnote{In fact one can simply rearrange the embedding locally, though easier  to do it takes longer to explain, so we avoid doing so.}.

\subsection{ Focusing on   A Few Local $K_{3,3}$ Subdivisions}
\label{uncompactionfocussing}
We consider an example which illustrates our approach of restricting our attention to $O(|V(G)|)$ local $K_{3,3}$ subdivisions.
We assume that $G_{i+1}$ was obtained from $G_i$ by contracting the edges of an induced matching all of whose vertices have degree at most $\Delta$  and (for simplicity in this introductory  example) there are no green vertices of degree 2, and every edge $e=xy$ has been contracted to a green vertex $v_e$ which has no non-green neighbours. We also assume (again for simplicity)  that we only want to find the desired two paths or a planar reduction in $G_i$, we do not need to find a ferociously strong planar reduction.

We note that since $D_{i+1}$ is 3-connected, $D_{i+1}-v_e$ is 2-connected so the boundary $B$ of the face  $f$ of $D_{i+1}-v_e$ containing $v_e$ is a cycle.

As we just observed, $D_{i+1}$ has a unique embedding and we can assume that we have the rotation scheme for it.That is we know
 the clockwise order of the neighbours of $D_{i+1}$ around $B$. We determine which of these neighbours see $x$ and which see $y$.

We check first whether  (A) there exist distinct $z_1,z_2,z_3,z_4$ appearing in the given order around
 $B$ such that $x$ sees $z_1$ and  $z_3$ and $y$ sees $z_2$ and $z_4$.
If so, there is a $K_{3,3}$ subdivision in $G_i$ with set of centers $x,y,z_1,z_2,z_3,z_4$ which uses only  edges
which are in   $B$ or incident to one of  $x$ or $y$.

If not, and (B)  $x$ and $y$ have three  common neighbours then they both must have exactly three neighbours.
Now  there are three paths from $v^*$ to $v$ in $D_{i+1}$. For any  three  such  paths $P_1,P_2,P_3$
 there is a $K_{3.3}$ subdivision whose set of centers is $N(v) \cup \{x,y,v^*\}$ and whose edges are those contained in the $P_i$ and those incident to
$x$ or $y$ other than $xy$.

If neither of these possibilities occur, then after adding the edges from $x$ to $B$, we see that the edges from $y$ must lie on the boundary of  one of the faces into
which this partitions $f$. In this case we can add $y$ and the edges from it into this face.
We can and do determine which of these possibilities occurs and embed $x$ and $y$ if possible in constant time per edge, simply by regarding the  order that  the set of neighbours of $x$ and $y$ appear n the rotation
scheme for $v$. We let $D^*_{i+1}$ be the resultant 3-connected planar graph
\footnote{This graph is 3-connected because since $D_{i+1}$ is 3-connected for any 2-cut one of the components would have to be one endpoint of an uncontracted edge but they all have degree 3
since $G_i$ is 3-connected.}.

We know that either the desired two paths exist or  we can obtain a   reduction  of $G_i$ which  separates off 5 of the 6 centers of every $K_{3,3}$ subdivision from $C$ by a 3-cut and hence is planar.
We claim that if we can find a reduction of $G_i$ cutting off  5 of the 6 centers  of the special $K_{3,3}$ subdivision we defined for each $xy$ which we could not embed
then this reduction is  planar,  proving the desired two paths do not exist in $G_i$. 
We can of curse assume that the only cuts used in the reduction separate off 5 of the 
6 centers of such a subdivision.   In order to prove our claim  we first  show that
for such a  reduction  and every edge $xy$ which we did not embed,  one of $x$ or $y$  is not in the 3-cut $X$  of the reduction separating off 5 of the 6 centers of the corresponding subdivision
but rather is  separated  from $C$ by $X$.
We then show that this implies that  the reduction   is a minor of $D_{i+1}$ and hence is planar.

Now, both $x$ and $y$ have edges to three other centers of the subdivision, and four  other centers, if (B) holds. So, without loss of generality,  if neither is  cut off by $X$
then either both $x$ and $y$ are in $X$ or  (A) holds  and $X=x,z_2,z_4$.
In the latter case, there must be a path joining $y$ to $v^*$ outside of $X$. This path passes through $B$ and hence intersects
the component of $B-X$ containing either $z_1$ or $z_3$. So, at least one of these centers is also not cut off by $X$, which is a
contradiction. So, $x$ and $y$ are in $X$.
But, there are three vertex disjoint paths of $D_{i+1}-v$ from $v^*$ to $N(v)$ and hence from $v^*$ to $B$. As  $X$ contains both $x$ and $y$
there are 2 paths of $D_{i+1}-X$ from $v^*$ to $B$. Furthermore for some vertex $z$, $B-z$ is disjoint from $X$ and hence there are two paths of $D_{i+1}-X$
from $v^*$ to any two vertices of $B-z$. But $B$ contains at least 3 centers of the subdivision so $X$ separates off at most four of the centers from $v^*$.
Thus we see that, without loss of generality, $y$ is not in $X$ but separated from $v^*$ by $X$.

Now, letting $U$ be the  union of the components cut off from $C$ by  the 3-cuts of our reduction,  if we contract all the edges of our matching incident to a vertex of $U$
 then we are left with the three-connected planar $D^*_{i+1}$.
 We claim  that for every 3-cut $X$  of the reduction   there is a component of $U$  which was cut of by $X$ which has a non-empty  intersection with $D^*_{i+1}$. It follows that our reduction is obtained from $D^*_{i+1}$ by  putting cliques on some laminar 3-cuts
 and deleting the components they cut off. So it is also planar.

 To prove our claim, we note that if some component cut off by $U$  is to completely disappear  it consists of a set of vertices each of which is  matched by  our matching to a vertex of $X$.
 Since the contracted matching is induced, it is a vertex, and sees all three vertices of $X$. Once again applying the fact that the matching is induced, we obtain that there is
 only one such component.This is impossible since  $X$ cuts five of  centers of a 
 subdivision from $v^*$.

So, rather than looking for a set of $3$-cuts separating off 5 of the 6 centers of every $K_{3,3}$ subdivision we only need to consider one specially defined local
subdivision per uncontracted edge. The fact that they are locally defined means we can deal with them using algorithms for graphs of bounded tree width, as we discuss
in Section \ref{uncompactiondetails}.

\section{Uncompaction: The Details}
\label{uncompactiondetails}

 In this section, we solve the following problem which as discussed in Section \ref{uncompactionoutline}, is the core of our uncompaction procedure.

Given a 3-connected graph $G_{i}$ compacted to a 3-connected  graph $G_{i+1}$ for which we have found  a  ferociously strong planar reduction and associated coloured minor $F_{i+1}$, we wish to find either a  ferociously strong  planar reduction of $G_i$ and associated coloured minor $F_i$ or  vertex disjoint paths $P_1$ and $P_2$ of $G_i-v^*$ such that $P_i$ joins $s_i$ to $t_i$.

We note that we fully described how to solve this problem in linear time  if $G_{i+1}$ is obtained from $G_i$ by
contracting a set of triangles in Section \ref{uncompactionoutline}. So we assume this is not the case.

Our algorithm has to do two things. It has to identify which parts of the uncompaction can be performed in such a way
that we retain a ferociously strong planar reduction using the same cuts, and which lead to non-planarity which needs to be cut
off by new cuts.  It also has to either find the desired 2 paths or a new ferociously strong planar reduction cutting off the
non-planar bits it has identified.

We can do the first part by a simple local analysis. In doing the second, as mentioned in Section \ref{uncompactionoutline},
we exploit the fact that every  3-cut exists  in a local area of the graph, which has bounded tree width. So, we can
look for these cuts  or  the desired 2 paths using algorithms for graphs of bounded tree width.

In the next section, we define tree decompositions and tree width and present the specification of some algorithms
which will be useful in looking for 3-cuts or the two desired paths.

Then we describe exactly which graphs we decompose and how we handle the pieces. This will vary depending on whether we
obtained $G_{i+1}$ from $G_i$ by contracting a matching, deleting some edges or deleting some vertices.

Finally, we present the details of the algorithms  whose specifications are given in the next section.

We remark, that  along with the minor $F_{i+1}$ we  store  a three colouring of $V(G_{i+1})$ with red,yellow, and green.
The vertices of $F_{i+1}$ are the green vertices and the components of the red and yellow graphs.
This  naturally leads to a three colouring of $V(G_i)$ with red, yellow and green. By doing a breadth first search on the
red and yellow graphs, we can determine the vertex of $F_{i+1}$ to which each non-green vertex of $G_i$ has been
contracted. If we are uncontracting a matching $M$ , it is also an easy matter to determine which green vertex of
$F_{i+1}$ corresponds to a green edge of $M$.
We can store this in an array indexed by the vertices of $G_i$. So,  after this linear time preprocessing,
we can  determine the pair of vertices of $F_{i+1}$ joined by an edge of $G_i$ in constant time.

\subsection{Building and Exploiting Tree Decompositions}

A {\it tree decomposition}  $[T,{\cal S}]$ of a graph $J$ consists of a
tree $T$ and a subtree $S_v$ of $T$ for each vertex $v$ of $J$ such that
if $uv$ is an edge then $S_u \cap S_v \neq \emptyset$. For each node $t$
of $T$, we set $W_t =\{v|t \in S_v\}$ and define the torso of $T$ to be the
graph obtained from $G[W_t]$ by adding edges so that $W_s \cap W_t$ is a clique
for each neighbour $s$ of $t$.

The {\it adhesion} of a tree decomposition is the maximum of $|W_s \cap W_t|$.
The {\it torso} corresponding a node $t$ of  a tree decomposition of  is the graph
obtained from $J[W_t]$ by adding edges so that $W_s \cap W_t$ is a clique

Many results in graph theory state that a given class of graphs has a tree decomposition
of bounded adhesion where the torsos have certain properties. Indeed our result on 2-DRP
is that  the desired two paths do not exist precisely if the auxiliary graph has a tree decomposition of adhesion
three  where the tree is a star and the torso of  its  centre is planar and contains the four terminals.

We can exploit a decomposition of bounded adhesion with simple torsos   to solve a specific problem  on $J$ as follows.
We root the tree $T$ of the tree decomposition and for each node
$t$ of  $T$ we let $J_t$ be the subgraph induced by those
$v$ for which $S_v$ contains $t$ or one of its descendants. Our dynamic programming
algorithms work by processing the nodes of $T$ in post-order and  for each $t$ in turn solving a problem on $J_t$ using the solutions on $\{J_c |c~is~a~child~of~t\}$.
This allows us to restrict our attention to subproblems on the simple torsos.

This is especially effective if the torsos are small.
The {\it width} of a tree decomposition is $max_t ~|W_t|-1$.
The {\it tree width} of $J$ is the minimum of the width of its decompositions.

 We  give here the specifications  of the (dynamic programming) algorithms which we will
apply to graphs of bounded tree width, to find either a reduction or a 2-path.
We also present a lemma which shows that the graphs we will be considering have bounded tree width.
First however, we point out some useful properties  of  the type of cuts which  we will be looking for.

\subsubsection{Some Special Types of Cuts}

Our reductions use some  special types of 3-cuts.

\begin{definition}
For a vertex $v$ in a 3-connected graph $H$, we say a 3-cut $Z$ disjoint from $v$ is {\it $v$-non-shiftable}
if  every vertex of $Z$ has two neighbours which are not in the component
of $G-Z$ containing $v$.
\end{definition}

\begin{definition}
For a set $W$  of vertices  in a 3-connected graph $H$, we say a 3-cut $Z$ disjoint from $W$ is {\it $W$-ferocious}
if  there are two connected subgraphs with an edge to every vertex of $Z$ contained in components of $H-Z$
not intersecting $W$.
\end{definition}

We note that $Z$  is $W$-ferocious precisely if  we can colour the vertices it cuts off from $W$ using red and yellow
so that there is one component in the red graph, this component and every component  of the yellow graph has an edge
to all the vertices of $Z$ and either there is only one yellow component or there is no edge from the red component
to any of the yellow components.
Clearly, in a ferociously strong reduction, we are using $v^*$-non-shiftable $\{v^*\}$-ferocious cuts. This explains why our algorithms for graphs of bounded tree width  will focus on such cuts.

We present here three   useful auxiliary lemmas. The first tells us that if we have chosen some $v^*$-non-shiftable
cuts which cut off some of the non-planarity we are trying to deal with, this does not affect whether some other specific
bit of non planarity can be cut off.  The second presents a similar result for ferocious cuts we consider. The third  is a technical lemma which allows us to  reduce the problem of looking for a reduction  or the  desired 2-paths  to similar problems on each of the  torsos of any tree decomposition where the $W_s  \cap W_t$ correspond to  $v^*$-non-shiftable cuts.

\begin{definition} Minimal
cutsets $Y$ and $Z$ of $J$  are laminar if there do not exist two distinct components of $J-Z$
containing vertices of $Y$.
\end{definition}

\begin{remark} Since the cutsets are minimal there  do not exist two distinct components of $J-Z$
containing vertices of $Y$ precisely if there  do not exist two distinct components of $J-Y$
containing vertices of $Z$.
\end{remark}

\begin{lemma}
\label{lamlem}
Suppose that  $J$ is a root graph   and that
for some $K_{3,3}$ subdivision $Z$ is a 3-cut separating five of the six centers  of the subdivision  from $C$  in $J$ and maximizing the size of the components of $J-Z$ intersecting $C$. Then $Z$ is laminar with all the $v^*$-non-shiftable  3-cuts of $J$.
\end{lemma}

\begin{proof}
Suppose for a contradiction that there is a $v^*$ non-shiftable  3-cut  $Y$  of $J$  which is not laminar with $Z$.
Clearly $Z$ does not contain $v^*$. We let $K$ be the component of $J-Z$  containing $v^*$. By Claim \ref{k33subdivisionclaim},  there are two connected subgraphs of $J-Z$ disjoint from $K$ which have edges
to all of $Z$. Since $Y$ is not laminar with $Z$, it must intersect both these subgraphs and $K$, and hence is disjoint from $Z$. So $K$  contains exactly one vertex $y$ of $Y$.

Since $J$ is 3-connected, the component $U$  of $J-Z-y$ containing $v^*$ must have edge to two vertices of $Z$.
Since $Y$ is $v^*$-shiftable, $y$ has edges to two vertices outside  the component of  $J[V(K) \cup Z]$ containing $U$
and  its neighbours. One of these $a$ must be in $K$, and hence in a component of $J-Z-y$ which has edges only to
$y$ and at most one vertex of $Z$. But this contradicts the fact that $J$ is a root graph.
\end{proof}

In the same vein, we have:

\begin{lemma}
\label{lamlem2}
Every pair of $W$-ferocious 3-cuts of $J$ are laminar.
\end{lemma}

\begin{proof}
Suppose for a contradiction that there are non-laminar $W$-ferocious cuts $Y$ and $Z$.
Let $K$ be a component of $J-Z$  intersecting $W$.  there are two connected subgraphs of $J-Z$ disjoint from $K$ which have edges  to all of $Z$. Since $Y$ is not laminar with $Z$, it must intersect both these subgraphs and $K$, and hence is disjoint from $Z$. So $K$  contains exactly one vertex $y$ of $Y$.

Since $J$ is 3-connected, there is a  component $U$  of $J-Z-y$ containing a white vertex of $K$
which  must have edge to two vertices of $Z$.
We know  that $y$ has two neighbours in the components of $J-Y$ disjoint from $W$ and hence
outside  the component of  $J[V(K) \cup Z]$ containing $U$.
and  its neighbours. One of these $a$ must be in $K$, and hence in a component of $J-Z-y$ which has edges only to
$y$ and at most one vertex of $Z$. But this contradicts the fact that $J$ is 3-connected.

\end{proof}

Finally, we need:

\begin{definition}
A  (ferociously strong) quasi-reduction of a graph $J$ with respect to a triangle $Tri$ within it, is a   three  colouring of $J$ using  green, yellow, and red  such that the vertices of $Tri$ are coloured green,  the  minor obtained from the  subgraph formed by the non yellow vertices   by contracting its red components to vertices   is planar
and for every component $K$ of the
graph $J^*$  formed by  the  non-green vertices  there are edges from $K$ to exactly 3 green vertices and either (a)   $K$ is coloured yellow and  exactly one  of the components
of $J^*$  with the same three green attachment vertices  is coloured red, or     (b)   $K$ is coloured red   and there is   at least one  other component
of $J^*$ with the same three green attachment vertices  all of which are coloured yellow, or (c) there is no other component of $J^*$ with the same attachments as $K$ and both the red and yellow subgraphs of $K$ are connected and have edges to each of the three green  attachment vertices of $K$.
\end{definition}

\begin{lemma}
\label{cavredlem}
Suppose that  we are given a rooted  tree decomposition $[T,\{S_v | v \in V(J)\}]$  of adhesion 3  of a root graph $J$  such that
$C$ is contained in $W_r$ for the root $r$, and for every arc from a node $t$ of the tree to its parent $p(t)$,  each vertex of  $W_{p(t)} \cap W_t$, has at least two neighbours in $J_t$. Then given a ferociously strong reduction of the torso of the root  and a quasi-reduction
of the torso  with respect to $W_t \cap W_{p(t)}$ of every other node, we can find a ferociously  strong reduction of $J$ in linear time.
\end{lemma}

\begin{proof}
We will take the 3-colouring of the torso and extend it to a 3-colouring of the whole graph.
We do so by considering the nodes in post order. At all times we have a 3-colouring of the set $S$ of  vertices in the torsos
of the nodes we have considered which is a ferociously strong reduction of the graph obtained from $J[S]$ by adding a
clique on $W_s \cap W_t$ for every arc of the tree between a considered and unconsidered vertex.

When considering $t$, if $W_t \cap W_{p(t)}$ is green in our colouring for $S$ then we simply add $t$ to the set
of considered nodes and leave the colouring on its torso unchanged. If  one of red or yellow appears more often then the
other on $W_t \cap W_{p(t)}$,
then  we colour every vertex of $W_t-W_{p(t)}$ with this colour. Otherwise, $W_t \cap W_{p(t)}$ contains  exactly one red vertex $r$ and one yellow vertex $y$,  and one green vertex $g$. We  let $r'$ and $y'$ be
two neighbours of $g$  in $G_t$  with $y=y'$ if possible  and $r=r'$ if possible. We find two internally vertex disjoint paths
from $y$ to $y'$ and $r$ to $r'$ avoiding $g$ in the (3-connected) torso for $t$. We colour the first  yellow and the other red, and then extend this, using breadth first search  to obtain a partition of the torso  into two connected  coloured subgraphs.
\end{proof}

\subsubsection{Some graphs of bounded tree width}

The {\it face vertex} incidence graph of a subdivision $J$  of a 3-connected planar graph is the bipartite
graph whose vertex set is the union  of the vertices of $J$ and the faces in its unique embedding in which a face is joined
precisely to the vertices on its boundary.  Its diameter is the maximum of the minimum distances between the
pairs of vertices within it. Robertson and Seymour\cite[Theorem 2.7]{gmiii} proved:

\begin{lemma}
\label{brbtw}
A planar graph of diameter $d$ has tree width at most $3d+1$.
\end{lemma}

We also have:

\begin{observation}
If contracting the edges of a matching $M$ in $J$ yields a graph of tree width $w$ then $J$ has tree
width at most $2w+1$.
\end{observation}

We simply modify our tree decomposition by, for each edge $xy$ of $M$ contracted to a vertex $v$,
deleting $S_v$ and adding $S_x$ and $S_y$ with $S_v=S_x=S_y$. Making copies of trees, corresponding
to vertices we duplicate we obtain.

\begin{lemma}
If $J$ is a graph obtained from a graph of  tree width at most $w$  by taking up to $k$ twins of some of its
vertices and adding any set of edges within the set of vertices formed by a vertex and its twins then the
tree width of $J$ is at most $kw+k-1$.
\end{lemma}

A slightly more sophisticated argument shows:

\begin{lemma}
If $J^*$ is a graph obtained from a graph $J$ of  tree width at most $w$  by taking up to $k$ twins of  a set $S$ of vertices
all of degree at most $d$ which form a stable set  such that every set of vertices formed by a vertex and its twins is stable
then the tree width of $J$ is at most $dw+d$.
\end{lemma}

\begin{proof}
This is true for the graph obtained from $J$ by taking $d$ twins of each vertex $v$ in $S$.
Since the tree width of a minor of a graph is at most the tree width of the graph,
contracting each such twin to a different neighbour of the corresponding $v$ shows that the
graph obtained from $J$ by adding cliques on the neighbourhoods of the vertices of $S$
has tree width at most $dw+d$.  Finally, adding a vertex whose neighbourhood is a clique
to a graph cannot increase the tree width unless it is less than the number of neighbours of the
vertex. So, adding, for each $v$ in $S$,  $k$ vertices seeing exactly the neighbourhood of $v$, we obtain a supergraph
of $J^*$ with tree width at most $dw+d$.
\end{proof}

\subsubsection{Algorithms for  Graphs of Bounded Tree Width}

We will need four algorithms specifications which follow.  We show here that the first three are easy using the fourth as a subroutine. We delay presenting the details of the fourth.

\begin{algspec}\label{fsrtwk}
\textsc{Ferociously Strong Reductions For Tree Width $k$}
~\\
\emph{Input:} A   root  graph $J$ with  tree width at most $k$.
~\\
\emph{Output:}
Either
\begin{enumerate}
\item
A  ferociously strong planar reduction $L$   of $J$ and corresponding 3-coloured minor, or
\item
two vertex disjoint paths $P_1$ and $P_2$  of $J-v^*$ such that $P_i$ contains $s_i$ and $t_i$.
\end{enumerate}
\emph{Running time:} $O(|V(G)|)$.
\end{algspec}

\begin{algspec}\label{qsrtwk}
\textsc{Quasi-Reduction For Tree Width $k$}
~\\
\emph{Input:} A   3-connected graph $J$ with  tree width at most $k$,  containing a triangle $Tri$. .
~\\
\emph{Output:}
A quasi reduction of $J$ with respect to $Tri$.
~\\
\emph{Running time:} $O(|V(G)|)$.
\end{algspec}

\begin{algspec}\label{wfctwk}
\textsc{$N$-Constrained  $W$-Ferocious  Cutsets For Tree Width $k$}
~\\
\emph{Input:} A   3-connected graph $J$ with  tree width at most $k$ and disjoint subsets $W$ and $N$ of its vertices.
~\\
\emph{Output:}
A  set  ${\cal F}$ of $W$-ferocious 3-cuts each of which is
not cut off from $W$ by another $W$-ferocious 3-cuts,
contains a vertex not in  $N$, and
subject to this maximizes the number of vertices cut off from $W$ by the cutsets.

A 3-colouring of the vertices of $J$  such that those vertices not cut off  from $W$ by an element of ${\cal F}$ are coloured green and for each cut $Z$ of ${\cal F}$ those cut off from $W$ by $Z$  are 2-coloured red and yellow so as to show that $Z$ is $W$-ferocious.
~\\
\emph{Running time:} $O(|V(G)|)$.
\end{algspec}

We remark that the cuts of ${\cal F}$ are laminar.

\begin{definition}
For a vertex $v$ and sets $N$ and $W$ of a  3-connected graph $G$,  a closest  $N$-constrained 3-cut   for
$v$ from $W$  exists if $v$ can be cut off from $W$ by a cutset of size 3, which is not contained in $N$.
In this case, it is such a 3-cut which minimizes the size of the component containing $v$.
\end{definition}

\begin{algspec}\label{wqfctwk}
\textsc{$N$-Constrained  $W$-Ferocious  or $Q$-closest  Cutsets For Tree Width $k$}
~\\
\emph{Input:} A   3-connected graph $J$ with  tree width at most $k$ and  subsets $W$, $N$ and $Q$
of its vertices   such that $W$ and $N$ are disjoint  and the $N$-constrained  closest 3-cuts for  $v$ from $W$ for the $v \in Q$
are laminar with each other and with the $W$-ferocious 3-cuts which are not contained in $N$.
~\\
\emph{Output:}
A  set  ${\cal F}$ of 3-cuts each of which contains a vertex not in  $N$ and is either
\begin{compactenumerate}
\item $W$-ferocious or
\item a closest $N$-constrained $3$-cut for $v$ from $W$ for some $v \in Q$,
\end{compactenumerate}
such that no cut is separated from $W$ by another $W$-ferocious or closest $N$-constrained cut and subject to this maximizes the number
of vertices cut off from $W$ by the cutsets.
For each cut $Z$ in ${\cal F}$, either
a 2-colouring of the vertices of $J$ cut off from $W$ by $Z$ showing that it is $W$-ferocious or a vertex $v$ of $Q$ for which it is
a closest $N$ constrained 3-cut for $v$ from $W$.
~\\
\emph{Running time:} $O(|V(G)|)$.
\end{algspec}

Again, the cuts of ${\cal F}$ are laminar.

Now,  Algorithm  \ref{wfctwk} simply applies Algorithm \ref{wqfctwk}  with $Q$ empty.

Algorithm \ref{fsrtwk}  first determines if the
two paths exist using the algorithm in \cite{gmii}. If they do it returns them. Otherwise, it applies
Algorithm \ref{wfctwk} with $W=\{v^*\}$ and $N$ empty.  We then add edges so that each of the cutsets in the set ${\cal F}$ returned is a clique.

Algorithm \ref{qsrtwk}  applies
Algorithm \ref{wfctwk} with $W$ the vertices of $Tri$,  and $N$ empty.  We then add edges so that each of the cutsets in the set ${\cal F}$ returned is a clique.

\subsubsection{Algorithms using MSO formulas}

We describe a linear extended MSO formulations of the output and apply Arnborg et al's result to obtain linear time algorithms \cite[Theorem 5.6]{bar:ALS}.

\textbf{Algorithm \ref{fsrtwk}.} We search for the two outputs independently.

\textbf{Ferociously strong planar reduction}

\begin{lemma}
  The problem of finding a ferociously strong reduction in a \tricname\ with center $C$ can be stated in LinEMSOL.
\end{lemma}

\begin{proof}
  The idea is to attempt to find a 3-colouring of $G$ so the first colour is the vertices in the reduction has colour 1 and each remaining component is 2-coloured so that vertices of colour 2 within that component are connected and adjacent to all of its separator.

  We use 3 predicates: $\FRed$ which are true for vertices in the reduction and $\FFirstComp, \FSecondComp$ for the partition of remaining vertices.

  We use some well known formulas such as $\FConn(X)$ to determine if a set of vertices is connected, $\FRedConn(X)$ to test if $X$ is connected using only $\FRed$ vertices and $\FPart{k}(X,X_1,\ldots,X_k)$ to determine $X_1,\ldots,X_k$ is a partition of $X$ into $k$ sets. We use $\FPlanar(X)$ to test if the graph induced by $X$ is planar (we can do this, for example, by testing if it contains a $K_5$ or $K_{3,3}$ minor).
  \[
  \FIsRedEdge(x,y) = ~xy \in E(G) \lor \left(\exists x',y',~\neg \FRed(x') \and \neg \FRed(y') \land \FNonRedConn(\{x',y'\}) \right)
  \]
  This allows us to test the graph of $\FRed$ vertices (that should be in the reduction) with edges defined by $\FIsRedEdge$ for planarity using a formula $\FPlanarRed$.

  We need two more sub-formulas we will use often.
  \[
  \begin{array}{l}
  \textsc{Maximal}(X) = \FNonRedConn(X) \land (\forall y \neg \FNonRedConn(X \cup \{y\})) \\
  \textsc{Boundary}(X, Z) = \exists Z=\{z_1,z_2,z_3\} \land (z_1 \not \in N \lor z_2 \not \in N \lor z_3 \not \in N) \\
  \land \FAdjacent(z_i, X) \forall i \land \forall y \in V(G) - Z \neg \FAdjacent(y, X)
  \end{array}
  \]

  We abuse notation and write $\exists Z=\textsc{Boundary}(X)$ for $\textsc{Boundary}(X, Z)$ and write $\exists X_1,X_2 = \FPart{2}(X)$ for $\FPart{2}(X,X_1,X_2)$.

  The following formula ensures every component not in the reduction is partitionable into two with the desired adjacencies.
  \[
  \begin{array}{l}
  \textsc{Paritionable} = \forall X, \textsc{Maximal}(X) \rightarrow \\
  (\exists X_1, X_2 = \FPart{2}(X))\\
  (\exists Z = \textsc{Boundary}(X)) \\
  \land \FFirstComp(X_1) \land \FSecondComp(X_2)
  \land \FFirstCompConn(X_1) \land \FSecondCompConn(X_2)\\
  \land \left(\forall x \in X, y \in Z,
  \FAdjacent(X_1,\{y\}) \land \FAdjacent(X_2,\{y\}) \right)
  \end{array}
  \]

  Now we simply wish to find a colouring satisfying all these conditions (i.e., $\FPlanarRed \land \textsc{Paritionable}$ reduction planarity and partitionability of non-reduction components) which maximizes the number of reduction vertices (a clearly linear objective).
\end{proof}

\textbf{2-DRP}

Suppose we want to determine if there are two vertex disjoint paths in $G$, one from $s_1$ to $t_1$ and the other from $s_2$ to $t_2$.
  \[
  \begin{array}{l}
    (\exists~X_1,X_2,X_3 = \FPart{3}(V(G)))\\
  \land \FFirstComp(X_1) \land \FSecondComp(X_2)
  \land \FFirstCompConn(X_1) \land \FSecondCompConn(X_2)\\
  \land \FFirstComp(s_1) \land \FFirstComp(t_1) \land \FSecondComp(s_2) \land \FSecondComp(t_2)
  \end{array}
  \]

We only need to check if this formula is satisfiable.

\textbf{Algorithm \ref{qsrtwk}}. This is almost the same formula as for ferociously strong planar reduction but the vertices of $Tri$ are forced to be reduction vertices. I.e., it is the conjunction of the above formula and $\FRed{Tri}$ for the conditions (the objective remains the same).

\textbf{Algorithm \ref{wfctwk}.}

The constraint is $\FRed(W) \land \FPlanarRed \land \FFerocious$ where $\FFerocious$ is defined below.
\[
  \begin{array}{l}
  \textsc{BoundaryN}(X, N, Z)\exists Z=\{z_1,z_2,z_3\} \land (z_1 \not \in N \lor z_2 \not \in N \lor z_3 \not \in N) \\
    \land \FAdjacent(z_i, X) \forall i \land \forall y \in V(G) - Z \neg \FAdjacent(y, X) \\
  \end{array}
  \]
  We again abuse notation and write $\exists Z=\textsc{Boundary}(X, N)$ for $\textsc{Boundary}(X, N, Z)$. We need the following two sub-formulas.
  \[
  \begin{array}{l}
  \FIsFerocious(X, Z) = \exists X_1,X_2 = \FPart{2}(X)\\
  \land \FFirstComp(X_1) \land \FSecondComp(X_2)
  \land \FFirstCompConn(X_1) \land \FSecondCompConn(X_2)\\
  \land \left(\forall x \in X, y \in Z, \FAdjacent(X_1,\{y\}) \land \FAdjacent(X_2,\{y\}) \right)
  \end{array}
  \]
  So our new constraint is the following.
  \[
  \begin{array}{l}
  \FFerocious = \forall X, (\exists Z=\textsc{BoundaryN}(X, N) \land \textsc{Maximal}(X)) \rightarrow \\
  \FFerocious(X, Z) \land \\ (\neg \exists X' \supset X, (\exists Z'=\textsc{BoundaryN}(X', N) \land \textsc{Maximal}(X') \land \FFerocious(X', Z')))\\
  \end{array}
  \]
  
  The objective is maximizes the number of vertices $v$ with $v \in N \land \FNonRed(v)$.

  From the solution obtained, we can easily deduce the 3-cuts from the set of non-reduction components.

\textbf{Algorithm \ref{wqfctwk}.}
The constraint is $\FRed(W) \land \FPlanarRed \land \textsc{FerociousOrClosest}$ where $\textsc{FerociousOrClosest}$ is defined below.

We use the $\FIsFerocious(X, Z)$ from the above along with
  \[
  \begin{array}{l}
    \FClosest(X, Z) = (\exists Z=\textsc{Boundary}(X), \exists v \in Q, \\
    \FSeparate(X, v) \land \forall Y \subseteq W \neg \FSeparate(X, v))\\
  \end{array}
  \]

  where $\FSeparate(X, v) = \textsc{BoundaryN}(X, N, Z) \land v \in X$, Finally,

  \[
  \begin{array}{l}
    \textsc{FerociousOrClosest} = \forall X, (\exists Z=\textsc{BoundaryN}(X, N) \land \textsc{Maximal}(X)) \rightarrow \\
  (\FFerocious(X, Z) \lor \FClosest(X, Z)) \land \\
  (\neg \exists X' \supset X, (\exists Z'=\textsc{BoundaryN}(X', N) \land \textsc{Maximal}(X') \land \\ (\FFerocious(X', Z') \lor \FClosest(X', Z'))))
  \end{array}
  \]

\subsection{Adding Back Edges}

If $G_{i+1}$ was obtained from $G_i$ by deleting a set $X$ of edges,  we want to either find the desired two paths in $G_{i}-v^*$ or find a reduction $F_i$ of $G_i$.
We know that any such reduction is also a reduction of the graph $G'_i$ in which we maintain the contractions which yielded the red and yellow  vertices of $F_{i+1}$.  We need to add the edges of $X$  to $F_{i+1}$, determine which edges create
non-planarity and then either use the non-planarity to find the desired two paths or find the cuts which
separate the non-planarity from $C$ thereby yielding a ferociously strong reduction.

In a preprocessing step, we show we can reduce to the  graph $D'_{i+1}$ formed by the  planar graph $D_{i+1}$ (obtained from the green subgraph of $F_{i+1}$  by adding a triangle on the neighbourhood of each red vertex) and the subset $X'$ of $X$
with both endpoints in this set.

We then  use a linear time algorithm of Eppstein and Reed\cite{ERplanarlaminar} which  given a  vertex $v$ of a   3-connected planar graph finds a maximal set  of   $v$-non-shiftable laminar 3-cuts within it  We apply this with $v=v^*$ to $D_{i+1}$.
Because of the connectivity conditions on the edges of $X$, this also yields a decomposition of $D'_{i+1}$.
By Lemma \ref{cavredlem}, we have reduced our problem to that of finding (i) for each non-root node $t$  a quasi-reduction
of the torso with respect to $W_t \cap W_{p(t)}$, and (ii) either the desired two vertex disjoint paths
or  a ferociously strong reduction in the torso of the root.

By Lemma \ref{lamlem} we know that all of the cuts used by a ferociously strong reduction will be $W_t \cap W_{p(t)}$
for some arc of the resultant  tree decomposition. Thus if the  torso for $t$  is non-planar then if $t$ is the root the desired two
paths exist while otherwise the only quasi-reduction of the torso with respect to $W_t \cap W_{p(t)}$  has no green vertices except these three.If the torso is planar then we
simply colour it all green. Otherwise working locally and using our algorithms for graphs of bounded tree width, we find the
desired two paths or quasi-reduction. Forthwith the details.

We know that  for each edge $xy$ of $X$ there are $c$ internally vertex disjoint paths
of $G_i$ between $x$ and $y$. Thus,  if  the edge has one endpoint which is  in the subgraph contracted to a  red or yellow vertex of $F_{i+1}$ then the other
endpoint is either in the same vertex of $F_{i+1}$, in one of its neighbours, or in  a twin to which it is already joined by an edge.  So, we can add  such edges. Next, we add any edge of $X$  joining two of the three neighbours of a yellow vertex, as all these edges  lie in a face and no pair of them cross.
We let $X'$ be the set of edges of $X$ we have not yet added back.

Now, for each red vertex $w$, we delete $w$ and its twin set and add edges so $N(w)$ is a clique to obtain
 the   3-connected  planar graph $J=D_{i+1}$. It is an easy matter to do this in linear time.  Lemma \ref{cavredlem} implies that, letting $J'$ be the graph obtained from $J$ by adding the edges  in $X'$,  given the desired ferociously strong reduction or two disjoint paths in $J'$ we can find the same object in $G'_i$.
For every edge $xy$ of $X'$,    any cutset
$Z$ of $J'$ separating $x$ from $y$ is also  a cutset in $G_{i+1}$. So, $x$ and $y$ are joined by at least $c>24$
internally disjoint paths of $J$.
Furthermore, since $x$ and $y$ do not have a common yellow neighbour, none of these paths is an edge.

Now, we find a maximal set of $v^*$ non-shiftable laminar  3-cuts of $J$ and corresponding tree decomposition
of adhesion 3  for $J$ using the algorithm of Eppstein and Reed\cite{ERplanarlaminar}.
Because  the endpoints of every edge we want to add are joined by $24$ paths of $J$,  for every edge $xy$ of $X'$  there is a node of the tree decomposition containing  $x$ and $y$. For any such $x$ and $y$ which lie in a 3-cut $X$ of the tree decomposition we add the edge $xy$ to $J'$, maintaining planarity since these
cuts  are laminar and any two vertices in a minimum cut  lie in a face.  We also delete such edges from $X'$.  So we can assume each edge we need to add  lies in the torso of a unique node of the tree decomposition.

We traverse the tree starting at the root. For each node $t$,  we try to add the set of edges $X_t$  it remains to add  with both endpoints in the torso  to the torso $J_t$ for $J$ so as to maintain planarity. If we are able to do so, we add them and move on to consider
the children of $t$. If we are unable to do so  then  (i) if $t$ is the root we will find the desired two paths, if (i) $t$ is not the root
then for the parent $p$ of $t$, we will find a partition of the torso of $T-W_{p(t)}$  into  a red
 vertex and a non-empty set of  yellow vertices each with an edge to all of $W_{p(t)} \cap W_t$.

We first add any edge of $X_t$ which is parallel to an edge of $J_t$ (we can do this in linear time by lexicographically sorting
both edge sets)
We note next that  because $J_t$  is 3-connected, every two non-adjacent vertices lie together on at most one face.
Our first step is to determine for each edge of $X_t$, whether or not its endpoints  lie on the same face,
and if so which.  This can be done in linear time as follows.

We take an ordering of the nodes of the  face-vertex incidence graph  for $J_t$ so that each
node sees only 5 nodes earlier in the order. We have an array indexed by the nodes listing the other nodes which it is adjacent to  which appear earlier in the order.
Using  two queries to the entries in this array for  the endpoints of an edge of $X_t$, we can determine in constant time whether they  lie on the boundary of  a face which appears before
 both  the  endpoints in the order.
We then traverse the order again, and using a query into the array for a vertex and then a query for each face on the list returned, we can
 construct, in linear time,  another  array containing for each vertex $v$  a list of the  at most 25 vertices which lie on a face with $v$ which lies
after the other vertex but before $v$  in the order.
Using this
array, in constant time per edge, we can determine  which edges of $X_t$  lies on a face which appears between its endpoints in the order.  Finally, for each face, we determine all of the at most 25 pairs of vertices on its boundary both of which appear earlier in the order, We create a list  which contains for each face and pair of vertices on its boundary and before it in the order,
of the pair and the face.
We sort  the union of $X_t$ and this list (ignoring the face names in the sorting)
 using lexicographic bucket sort in linear time. An edge  of $X_t$ which is a pair of vertices on a face which lies after both vertices  will be consecutive
with the corresponding triple and so we can determine all such edges.

We let $X^1_t$  consist of those edges in $X$ whose endpoints lie
together on a face, and $X^2_t$ be $X_t-X^1_t$.

If $X^2_t$ is non-empty
then we let $a_1b_1$ be one of its edges labelled so $d(a_1)>d(b_1)$.

Otherwise, for each face $f$ of $J_t$ , we let $X_f$ be the set of edges of $X^1_t$ both of whose endpoints are in  the boundary of $f$.  We  partition $X^1_t$ into these sets. For each $f$ we determine if the auxiliary  graph  consisting of the boundary of $f$, a vertex adjacent to all
of  this  boundary, and the edges of $X_f$ is planar and if so find an embedding of it. Since every edge of the torso  is in at most two faces,
and the number of auxiliary edges in an auxiliary  graph is  at most the number of edges of the torso it contains, this takes  total time linear in the size of the torso.
If  the auxiliary graph is planar for every $f$,  we  can add the edges of $X_f$  to our embedding of $J_t$ as suggested by the auxiliary  embeddings, and move on to the children of $t$.
Otherwise, we let  $f$ be a face for which the corresponding graph was non-planar.
We note that  there must be two edges  $a_1a_2,b_1b_2$ of $X_f$ with distinct endpoints which  appear in the order $a_1,b_1,a_2,b_2$ around $f$.

We define $B$ to be $C$ if $t$ is the root, and to be $W_p \cap W_t$ otherwise. For each vertex  $w$ of the torso,
We let  $d_w$ be the distance in the face vertex incidence graph from $w$ to  some specific vertex $v'$
of $W_p(t) \cap W_t$ (or $v^*$ if $t$ is the root). It is an easy matter to construct the face vertex incidence graph and determine
these values in linear time using breadth first search.

\begin{lemma}
For some constant $k$, in linear time, we can find a  minor  $M$  of  $J_t$  of tree width at most $k$ which contains the vertices of $B$ (so we did not contract edges incident to these vertices) , such that
$M+a_1b_1$  contains a $K_{3,3}$ subdivision and  we cannot separate five of the six centers  of the subdivision from $B$ by a 3-cut of $M$ other than $B$ itself.
\end{lemma}

Having found $M$ ,  if  $t$ is the root, we can find the desired two paths of $G_i-v^*$ within it and are done.
Otherwise we  apply Algorithm \ref{fsrtwk} to $M$ and the triangle
on $B$.  In the output, the only green vertices will be in $B$.Since $J$ is planar, $J_t-W_{p(t)}$ must be connected.
So we can extend the red-yellow colouring of $V(M)-B$ to a red-yellow
colouring of $J_t+X_t-B$  using breadth first search and are done. So to complete this case, it remains to prove the lemma.

\begin{proof}

A cutset of $J_t$ separating  $a_1$ and $b_1$ also separates them in $G_{i+1}$ Thus, there must be $c$ vertex disjoint paths from $a_!$ to $b_1$ in the torso.

If $d(a_1)$ is  at most $600$ then we contract the components of  $J_t$   formed by  $w$ with $d_w>620$ into vertices
(it is easy to do this in linear time).
We can create no 2-cuts by doing so, the resultant graph $M$   is still 3-connected  and inherits a unique embedding.
It has bounded tree width.

There are  four internally vertex disjoint  paths from $a_1$ to $b_1$ in  $M$ as  any  three cut
of $M$  is disjoint from the vertices with $|d_w-i| \le 1$
for some $i$ between $606$ and $618$. Hence, any minimal such cutset  $Z$   of $M$ separating $a_1$ and $b_1$ is also a cutset of $J_t$ which is impossible.

If $a_1b_1$ is in $X^t_2$, these paths, together with the edge $a_1b_1$ and the boundary $Bd$  of the face of $M-a_1$ containing $a_1$ contain a $K_{3,3}$ subdivision
of $M+a_1b_1$, two of whose centers are
$a_1$ and $b_1$ and the other  four  of which lie in $Bd$.
If $a_1b_1$ lies in a face and crosses $a_2b_2$ then there are two internally disjoint $a_i$ to $b_i$ paths in $M-bd(f)$.
There is a $K_{3,3}$ subdivision in the union of these four paths, $bd(f)$ and the two crossing edges four of whose centers are $\{a_1,b_1,a_2,b_2\}$ and the other two of which are off $bd(f)$.

We claim $M$
contains   no 3-cut other than $B$ separating off $B$ from five of the six centers of our subdivision.
Any such cut $Z$   is disjoint from the vertices with $|d_w -i| \le 1$
for some $i$ between $604$ and $618$. Hence, it is also a cutset of the torso.

If $a_1b_1$ is in $X^2_t$ then since there are five disjoint paths from each of $a_1$ and $b_1$ 
to the other five centers, both $a_1$ and $b_1$ are separated from $B$ by $Z$. Thus, 
since there are three paths of $J_t$ from $a_1$  to $B$ all of which pass through $Bd$, and $b_1$ is not on $Bd$ but four of the centers are,
 it is easy to see that one of $a_1$ or $b_1$ is  in a component  $K$ of $J_t-Z$  disjoint from $B$.
 
 Otherwise if some $a_i$ or $b_i$ is not separated from $B$ by $Z$ then the path from $B$ to it 
 must go through $12$ cycles which are disjoint except at $a_{3-i}$ and $b_{3-i}$. 
 Hence $Z$ must contain both $a_{3-i}$ and $b_{3-i}$. But there are three paths from 
 $B$ to $bd(f)$ each of which must pass through $12$ cycles which are disjoint except 
 at $a_i$ and $b_I$. so $b_i$ must also be in $Z$. But now $Z$ is in $bd(f)$ and $J_t-bd(f)$ 
 is connected because $J_t$ is 3-connected. So, there are two centers of the subdivision in
 the same component of $J_t-Z$ as $C$, a contradiction.  So every $a_i$ and $b_i$ is separated from $B$ by $Z$ and so one must be cut off by $Z$.

 Hence, since the $a_i$ and $b_i$  have degree  at least 24, the component of $J_t-Z$ containing
 such a vertex 
has more than  $4$ vertices.  Since the torso is 3-connected, by shifting from a vertex in the cut that sees
none of the rest of the cut  and only one vertex in $K$ to its neighbour on $K$, up to 3 times,
we obtain a $v^*$-non-shiftable  cut of $J_t$ which contradicts our choice of a  maximal laminar set of
such cuts(we use this argument below repeatedly).

If $d_{a_1}>600$ and the distance in the face incidence graph between $a_1$ and $b_1$ is at most 150
then we let $J^+_t$ be the subgraph obtained by contracting the components of  all the vertices  at distance more than $300$  from $a_1$ in the face-vertex incidence graph for  the torso. It is an easy matter to construct this graph in linear time.

There are  four internally vertex disjoint  paths from $a_i$ to $b_i$ in  the subgraph $K^+$ of $J^+_t$
consisting of those vertices at distance less than  200 from $a_1$  as  any  three cut
of this subgraph is disjoint from the vertices  at distance $\{i-1,i,i+1\}$ from $a_1$
for some $i$ between $160$ and $190$. Hence, any minimal such cutset  separating $a_i$ and $b_i$ is also a cutset of the torso  which is impossible.

If $a_1b_1$ is in $X^t_2$, these paths, together with the edge $a_1b_1$ and the boundary $B$  of the face of $J^+_t-a_1$ containing $a_1$ contain a $K_{3,3}$ subdivision
of $K^++a_1b_1$, two of whose centers are
$a_1$ and $b_1$ and the other  four  of which lie in $B$.
If $a_1b_1$ lies in a face and crosses $a_2b_2$ then there are two internally disjoint $a_i$ to $b_i$ paths in  $K^+$
disjoint from $bd(f)$.
There is a $K_{3,3}$ subdivision in the union of these four paths, $bd(f)$ and the two crossing edges.

Since there are no $v$-non-shiftable 3-cuts  other than $B$  in the torso,
$J^+_t$ contains a subgraph $J'_t$ consisting of (i) disjoint  concentric cycles $C_1,...,C_{15}$  such that  all $w$
with   $|d_w- d_{a_1}<200$  lie within
$C_{15}$, $v^*$  lies outside $C_1$,  and all the vertices of any $C_i$  are at distance between 200 and  280 from $a_1$
together with 4  vertex disjoint paths
$P_1,P_2,P_3,P_4$   from $C_1$ to $C_{15}$ whose interior lies between  them such that the intersection of each $P_i$ with a $C_j$ is a path. We  can find such cycles  and paths and contract  $C_5$ into   a set of vertices  $X=\{x_1,...,x_4\}$ so that  $P_i \cap C_5$ is in $x_i$, in linear time.

In the same vein, letting $J^-_t$ be the subgraph obtained by contracting  the components in the graph induced by $w$ such that $d_w>300$, we can find in linear time
 a subgraph $J^*_t$  of $J^-_t$ consisting of disjoint  concentric cycles $C'_1,...,C'_{15}$  such that  all the vertices at distance at most 200 from $B$ lie within
$C'_{15},~~a_1$  lies outside $C'_1$, every vertex$w$  in a $C_i$ satisfies $200<d_w<280$,
together with 4  vertex disjoint paths
$P'_1,P'_2,P'_3,P'_4$   from $C'_1$ to $C'_{15}$ whose interior lies between  them such that the intersection of each $P'_i$ with a
$C'_j$ is a path. We contract  $C'_5$ into a set of vertices  $X'\{x'_1,...,x'_4\}$ so that  $P'_i \cap C'_5$ is in $x'_i$.

We claim that  in the graph obtained from the torso by contracting into these 8 vertices, there are four  vertex disjoint  paths from
$X$ to $X'$. If there is a 3-cut $Z$ showing this to be false
it must be disjoint from some $P_i$ and some $C_j$ with $j>5$. as well as some $P'_k$ and $C'_l$ for $l>5$.  so, $Z$ must separate  $C_j$ from  $C'_l$ and hence separates $B$ from the centers of the subdivisions
which is impossible. Now we consider, the graph $M$ obtained from this contraction of the torso by contracting these four paths to edges and deleting all the other vertices outside of both $C_5$ and $C'_5$, contracting the  components  of vertices inside $C_5$ whose distance  from $a_1$ exceeds $300$ and those within $C'_5$ whose distance from $v'$ exceeds 300.
Twice mimicking the argument  we just used, we see that there   can be no 3-cut of $M$  separating the centres of the subdivision from $B$ other than $B$  itself. Furthermore $M$ has  radius at most $1200$ and hence bounded tree
width so we are done.

If $d_{a_1}>400$ and $d_{a_1,b_1}>150$,  we proceed in almost the same way.
Now our concentric circles around $a_1$ will lie at distance between $60$ and $120$ from it.
Then having obtained a graph as we did above, we add an induced path of the torso  from $b_1$  to it.
It is easy to see that this increases the tree width by at most 2 and that this graph together with $a_1b_1$  has a $K_{3,3}$ subdivision which
cannot be separated from $B$ by any 3-cut other than $B$. we are done.

\end{proof}

\subsection{Adding Back Vertices}

The approach we take when $G_{i+1}$ is obtained from $G_i$ by deleting a set $S$ of vertices is quite similar to that
when it is obtained by deleting a set of edges. Again, we know that we will not uncontract any red or yellow vertex.

If a vertex  $s$ of $S$ sees a yellow or red vertex $v$ , then  it can only add a loop or parallel edge to $F_{i+1}$, since every other non-neighbour  of $F_{i+1}$
is separated from $v$ by a 3-cut of $G_{i+1}$ and $3<c$.  So, we can contract every such vertex into
the corresponding  $v$.  We can also add as a yellow twin for a red $w$  any vertex of $S$  which has the same neighbourhood
as $w$. We let $S'$ be the remainder of $S$.

We once again, for each yellow vertex $w$, delete $w$ and its twin set and add a clique on $N(w)$,
to obtain a 3-connected planar graph $J$ and corresponding graph $J'$ obtained by adding the vertices of
$S'$ and the edges from them. It is enough to find the desired  ferociously strong planar reduction or
two paths in $J'$.

We again build a tree decomposition using a maximal set of laminar $v^*$-non-shiftable 3-cuts in $J$,
and consider each torso. Once again, we know that all the neighbours of a vertex $v$
 in $S$ lie in $W_t$ for some $t$. If this $t$ is not unique then $v$ has three neighbours which correspond to an edge
 of our tree decomposition. But then the neighbours of $v$ form a 3-cut, separating  $v$ and  the at least one other component of
 $J'-N(v)$ from $v^*$, hence we can reduce our problem by simply deleting $v$ and everything lying below the cut in the tree decomposition (and keeping a triangle on the cut- which does not change the torso above).
 To do this in linear time, we lexicographically sort the neighbourhoods of vertices of $S$ of degree 3, and the $W_s \cap W_t$
 for the edges $st$ of the tree. We determine which edges of the tree correspond to neighbourhoods, and by traversing the tree in postorder
 find those which are not cut off from the root by any other such edge. We then perform a reduction for each of these
 edges.
After this preprocessing, all of remaining vertices of $S$ have neighbours in exactly one torso. We let $S_t$ be those
whose attachments are to $t$.

For each torso, we will either determine that we can extend the planar embedding of the torso
by adding the relevant vertices of $S$ to it, or find the desired 2-paths or ferociously strong reduction (if we are dealing with the root), or a quasi-reduction with respect to $W_p \cap W_t$ if $t$ is not a root.

Now each vertex of $S_t$ has degree at most $d$ and so, mimicking our approach when we added back edges, we can determine
in linear time, for each vertex of $S_t$ and each pair of its neighbours, the set of  at most 2 faces whose boundary  contains both of the elements of the pair (there may be two faces here because the neighbours of $S_t$ can be adjacent).
If there is any pair which do not lie  together in a face then, as in the case of adding back a matching
it follows that  there is a $K_{3,3}$ subdivision which we can use to, in linear time,  either find the torso (if $t$ is the root)
or obtain a reduction on $W_t \cap W_{p(t)}$ otherwise.

For each face $f$ we let $S_f$ be the subset of $S$ which have at least two neighbours on  $bd(f)$.
For each  $S_f$ we first preprocess by  partitioning the vertices  with 3 neighbours on $bd(f)$  according to their neighbourhoods.
We only take one element from each partition class to obtain $S'_f$.

Suppose that some element  of this partition contains a vertex $v$  of degree bigger than 3. Then  there is a $K_{3,3}$ subdivision
with  $v$ ,  another vertex of the partition element,  their common neighbours on $bd(f)$, and $v^*$. Now, any $3$-cut of the torso separating 5 of the 6 centres  from $B$
 cannot use $v$. so it cuts off a neighbour of $v$ and   all but three of the at least  $24$ 
neighbours of  this neighbour. So,    our set of $v^*$-non-shiftable cuts was not maximal.  As in the case of adding
back a matching  we can use this $K_{3,3}$  to, in linear time, either  find the desired two paths (if $t$ is the root) or obtain a quasi-reduction of the torso with respect to  $W_t \cap W_{p(t)}$ with only three green vertices otherwise.

We consider the graph with vertex set $V(bd(f)) \cup S'_f \cup x^*$ and edge set  the union of $E(bd(f))$,
the edges from $S'_f$ to $bd(f)$ and an edge from $x^*$ to every vertex of $bd(f)$. As in the  case of adding back a set of edges, we can determine in linear time  the set of $f$ such that we can add  $S'_f$  in $F$ whilst maintaining planarity.
If $S'_f$ was not planar, then we can find $x$ and $y$ in $S'_f$ and  distinct neighbours
$a,c$ of $x$ and $b,d$ of $y$ such that $a,b,c,d$ appear in the given order around the face. As in the case of adding
back a matching it follows that there is a $K_{3,3}$ subdivision which we can use, in linear time, to either  find the two paths (if $t$ is the root) or obtain a quasi-reduction of the torso with respect to  $W_t \cap W_{p(t)}$ otherwise.

So, we can add the vertices of  $S'_f$ all of whose neighbours are on $bd(f)$ to the planar embedding of the
torso. Any element of $S_f-S'_f$  all of whose neighbours are in $bd(f)$ must have 3 neighbours  all of which lie in $bd(f)$ and we can do a reduction
on these neighbours and add these vertices in also.

If the torso is $K_5$, we immediately  find the desired two paths or quasi-reduction.
Otherwise, for any other $v \in S$, there are
three neighbours $y_1,y_2,y_3$ of $v$  every two of which are on a common face but such that no face contains all three.
In this case we have a $K_{3,3}$ subdivision five of whose centres are $v,v^*, y_1,y_2,y_3$ and the sixth of which is in a component of $J_t-y_1,y_2,y_3$ not containing $v^*$. It is not hard to see that any 3-cut separating 5 of the 6 centres
of this subdivision from $v^*$ in $J$ cannot contain $v$ and so separates $v^*$ from $v$.
If there is such a 3-cut
then  since there is no   $v^*$-non-shiftable 3-cut in the torso, it has one component not containing $v^*$ which is a triangular face,
a vertex,or an edge,  Furthermore, each vertex in the component has at most three edges out of the component.
 Now, $v$ must have  at most 6 neighbours, 3 of which form such a cut.  It is an easy matter find all such cuts  in linear time. They cannot cross because every two neighbours of each vertex of $S$
are joined by $c$ paths. If  for some $v$ there is no such 3-cut then as in the case of adding
back a matching  we can use this $K_{3,3}$  to, in linear time, either  find the torso (if $t$ is the root) or obtain a reduction on
$W_t \cap W_{p(t)}$. Otherwise, it is an easy matter to reduce on all the cuts in linear time.

\subsection{Uncontracting A Matching}

Suppose  that $G_{i+1}$ was obtained from $G_i$ by contracting the edges of a matching $M$.
We recall that $M$ is   induced  which has the advantage that when we
uncontract a vertex, we do not uncontract any of its neighbours.

 We define $G'_i$ to be the graph obtained from $F_{i+1}$ by uncontracting
all the green vertices corresponding an edge of $M$.
We want to  both uncontract these green vertices to obtain a submatching $N$ of $M$
 and   either do some contracting  to obtain $F_i$ or find the desired two paths.  We will interleave these processes,
 as we now sketch.

In a first phase, we uncontract  some of the green vertices   and  determine for some of the others that
at least one endpoint  of the corresponding edge of $M$ will have to be cut off by a 3-cut of the reduction for $F_i$ which
does not contain it. We delete the corresponding edges from $N$.
For edges of the first type,we may then contract  one or both of their endpoints into a twin set.
We  do not uncontract a vertex  of the second type,   but rather add it to either to a set $Y^*$  of vertices which will be uncontracted later    and   both of whose endpoints must be cut off by some 3-cut in any reduction, or to a set $Y^+$ which will be uncontracted at a later point to an edge at least one of whose endpoints  must be cut  off by some 3-cut  in any reduction (which may contain the other endpoint).

We will ensure that at the end of the first phase, the green and red vertices still form a planar graph, the  red vertices still form a stable set of degree three nodes
with distinct neighbourhoods  and that the yellow vertices each are a twin  of some red vertex. Furthermore,
each set of a red vertex and its yellow twins is either stable or an edge.

Furthermore, we will ensure that for every  vertex which remains to be uncontracted  to a matching edge in the second phase,
each endpoint either sees all or none of every twin set and no twin set has an edge to both endpoints.
In  the last part of our uncontraction process,  we choose a set of  these vertices  to uncontract  which leaves the  graph formed by the red and green vertices  planar. Note that the properties of the last paragraph are also maintained.


Now, uncontracting  the vertex corresponding to any remaining edge of $N$  yields a nonplanar graph.
This allows us to show that for every vertex  which we chose  not to  uncontract during the second phase,
there is a local $K_{3,3}$ subdivision containing the endpoints of the edge five of whose centres
must be cut off from $v^*$  in any planar reduction.
So we add the corresponding vertex to $Y^+$. We now uncontract all the vertices  of $Y^+$ to obtain a matching $N^+$ and the vertices  of $Y^*$ to a matching $N^*$ with set of endpoints $S^*$.
We have  that our  desired reduction is a  set of  strongly ferocious 3-cuts which cuts off all of the red,  and yellow vertices, all of $S^*$, and at least one endpoint of each edge of $S^+$. We will  show that any set of cuts which does this is in fact a reduction.
This allows us to find the desired reduction or pair of paths by applying our algorithms
for graphs of  bounded tree width  locally in an auxiliary 3-connected graph.
Forthwith the details.


\subsubsection{Phase I}

In the first phase, we consider the edges of $N$ in turn, repeatedly uncontracting some, and not uncontracting others.
We will ensure that throughout the  first phase, (i) the green and red vertices still form a planar graph $J$ which is the subdivision of a planar graph where the vertices of degree 2 all 
have a red neighbour, (ii) the  red vertices still form a stable set of degree three nodes
with distinct neighbourhoods, (iii) we have a graph $J'$ obtained by adding at least one  yellow twin of each red vertex  such that  each set of a red vertex and its yellow twins is either stable or an edge, and  (iv) a graph  $J^*_i$  obtained by uncontracting some of the green vertices, such that Uncontracting any subset of these edges leaves a 3-connected graph. 

The modifications we describe can easily be performed in constant  time  per edge examined since  the vertices of $N$ all have degree at most $d$ and red and yellow vertices have degree at most four. There are three cases in which we deal with the edge in this first phase, as set out below.

We examine the edges  contracted into green vertices which have 2 non-yellow neighbours,  first. For each of these either Case 1 holds and we contract  one of its endpoints into a non-green neighbour, or Case 2 holds and we either contract both of its
endpoints into a non-green twin set, or we uncontract it into two green vertices, or we decide not to uncontract, remove it from
$N$ and do not consider it again. Furthermore, because our matching is induced we do not need to reexamine an edge,
our uncontractions and contractions cannot change the number of non yellow neighbours  of the vertex it has been contracted into.
We next check which vertices of degree three satisfy Case 1. Again since the matching is induced we do not need to reexamine
edges. Finally we check for edges satisfying Case 3. Since, our decision on these vertices is not to uncontract them, again we
do to need to reexamine any edges.

Case 1: $xy$ is an edge contracted to a green vertex $v$ for which
there is some red vertex $w$ such that (i) $y$  has only edges
to the twin set for $w$, the neighbours of $w$,   and $x$
 and (ii)  $x$ also  sees a vertex of the twin set.

We contract all the edges from $y$ to the twin set, except  for that to one twin adjacent to $x$ if $y$ sees all the twins. We obtain a bigger twin containing  $y$ which we colour red if it contains $w$ and yellow otherwise. We  obtain an embedding of our new $J$ 
from our old embedding by relabelling $v$ as $x$. Apart form this, $J$ remains unchanged.
\vskip0.2cm

Case 2: Not Case 1, and $xy$ is an edge contracted to a green vertex $v$ with  exactly two non-yellow neighbours
$w$    and  $z$.

If  $w$  is red and  $z$ is green   then
since we contracted an induced matching   $z$ is not to be uncontracted.
Since  both  $x$ and $y$ have degree 3 in $G_{i+1}$ they  both see $w$ or one of its twins.
Thus,  since we are not in Case 1,  both  $x$ and $y$ see $z$.
We contract each of the twin set  into one of $x$ or $y$, where each of $x$ and $y$ have at least one twin contracted  into them.
We colour red and leave labelled $w$ the new vertex containing $w$.
We colour the other  new vertex yellow. We modify  $J$ by suppressing the degree 2 vertex $v$,  so we do not create new cutsets of size 2.

Suppose next that both  $z$ and $w$ are red  and $x$ sees only $z$ and its twins  while  $y$ sees only $w$ and its twins.
Then, $x$ sees all the twins of $z$ and $y$ sees
all the twins of $w$, so we  uncontract $xy$ and modify $J$  by subdividing the edge $vz$,  labelling the new vertex $x$,  and relabelling $v$ as $y$.

Otherwise, since we have not already uncontracted $xy$ and every twin of $w$ or $z$ sees one of $x$ or $y$,
 we can find a matching from $xy$ into both twin sets. We let $A$ and $B$ be the two resulting three edge paths.
 We know there are three   paths  disjoint except at $v^*$, from $v^*$
to $N(z) \cup N(w)$ in $F_{i+1}$ disjoint from $\{z,w\}$ and hence also $v$. For any set of three such paths there is a $K_{3,3}$
subdivision in the union of the edge set of paths, A,B,and the edges incident to $A$ and $B$ with four centers
the endpoint of the paths, and a center in each of $A$ and $B$. Any 3-cut separating $v^*$ from the other  centers of this subdivision must intersect each of the three paths, so contains none of $A$ or $B$ or the twin sets of $w$ or $z$.  We do not uncontract $xy$ but simply record this fact by adding $v$ to $Y^*$ whilst  deleting $xy$ from $N$ and adding it to $N^*$. We associate the vertices $x$ and $y$ with this subdivision. \

Case 3:  $xy$ is an edge of $N$  corresponding to a green vertex $v$  of degree at least three such that there is a  red vertex $w$ for which both $x$ and $y$,  see $w$ or one of its twins.

We note that both $x$ and $y$ see a vertex which is neither a twin nor a neighbour of $w$ or we have
already uncontracted $xy$. We can and do  choose a twin $w'$ of $w$ so that there is a matching from $x,y$ to $w,w'$.
 For any  three paths from $v^*$ to $N(w)$ in $D_{i+1}$, there is a $K_{3,3}$ subdivision in the uncontraction of the edges of these paths
and the edges from $w$ and $w'$ whose centers are $v^*,w,w'$ and a vertex in the uncontraction of each element of $N(w)$. Any 3-cut $X$ separating the other 5 centres
of such a subdivision from $v^*$ contains a vertex of each of the three subdivision paths from $v^*$ and hence contains neither $w$ nor $w'$.  If it does not separate all the twins of $w$
from $v^*$, then it must see  a vertex in the uncontraction of each vertex of $N_{F_{i+1}} (w)$. Now, $F_{I+1}-N(v)$ contains a component containing $v^*$ with edges to all three of $N(w)$. Since $w$ also sees all three vertices of $N(w)$ and $F_{i+1}$ is 2-connected and a subdivision of a 3-connected  planar graph obtained by subdividing some edges with
at least one yellow endpoint, we see that $F_{i+1}-N(v)-v$ has a unique component. Both $x$ and $y$ have endpoints to this component and edges to one of $w$ or $w'$. So,
$X$  must separate all of the twins from $v^*$.  Again exploiting the connectivity of $F_{i+1}$ we see  that at least one of $x$ or $y$ is also cut off by $X$. So we do not uncontract $xy$ rather we add $v$ to $Y^+$ and remove $xy$ from $N$ and add it to $N^+$.
For one of the $w$s for which both $x$ and $y$ see $w$ or one of its twins,  we associate the $K_{3,3}$ discussed above
with  this vertex of  $Y^+$.

This completes the description of Phase 1.

\subsubsection{Phase II}

Now, for any edge $xy$ still in $N$, and red neighbour  $w$, either
$x$ sees all the twin set of $w$ and $y$ sees none, or $y$ sees all the twins and $x$ sees none of them.
 In considering  the remaining uncontractions, we  focus on the graph  without the yellow  twins as they will remain
twins of the red vertices no matter which edges  of $N$ we uncontract.
We thereby   obtain a  subdivision of a 3-connected  planar graph $D_{i+1}$.

We note further that for any green vertex of degree 2, seeing a red vertex $w$ and a green vertex $v$,
$v$ cannot see $w$ as then $v$ along with the
 neighbour of $w$ which is not $z$ would be a 2-cut separating $w$ from $v^*$ in $F_{i+1}$.
 We  temporarily contract all  green vertices of degree two into one of their  red neighbours (if this green vertex does
 not have a green neighbour which is ti be uncontracted we think of this as simply suppressing one or two vertices
 on an edge one endpoint of which is red).
 thereby obtaining  a 3-connected  graph $D^*_{i+1}$.  We then determine for each edge in $N$ if  an uncontraction  is possible

Following  the approach discussed in Section \ref{uncompactionfocussing}, we can determine for any $xy$ in $N$ contracted
to a vertex $v$  whether  an uncontraction leaves a  planar graph and uncontract all these edges.  We simply need to examine the rotation scheme on their neighbourhood and see how it  is split between edges incident to $x$ and $y$. Since $x$ and $y$ have degree at most $d$, this takes constant time per matching edge and linear time in total.
Furthermore, if the uncontraction does maintain planarity it is an easy matter to reinsert  any vertex $z$  of degree 2 which saw $v$  and some red vertex $w$ whilst maintaining planarity. If $z$ is adjacent to only one of $x$ or $y$ we simply subdivide the edge between this vertex and  $w$.  Otherwise, $wxy$ forms a triangular face of  the embedding,  so we can replace the edges $wx$ and $wy$  with $wz,zx,zy$   whist maintaining  planarity. We do all these uncontractions  and delete the corresponding edges from $N$.

As shown   in Section \ref{uncompactionfocussing}  for any  remaining edge $xy$
of $N$ contracted to a vertex $v$   either  (i)   $x$ and $y$  both have  exactly 3 neighbours in $D^*_{i+1}$ which are common, or (ii)  there  exist distinct $z_1,z_2,z_3,z_4$ appearing in the given order around
the boundary of the face $f$  of $D^*_{i+1}-v$ containing $v$ for which  $x$ sees $z_1$ and  $z_3$ and $y$ sees $z_2$ and $z_4$.

If (i) holds, then for each red neighbour of $v$ since we did not handle $xy$ in Case 3 of Phase 1, and by our observation
on the nonadjacency of neighbours of a degree 2 green vertex,   neither $x$ nor $y$ actually sees $w$, rather there there is a vertex $z$ of degree 2 adjacent to all of $x,y,w$. So, as in    \ref{uncompactionfocussing},
For any three paths $P_1,P_2,P_3$  from $v^*$ to the three common (green) neighbours of $x$ and $y$ avoiding $x$ and $y$ there is a $K_{3,3}$ subdivision with set of centers $x,y,v^*$ and the three common neighbours of $x$ and $y$  contained in the three paths and the edges incident to $x$ and $y$.  Any three cut separating $v^*$ from five of the six centres of this subdivision cannot use either  $x$ or $y$. So we can add $v$ to $Y^*$ and associate this subdivision with it. We add $xy$ to $N^*$. We do not uncontract $v$ and hence for each vertex of degree 2 adjacent to $v$ and $w$ we can  subdivide the edge from $v$ to $w$ to replace $z$ in the embedding.

If (ii) holds then because there  are 3 paths from $v^*$ to $v$ in $J$, we can actually choose
$z_1,z_2,z_3$ and $z_4$ so that there is  a $K_{3,3}$ subdivision   with set of centers $x,y,z_1,z_2,z_3,z_4$ which uses only  those edges
on the boundary of $f$   or from   one of  $x$ or $y$ to one of the other centers.
For a subdivision of this type,  and  any cutset  $X$ of size 3 separating $v^*$ from 5 of its 6 centers, one of $x$ or $y$ is not in $X$ and  is separated from $v^*$ by $X$. We add $v$ to $Y^+$ and associate this subdivision with it. We add $xy$ to $N^+$.
We can subdivide the edges of $v$ from its red neighbours to replace the degree 2 vertices we temporarily removed.

This completes the description of Phase 2 and the planar part of the uncontraction process.
After completing all these uncontractions, we can also unsuppress  any green vertices of degree 2 we have not yet uncontracted.

We consider  next the planar graph $J$ obtained from  all these uncontractions.
We let $J^*$ be the graph obtained from $J$ by adding back the yellow twins of the red vertices and  then
uncontracting  every vertex of $Y^+$ and $Y^*$.
We enumerate the edges of $N^+$ as $a_1b_1,...,a_lb_l$. We let $S^*$ be the set of
endpoints of the edges in $N^*$. To each vertex $v$  of $S^*$, we have associated a $K_{3,3}$ subdivision. To each red or yellow vertex $v$ , we associate a $K_{3,3}$ subdivision which has centres, $v^*$, v, one of the twins of $v$, and three vertices in the uncontractions of the neighbours of $v$ in $F_{i+1}$.  To any edge $a_ib_i$ of $N^+$ we have associated a $K_{3,3}$ subdivision. 

We let $R$ consist  of the red vertices, the yellow vertices,  and $S^*$
  We have seen that if a 3-cut is to separate 5 centres of one of the subdivisions of the last paragraph 
 from $v^*$ then it must be $v^*$-ferocious  and cut off the corresponding vertex of $R$ or one of the endpoints of the corresponding edge of $N^+$.

\subsubsection{ Strongly Ferocious 3-Cuts   Minimally Cut Off  Local Non-Planarity}

We know that if there is a ferociously strong  planar reduction then there must be a
set of $\{v^*\}$-ferocious cuts  cutting off  all of $R$
along with at least one  endpoint of each edge  of $N^+$ from $C$.

We now prove that the converse is true.

\begin{lemma}
Suppose that ${\cal F}$ is a family of  $\{v^*\}$-ferocious 3-cuts of $J^*$  such that no cut of ${\cal F}$ separates another from
$C$, every vertex of $R$  is separated from $C$ by a cut of ${\cal F}$ and at least one endpoint
of every edge of  $N^+$ is separated from $C$ by one of the cuts of ${\cal F}$. Then the graph  $J^+$ obtained from $J^*$ by deleting
the vertices it cuts off and adding edges so each cut is a clique is planar, and the desired two paths do not exist.
\end{lemma}

\begin{proof}
The graph $J$ obtained from $J^*$ by deleting all the yellow vertices and  contracting every edge of $N^+$,
and every edge of $N^*$ is planar.  We want to use this fact to show that  $J^+$ is planar.

We note that no element of ${\cal F}$ contains a red or yellow vertex. Hence no such cut can contain
the endpoints of an edge of $N^* \cup N^-$ as otherwise its contraction would correspond to a 2-cut of
$J$ without a red vertex. Thus, if  for each cut $Z$ of ${\cal F}$
we let $Z'$ be the vertices of $J$ which contain a vertex of $Z$.   $J^+$ comes from $J$  by  adding a clique on
each such $Z'$ and deleting every vertex
it cuts off.

So, to show that $J^+$ is planar, it is enough to show that for each $Z'$ there is a component of $J-Z'$ cut off by
$Z'$ which has edges to all of $Z'$.  Actually it is enough to show that there is a vertex of  $J$ cut off by $Z'$.
If this vertex is not a vertex of degree 2 then we know any cut separating it from $C$ has at least three vertices.
If it is a vertex of degree 2 then it is adjacent to a red vertex which is not in $Z'$ and we obtain the same result.
Thus, we are done unless every vertex cut off by $Z$ is contracted into $Z'$. Since our matching is induced,
this implies that each component cut off by $Z$ is a single vertex with an edge to each vertex of  $Z$, one of which is a matching edge. Since $Z$ is $\{v^*\}$ferocious there must be two such   components, and we again contradict the fact that we
contracted an induced matching.
\end{proof}

In the same vein, we have the following result.

\begin{lemma}
\label{ferociouscutlem}
Suppose that $Z$  is a  $\{v^*\}$-ferocious 3-cut  of $J^*$   then there must  either be a vertex  $v$ which is in $R$  such that $Z$ is the   $\{v^*\}$-ferocious 3-cut separating $v$ from $\{v^*\}$  maximizing the size of the component containing $v^*$, or an edge
$e$ of $N^+$ such that $Z$ is the   $\{v^*\}$-ferocious 3-cut separating at least one endpoint  from $\{v^*\}$  maximizing the size of the component containing $v^*$,\end{lemma}

\begin{proof}
We note that if a 3-cut  $Y$ of $J^*$  contains a yellow vertex and one of its twins, then   it corresponds  to a 2-cut of $J$. Because of our uncontractions in
Phase 1. there is only one (degree 2 green) vertex not  in the  component of $J^*-Y$ containing $v^*$ so $Y$ is not $\{v^*\}$-ferocious.
So,  because every red vertex of$J$ is joined to $v^*$ by 3 paths of $J$,  if a $\{v^*\}$-ferocious 3-cut $Y$ contains a yellow or red vertex, $v$  then the component  of $J^*-Y$ containing $v^*$
contains all of the twins of $v$. Thus, each of the two connected graphs which show that $Y$ is $\{v^*\}$-ferocious contain a
green neighbour of $v$ separated from these twins by $Y$.  So, for two neighbours of $w$  the uncontraction of  the neighbour contains both a vertex of $Y$  which sees a twin of   $w$,  and a vertex separated 
from $v^*$ by $Y$. 
But  deleting these $w$ and two of its neighbours from  $J$  leaves a connected graph containing $v^*$, and 
because of our uncontractions in Phase I Case 1,there are edges from this graph to the two
vertices int he uncontraction of the neighbours of $w$ supposedly cut off from $v6*$ by $Y$.
This is  a contradiction.

Thus, every $\{v^*\}$-ferocious cut $Z$  contains only green vertices.
We let $H^*$ be the  graph obtained from the subgraph of $J^*$ induced by $Z$ and the vertices it cuts off
from $v^*$ by adding a clique on $Z$ and an auxiliary vertex $v^+$ adjacent to the three
vertices of this clique.  We let ${\cal F}$ be a family of $\{v^*\}$-ferocious cuts which are not $Z$,
each of which is separated   from $v^*$ by $Z$ but not separated from $Z$ by any other $\{v^*\}$-ferocious 3-cut chosen so as
to maximize the vertices cut off from $Z$ by the elements of ${\cal F}$.

We are done unless  there is neither  a vertex of $H^* \cap R$ which is not cut off by a cut of ${\cal F}$ nor an edge of $N^+ \cap H^*$ neither of whose endpoints is cut off by such a cut. So, we can assume this is the case.
We claim that this implies that the graph $H^+$ obtained from $H^*$ by for each cut $Y$
in ${\cal F}$, putting  a clique on
$Y$ and deleting the parts of $H^*$ cut off  $Y$ is planar. This implies that it has no $K_{3,3}$ subdivision with
$Z$ the centres of one side. It is easy to see that this implies that $H^*$ has no $K_{3,3}$ subdivision  with
$Z$ the centres of one side. This is a contradiction.

Since there are three paths from $v^*$ to $Z$ in $J$,
the graph $H^-$ obtained from $H^*$ by deleting all the yellow vertices and  contracting every edge of $N^+ \cap H^*$,
and every edge of $N^* \cap H^*$ is planar.  We want to use this fact to show that  $H^+$ is planar.

We note that no cut of ${\cal F}$ contains a red or yellow vertex. Hence no such cut  can contain
the endpoints of an edge of $N^* \cup N^-$ as otherwise its contraction would correspond to a 2-cut of
$J$ without a red vertex. Thus, if  for each cut $Y$ of ${\cal F}$
we let $Y'$ be the vertices of $H^-$ which contain a vertex of $Y$.   $H^+$ comes from $H^-$  by  adding a clique on
each such $Y'$ and deleting every vertex
it cuts off.

So, to show that $H^+$ is planar, it is enough to show that for each $Y'$ there is a component of $H^--Y'$ cut off by
$Y'$ which has edges to all of $Y'$.  We proceed as in the proof of the last lemma.
\end{proof}.

\subsubsection{Finding the reduction}

We want to find a ferociously strong planar reduction or the desired 2-paths.
If $J^*$ has bounded tree width, then we can simply solve the problem directly.
Otherwise we proceed as follows:

We find   a set of $\{v^*\}$-ferocious reductions cutting off every vertex which is cut off by such a reduction.
If they cut off all of $R$ and one endpoint of every edge of $N^+$ then   they yield  a ferociously strong planar reduction.
Otherwise we find a vertex of $R$ or  an edge $a_ib_i$  of $N^+$
such that there is no $\{v^*\}$- ferocious 3-cut separating five  of the  centres
of the associated $K_{3,3}$ subdivision from $v^*$, and hence no 3-cut whatsoever doing so.

Having identified a local $K_{3,3}$ subdivision  five of whose  centres  cannot be separated  from
$v^*$ by a 3-cut, we follow the approach from the adding edges case, to find
a minor $M$  of bounded tree width which contains the desired 2-paths, and then find the
paths themselves. We omit the details of this part of the proof, and just describe how to find  our set of
$\{v^*\}$-ferocious cuts and corresponding reduction if it exists.

Our approach is to
look for our cuts locally in graphs of bounded tree width.
We start with the part of the graph ``farthest away"  from $v^*$
and work towards $v^*$.

In particular, for each vertex $w$ of $J$, we let $d_w$ be the distance of $w$ from
$v^*$ in the face vertex incidence graph of $J$. For each yellow twin $w$ of a
red vertex $v$, $d_w=d_v$. For each endpoint $w$ of an edge contracted down to
a green vertex $v$, $d_w=d_v$.

We make  the following:

\begin{observation}
For any $\{v^*\}$-ferocious 3-cut $Z$  of $J^*$, the set $Z'$ of at most  three vertices of
$J^*$ which contain vertices of $Z$ is a cutset of $J$.
\end{observation}

\begin{proof}
Otherwise the    components of $J^*-Z$ not containing $v^*$  contain no red or green vertices which are
not joined to vertices of $Z$ by edges of $N$.   If some such
component contains a yellow vertex  $v$ then $Z$ must contain a vertex in the uncontraction of each neighbour
of $v$ and $Z'$ is the neighbourhood of the red twin of $v$ in $J$ and we are done.  So, since $N$ is induced,
any such   component has only one vertex. which sees all three vertices of $Z$ and is joined to one by an edge of the matching.
Since $Z$ is ferocious, there must be two such components. This again contradicts
the fact that  $N$ is induced.
\end{proof}

\begin{corollary}
For any  $\{v^*\}$-ferocious 3-cut $Z$of $J^*$, there is an $i$ such that for every
$z$ in $Z$, $d_z \in \{i,i+2\}$.
\end{corollary}

We remark that any such cut, clearly only separates from $v^*$,   vertices $w$ with $d_w \le i$.

Now, we let $max$ be the maximum of the $d_w$.
If $max<200$ then since $J^*$  arises from a planar graph whose
face-vertex incidence graph has radius 300, by making twins of some vertices of
degree three and then uncontracting vertices to edges,  it has  tree width less than
$k=500000$.We apply Algorithm \ref{fsrtwk},
and are done.

Otherwise, we proceed in a sequence of $c=\lfloor \frac{max}{100} \rfloor-1$ iterations to create
a list  ${\cal S}$ of $\{v^*\}$-ferocious 3-cuts.

In the first iteration, we consider the subgraph  $L'_0$ of $J^*$ consisting of those vertices
with $d_w>max-200$. Again,  this graph has tree width at most $k$.
We  Apply Algorithm \ref{wfctwk}   where $W$ is the
set of vertices with $d_w \le max-150$ and $N$ is empty.We return the list of  resultant cuts and 3-colouring .

We note that a cut $Z$ of $L'_0$ disjoint from $W$ is $W$-ferocious if and only if it is a $\{v^*\}$-ferocious in $J^*$.
Furthermore, the components   of $L'_0-Z$ disjoint from $W$   are exactly the components of $J^*-Z$ not
containing   $v^*$.

We set ${\cal S}_1$ to be the  output ${\cal F}$. Since $\{v^*\}$-ferocious cuts are laminar, these cuts  cut off
every vertex which   can be cut off by a  $\{v^*\}$-ferocious cutset every element $w$ of which satisfies $d_w>max-150$.

For each cut  $Y$ of  ${\cal F}$, we  delete what it cut offs,  make a doubly linked list of these vertices,   and add two
auxiliary vertices adjacent to the three vertices of the cut. For each auxiliary vertex $w$ we set $d_w =max \{d_y |y \in Y\}$.
We have an array indexed by vertex triples and for each cut  we store a pointer to the list in this array, along with a 1 to indicate in which iteration this cut was formed.
We contract every component  $K$ of the subgraph of the resultant graph induced by the set  $\{w |~d_w>max-100\}$,  into a new vertex $v_K$.  We make a doubly linked list of the  vertices contracted into $v_K$ (this does not contain the elements cut off by any element of ${\cal F}$).

 In the resultant graph, the only cutsets of size  at most two  cutting off some vertices from $W$ contain one of these new vertices.  We can find all these cutsets and the components they cut off from $W$ in this new graph in linear time (in the size of $L'_0$) using the algorithm of Lipton and Tarjan.
 We delete them and add the corresponding edge as in that algorithm. Note that this may cause some of the contracted vertices
to be deleted. This yields a 3-connected minor, $L_1$,  of $J^*$.

 For each vertex $v_K$, for each  vertex $v$  (not cut off by a cut of ${\cal F}$)  which  was cut off from $W$ by    $v_K$,  we append to the  list for $v_k$, $v$ if it is a real vertex and the list
 corresponding to $v$ if it is a contracted vertex. If $v$ was cut off by a 2-cut we choose one of the contracted vertices
 in the cut and do the same.  We note that our auxiliary vertices cannot be involved in such cuts.

We let $R_1$ be the  union of (i) the set of those   $v_K$  which remain to which we can associate   a vertex of $R$ which was not cut off  by a cut of
${\cal F}$  and is one the list for $v_K$, or an    endpoint of an edge of $N^+$ neither of whose endpoints was cut off by such a cut which  is on the list for $v_K$,  (ii) any vertex $w$ of $R$ with $d_w >max-120$,  and (iii)  any endpoint  $w$ of an edge of $N^+$ neither of whose endpoints were cut off by a cut of ${\cal F}$ with $d_w >max-120$. 
In a natural way this  also associates a $K_{3,3}$-subdivision  with each vertex $v$ of $R_1$ such that
the  3-cut of $L_1$
separating $v$ from $v^*$, and containing an element $w$ with $d_w \ge max-150$ and closest to $v$ with this property is
the 3-cut  of $J^*$ separating  five of the centres of the $K_{3,3}$ subdivision from $v^*$ which maximizes the size of the component  containing $v^*$ (one of these cuts exists if and only if the other one does).

 Because of the way we construct $L_1$, every $\{v^*\}$-ferocious cut of $L_1$ which uses a vertex $w$ with $d_w \le max-150$ is   a $\{v^*\}$ ferocious cut of $J^*$. Also for any vertex $v$  of $R_1$, because of the associated $K_{3,3}$-subdivision,
 and our choice of ${\cal F}$, if there is a cutset of size three containing a vertex $w$ with $d_w \le max-150$ which separates
 $v$ from $v^*$, then such a cutset minimizing the size of the component containing $v$ is a $\{v^*\}$-ferocious cutset.

 We show now  that the converse holds.

 \begin{lemma}
 \label{converselemma}
 Any  $\{v^*\}$-ferocious cut  $Z$ of $J^*$ which uses a vertex $w$
 with $d_w \le max-150$ is either a $\{v^*\}$ ferocious cut of $L_1$ or is a  3-cut of  $L_1$   separating some
 $v$ in $R_1$  which minimizes the size of the component containing $v$ over all such cuts which
 use a vertex $w$ with $d_w \le max-150$.
 \end{lemma}

 \begin{proof}
We assume for a contradiction that the claim is false and consider a counterexample $Z$ chosen
so as to minimize the union $U$ of the components of $J^*-Z$ cut off by $Z$.

We let ${\cal F}_Z$ be a maximal family of $v^*$-ferocious cuts of $J^*$ which are separated from $v^*$ by $Z$
but are not separated from $Z$ by any other $v^*$-ferocious cut . By our choice of $Z$, every cut in ${\cal F}_Z$ is either
a cut of $S_1$, or a $W$-ferocious cut of $L_1$ or separates some vertex
of $R_1$ from $Z$ and contains a vertex $w$ with $d_w \le max-149$ and is closest to $v$ with this property.

We let $H_Z$ be the graph obtained from $J^*[Z \cup U]$
by putting a clique on $Z$, adding an auxiliary vertex  $v^+$ adjacent to $Z$, for each cut $Y$ in ${\cal F}_Z$
deleting what is
cut off by $Y$ and adding two auxiliary vertices adjacent to  $Y$.

By hypothesis, there is no vertex of $R_1$ in $H_Z$. Now, $H_Z
-Z$ must contain a vertex $v$ which is either in  $R$ or the endpoint of an edge of $N^+$ whose other endpoint is in $H_Z$,  or we contradict Lemma \ref{ferociouscutlem}.

The $K_{3,3}$-subdivision  of $J^*$  associated with $v$ is either completely contained in vertices with $d_w<max-100i-16$, or has one centre $v^*$ and consists of a $K_{2,3}$ formed by such vertices together with three paths from $v^*$  to the centres on the large side of the $K_{2,3}$ disjoint from the rest of the  $K_{2,3}$.

Now   $v$ either is a centre  of this subdivision or is adjacent to a  non-green vertex which is a centre which is   cut off from $v^*$ by $Z$
 Suppose next that only one centre of the subdivision is cut off from $v^*$ by $Z$.

Now, $v$ is not the endpoint of an edge  uncontracted in Phase II, whose endpoints had exactly three other common neighbours, as  there are four paths of $J^*$  from $v$ and otherwise disjoint to four of the other centres.
If $v$ were an endpoint of the other type of edge uncontracted in the second phase,
then $Z$ would have to consist of the other endpoint of the edge, and a vertex from each of the paths   from $v$ to 2 other centers in the subdivision.Both these paths are edges or paths of length 2 whose internal vertex has degree 2. So the component of 
$J^*-Z$ containing $v$  cannot be partitioned into two disjoint connected subgraphs 
each with a neighbour to every vertex of $Z$. Thus $J^*-Z$ has a second component disjoint 
from $v^*$. Since $J^*$ is 3-connected this component contains an edge to all of $Z$ and hence intersects the boundary of $f$ (defined in Phase II), hence it contains a centre of the 
subdivision other than $v$, and we are done.

If  $v$ is a red vertex then $Z$ contains a vertex in the uncontraction of each neighbour
of $v$. So it contains no other vertex, and  because of Case 1 Phase 1, every other vertex in the uncontraction of the
neighbourhood is joined to $v^*$ by a path disjoint from $Z$. So, $v$ only sees $Z$ and perhaps a twin which also sees all of $Z$.
In the latter case, we can switch to the  subdivision which has $Z$ as the centres on one side and these two twins and $v^*$ as centers in the other. This subdivision clearly appears in $L_1$ as it cannot have been cut off by any cut of ${\cal S}_1$. In this case, we are finished proving the lemma, using the two twins to show $Z$ is $\{W\}$-ferocious.   If  $v$ sees only $Z$ then  the part of $L_1-Z-v$ cut off from $W$ by $Z$ is a second connected subgraph seeing all of $Z$  and again we are done with the lemma.  The same argument works if $v$ is yellow.

If $v$ were an endpoint of an edge $vw$ from Case 2 of Phase I, then  since $Z$ contains no non-green
vertices, one of $A$ or $B$, \Wlog\ $A$,  contains a centre not cut off by $Z$, and $Z$  must contain $w$
or both would. But now there are paths  from $A$ to  three other centres disjoint from  $w$, so at least two
centres are cut off.

If $v$ were an endpoint of an edge from Case 3 of phase I, then  as set out there, $Z$ must contain a vertex which is not
in the  uncontraction of a neighbour in $J$ of  the red vertex $w$, and hence there must be some such uncontraction which is
disjoint from $Z$. Since a $\{v^*\}$-ferocious cut has no non green vertices, this means all the twins of $w$  are in the same component of $J^*-Z$, and hence cut off by $Z$ because one of them is adjacent to $v$. But two of these, $w$ and $w'$,
are  centres of the subdivision and we are done.

So, at most one centre of the subdivision is not separated from $v^*$ by $Z$ (at least two  centres on each side is impossible as we have seen). Furthermore, at most four of the six centres of the subdivision are separated from $\{v^*\}$ by any cut
$Y$ of ${\cal F}_Z$, as $v$ is not cut off by $Y$.  So, we can (i) replace the part of the subdivision cut off from $v$ by $Z$ using edges through $v^+$, and (ii) the part cut off by a cut of ${\cal F}_Z$ by edges from the auxiliary vertices we have added to it.

Now,  five of the centres of this subdivision are  not cut off by a  3-cut $Z^*$  of $H_Z$ from $Z$ as otherwise this would be a
3-cut separating off the centres of a $K_{3,3}$ subdivision of $J^*$ from $v^*$ as close to the subdivision as possible
and should have been added to ${\cal F}_Z$. So, by Claim 12, there are two disjoint connected subgraphs $C_1$ and $C_2$
of $H_Z-Z-v^+$ both of which have edges to all of $Z$. We want to transform these so that they are connected subgraphs
of $L_1$. To begin we consider their  restriction  to the non-auxiliary vertices of $H_Z$

For every cut  $Y$ of ${\cal F}_Z$, If $Y$ has more vertices in $C_i$ then in $C_{3-i}$ then we simply put  all the
vertices of the components of $L_1$ cut off by $Y$ into $C_i$. Otherwise, $Y$ has a vertex $y_1$ in $C_1$, a vertex $y_2$
in $C_2$ and a vertex $y$ in $Z$. If $Y$ is a cut of ${\cal F}$ then $L_1$ contains two auxiliary vertices cut off by $Y$ adjacent
to all of $Y$ and we add one to $C_1$ and the other to $C_2$. If $Y$ is a $W^*$-ferocious cut of $L_1$ then it cuts off two
connected subgraphs of $L_1$ adjacent to all of $Y$,  we add one to $C_1$ and the other to $C_2$

So,  $Y$ separates some vertex $v$
of $R_1$ from $Z$ and contains a vertex $w$ with $d_w \le max-149$ and is closest to $v$ with this property.
If $y$ had only one neighbour $y'$  in the component of $L_1-Y$ containing $v$  then $Y-y+y'$ would be the
3-cut separating $v$ from $v^*$ closest to $v$, so it contains no vertex with $dw \le max-150$. Furthermore,
it clearly also separates the subdivision associated with $v$ from $v^*$ so is strongly ferocious in $J^*$ and contradicts
our choice of ${\cal F}$.

Thus, $y$ has two neighbours $y'_1$ and$y'_2$ in this component. Now,  since no
2-cut of $L_!$ separates a vertex from $W$, there are disjoint paths $P_1$ and $P_2$
of $L_1-Z$ such that $P_i$ has endpoints $y_i$ and $y'_i$. We add $P_i$ to $C_i$.

So, $Z$ is a $W$-ferocious cut of $L_1$.
\end{proof}

More generally, we want to   have constructed at the end of every iteration between $i$ and $c-1$,

(A) a list  ${\cal S}_i$ of  $\{v^*\}$-ferocious cuts  such that  every element $w$ of each cut
satisfies $d_w>max-100i-50$,   which cuts off every vertex  which can be cut off by such a cut, together
with a list  for each cut $Z$ in the list of all the vertices cut off by $Z$ but not by any other 3-cut of ${\cal S}_i$,

(B) A  3-connected minor $L_i$ of $J^*$ obtained from $J^*$ by:
\begin{enumerate}
\item
first, for every element  $Y$ of ${\cal S}_i$ every vertex of which
satisfies $max-100i-50 < d_w \le max -100i+2$ which is not cut off from $v^*$ by another cut of ${\cal S}_i$,
deleting what is cut off by $Y$ and adding two auxiliary vertices of the cut such that for each auxiliary $w$,
$d_w=min \{d_y | y \in Y\}$,
\item
next contracting each  component $K$  of the graph induced  by $w$ with $d_w >max -100i$ which is left into a vertex $v_K$ ,
\item
finally finding all the 2-cuts in the resultant graph  and deleting all the components they
cut off from $v^*$ and adding edges on them (NB, any such edge must have one endpoint
in a contracted vertex).
\end{enumerate}

\vskip0.2cm

(C) the set $R_i$ of vertices of $L_i$ which either (i) were obtained by contracting some $K$ and which were not cut off by a 2-cut  in 3.  such that  there is some vertex of $R$ which is not cutoff by any element of ${\cal S_i}$ or some endpoint of an edge of $N^+$,
neither of whose endpoints are cut off by an element of ${\cal S}_i$ which is either in $K$ or in
a component cutoff by a 2-cut containing $v_K$, or (ii) are a vertex $w$  of $R$ not cut off by any element of ${\cal S}_I$,  or an endpoint of an edge of $N^+$ neither of whose endpoints was cut off by any element of ${\cal S}_i$ 
with  $d_w \ge max-100i-20$, we associate such an edge or vertex to the element of $R_i$
\vskip0.2cm

and

(D) for every non-auxiliary  vertex a list which if the vertex is not a $v_K$ of (B) is simply the vertex itself and otherwise
is the list of all vertices contracted into $v_K$  or cut off by a 2-cut having $v_K$ as an endpoint
which were not cut off by a cut of ${\cal S}_i$.

such that:

(E)  the   $\{v^*\}$-ferocious 3-cuts  of $J^*$ containing a vertex $w$ with
$d_w \le max-100i-50$ correspond to precisely the union of the $\{v^*\}$-ferocious  3-cuts in
$L_i$  containing such a vertex and the closest  3-cuts of $L_i$  containing such a vertex  to every $v$ in $R_i$.

We note that for a vertex$v$  in $R_i$ and associated vertex of $R$ or  edge of $N^+$,
the closest $\{v^*\}$-ferocious cut separating the vertex or an endpoint of the edge from $v^*$
must actually contain a vertex with $d_w \le max-100i-50$, must also separate $v$ from $v^*$,
and must be the closest such cut to $v$ which does so.

In iteration $2\le i<c$, we consider the subgraph  $L'_{i-1}$  of $L_{i-1}$  induced by those $w$ with
$d_w >max-(i+1)100$. We let $N_{i-1}$ be the  union of $R_i$ and the set of vertices with $d_w > max-100i+50$.

It is not hard to see that   this graph has tree width at most $k$.
We set $W_{i-1}$ to be those vertices of $L'_{i-1}$ with $d_w \le max-100i-50$.
If $i=c$ then we set $L'_{i-1}=L_{i-1}$ which again has tree width at most $k$, and set $W_{i-1}=\{v^*\}$.
We apply an algorithm  with the following specifications to  $L'_{i-1}$.

\begin{algspec}\label{technical}
\textsc{Technical Proof-Embedded Algorithm}
~\\
\emph{Input:}  Subgraph   $L'_{i-1}$ of a minor $L_{i-1}$ of  $J^*$ and associated sets $W_{i-1}$ and $N_{i-1}$
of its vertices,  a list ${\cal S}_i$  of cuts  and associated lists  and a   set $R_i$ of vertices of $L_i$ associated to each of which we have both a  list of vertices and a vertex of $R$ or edge of $N^+$ such that    (A), (B), (C), (D), (E)  are satisfied.
~\\
\emph{Output:}
(I) A list  ${\cal F}_i$  of   cuts each member $Z$  of which contains a vertex not in  $N_{i-1}$ and is either $W_{i-1}$-ferocious in $L'_{i-1}$ or is the 3-cut not contained in  $N_{i-1}$ which separate some $v$  of $R_{i-1}$ from $W$  minimizing  the number of vertices in the component of $L'_{i-1}-Z$ containing  $v$,  which separates off all vertices which can be separated by such cuts and then minimizes the size of ${\cal F_i}$,

(II) For every cut $Z$  of ${\cal F}_i$ which is $W_{i-1}$-ferocious in $L'_{i-1}$ two connected subgraphs  of
$L'_{i-1}$ which are in components of $Z-L'_{i-1}$ disjoint from $Z$ and have edges to every vertex of $Z$,

(III) For every cut $Z$  of ${\cal F}_i$ which separate some $v$  of $R_{i-1}$ from $W$  minimizing  the
number of vertices in the component of $L'_{i-1}-Z$ containing  $v$, the name of such a vertex $v$,

(IV) The graph $L^*_i$  obtained from $L'_{i-1}$  by deleting all the vertices  cut off from $W_{i-1}$ by a cut $Y$ of ${\cal F}_i$,
adding two auxiliary vertices $w^1_Y$ and $w^2_Y$ setting  $d_{w^i_Y}= min \{d_y | y \in Y\}$ adjacent to all of $Y$, contracting the components of the graph
induced  by $w$ with $d_w >max -100i$ which are left, then  finding all the 2-cuts in the resultant graph
one element of which is one of these contracted vertices  and deleting all the components they
cut off from $v^*$ and adding edges on them,

(V) the set $R_i$ of those vertices $w$ of  $L_i$  which either (i) are $v_K$ for some $K$  and which were not cut off by a 2-cut  such that  either there is some vertex of $R \cup R_{i-1}$ which is not cutoff by any element of ${\cal S}_i$ or some endpoint of an edge of $N^+$, neither of whose endpoints are cut off by an element of ${\cal S}_I$ which is either in $K$ or in
a component cutoff by a 2-cut containing $v_K$, or (ii) are in $R$ and were not cut off by any  cut  of${\cal F}_i$  or   are an endpoint of an edge of $N^+$  neither of whose  endpoints was cut off by such a cut and satisfy  $d_w >max-100i-20$.
With each such vertex, we associate either the corresponding vertex of $R$,
edge of $N^+$ or elements of $R \cup N^+$ which was associated with the element of $R_{i-1}$,

(VI) for every cut of ${\cal F}_i$ a list of the vertices it cuts off not cut off by any  cut of ${\cal S}_{i-1}$, and

(VII) for every vertex $v$ , a list of the vertices which have been contracted into it or have been cut off by a 2-cut containing it
and which were not cut off by any cut of ${\cal S}_i$(this may just be the vertex itself).
~\\
\emph{Running time:} $O(|L'_{i-1}|)$.
\end{algspec}

For the algorithm to be linear,we cannot copy ${\cal S}_i$ out each time, nor can we copy out the list corresponding
to each vertex. We store these as doubly linked lists and  input  them is a pointer to the first and last element on the list.
Furthermore, we cannot
 look at  all of each $L_{i-1}$ as there  may be a linear number of them,
all of which have size linear in  the size of $J^*$. Rather we must generate $L'_{i-1}$ without doing so.
This is straightforward if we have a list of  $\{w ~| ~d_w=i\}$ for every $i$ as we just want to add in the elements for the
next level of 100 values of $i$ at each step and then just focus on the graph we obtain. It is also
straightforward to make such a list in linear time.

The algorithm applies Algorithm \ref{wqfctwk} with $J=L'_{i-1}$, $W=W_{i-1}$, $N=N_{i-1}$, and $Q=R_{i-1}$.
It then determines all the components of the red and yellow graph and the vertices they attach to. This determine the
set ${\cal F}_i$ to be returned, and allows us to construct the lists as in (VI),  It also finds subgraphs as in (II). The algorithm returns vertices as in (III).

By using lexicographic bucket sort on the elements of ${\cal F}_i$ and the 3 green
attachment vertices we can then determine the set of vertices  of $L'_{i-1}$ cut off by each element  $Z$ of ${\cal F}$ and append their lists to get a list of those elements of $J^*$ cut off by $Z$. It is also   a simple matter to delete this set of vertices
and replace it with  two auxiliary vertices adjacent to all of $Z$.  We then contract the components of the graph
induced  by $w$ with $d_w >max -100i$ which are left appending their lists into a list for the new vertex $v_K$. We add
$v_K$ to $R_i$, if any of these vertices are in $R$ or $R_{i-1}$ or if both endpoints of some edge of $N^+$ are in this set.
If so, we associate the corresponding vertex $R$  or edge  of $N^+$ to $v_K$ (this may be the the one associated to an element of $R_{i-1}$)
We then    find all the cuts of size at most two  in the graph,one element of which is one of these contracted vertices
by applying the Lipton-Tarjan algorithm to the graph  obtained from this graph by adding three vertices adjacent
to all of $W_{i-1}$. We delete such components, and add an edge between any relevant 2-cut to obtain $L^*_{i-1}$.
 We  update the lists corresponding to the contracted vertices,  to obtain those required by (VII),  and
whether these vertices are in $R_i$ accordingly. It is then a simple matter to examine all the other vertices to determine
whether or not they are in $R_i$.

This completes the description of the algorithm.

We append ${\cal F}_i$ to $S_{i-1}$,  and define  $L_i$ to  be the graph obtained from $L_{i-1}$ by
replacing $L'_{i-1}$ by $L^*_{i-1}$, we have that (A),(B),(C),(D), and (E)  hold. So, we can move on to the next iteration.

Now in each iteration, we consider a graph which is obtained from a subgraph of  two consecutive levels by adding two triangles on  cutsets which cut off something in these levels, or contracting some part of some lower levels into a vertex. It is not hard to
see that the size of this graph is linear in the size of the two levels concerned and the level below them. So, the total size
of all the graphs we consider is linear.

As mentioned earlier, If $R_c$ is non empty, or there is a vertex or edge in $R \cup N^+$ in $L^*_C$,then we have a
local $K_{3,3}$  subdivision which is not cut off from $C$ by a 3-cut and hence we can find the desired two paths
in linear time.

Otherwise, as we show momentarily, there is a ferociously strong planar reduction using as cuts
the subset  ${\cal F^*}$ of those elements of  ${\cal S}_c$ which  are not separated  from $C$ by another cut in the
family. We now describe how to find these cuts, the components each cuts off from $C$, and for each $Z$ in the cut  a colouring of
the components cut off by $Z$ with red and yellow which yields $F_i$.

Finding the cuts is relatively easy, We append all the lists for all the vertices in $L^*_c$.
We find the components of the graph obtained from $J^*$ by deleting these vertices. The cuts are precisely the sets of three
vertices onto which these components attach. For each such cut $Z$ we let $H_Z$ be the subgraph of $J^*$ induced by $Z$ and the vertices it cuts off.

Each such Z  is in ${\cal S}_c$ and hence in some ${\cal F}_i$. We can determine which $i$ for each cut by using, e.g.
bucket sort on the list of these cuts and the list of the elements of all  the ${\cal F}_i$ each  with an $i$ attached. For each such cut
$Z$ we can find the corresponding list of all the vertices cut off by $Z$ and no  cut of ${\cal S}_{i-1}$. We repeat the
trick of the last paragraph to find the cuts of ${\cal S}_{i-1}$ which are cut off by $Z$ but by no cut of ${\cal S}_{i-1}$
and the set of vertices each of these cuts cut off. Putting triangles on these cuts and deleting what they cut off
from $H_Z$ yields a graph $H'_Z$. We need to find two connected subgraphs of $H_Z$ disjoint from $Z$ each of which has an edge to every vertex of $Z$.

If $Z$ is strongly ferocious in $L'_{i-1}$ then this points out a 2-colouring of the components of  $L'_{i-1}-Z$ separated from $W$ by $Z$  which naturally yields a 2-colour of $H'_Z$ with this property,by giving every vertex on the list of the contracted vertex the
same colour as a contracted vertex. We can extend this to a two colouring of $H_Z$ as follows. For any cut $Z'$ cutting off some
of $H_Z-H'_Z$, if one colour appears more than the other in $Z'$ we use it on all the vertices cut off by $Z'$. Otherwise,
$Z'$ consists of a red vertex $r$ , a yellow vertex  $y$ and a third vertex $z$.  Now, $z$ has two neighbours cut off by $Z'$
because $Z'$ is $\{v^*\}$-ferocious. There are two paths of $J^*$ from these vertices to $r$ and $y$ disjoint from $z$ because
$J^8$ is 3 connected. We colour the path containing $r$ red and the other yellow. We are done.

If $Z$ is the closest 3-cut to some vertex  $v$ of $R_i$ then it cuts off the associated local $K_{3,3}$  subdivision and we can
proceed as we did in the adding back edges case to find the desired 2-colouring by first finding a bounded tree width minor $M$
containing the subdivision but no 3-cut other than $Z$ separating it from $Z$.

\subsection{The Details of Algorithm $N$-Constrained $W$-Ferocious or $Q$-Closest Cutsets for Tree Width k}

It remains to give the details of Algorithm \ref{wqfctwk}, which we do now.

We assume that we have applied
Bodlaender's algorithm to the input and have obtained a tree decomposition of width
at most $k$ for it.

The periphery of $Z$ denoted $per(Z)$  or
$per(J(Z))$ is the set of vertices of $Z$ which have edges to $att(Z)$. We may sometimes abuse
notation by replacing Z here by the subgraph it induces.

 We will four colour the graph, red, yellow, green, and, blue so that the yellow,red, and blue vertices are those cut
 off from $W$ by the cuts of ${\cal F}$,  for any $W$-ferocious cut $Z$ of ${\cal F}$, the vertices cut off by $Z$ are coloured
 red and yellow so that there is a component of each colour, and all such components attach to all the vertices of $Z$,
 and the vertices cut off by the other cuts of ${\cal F}$ are coloured blue.
 To each $Z$ in ${\cal F}$ which is not $W$-ferocious we associate a vertex $v_Z$ of $Q$ such that $Z$ is a
 closest $N$-constrained 3-cut for $v_Z$ from $W$.

After contracting the red and yellow components in this colouring, it is  an easy matter to contract further to obtain the colouring we really want,as follows.
For any component $K$ of the red-yellow graph, if some other component $K'$
of the red-yellow graph has the same set of attachment vertices we contract each
red-yellow component with these attachments to a vertex, colouring  one red, and the rest yellow.
Otherwise, we let $l$ be a leaf of a spanning tree for $K$, we contract $K-l$
into a single vertex and colour it red and $l$ yellow.
Since we can lexicographically sort triples in linear time, it is an easy matter to do
this recolouring in linear time.

Now, for every blue component $K$ we must have a vertex, $v_K$ of $Q \cap K$ such that
$v_k$ cannot be separated from the set $Z_K$ of green vertices attached to $K$, by a cutset of size 3
other than $Z_K$ itself. Equivalently, we most have for every vertex $z$ of $Z_K$, four paths from $v_K$
to $Z_k$ disjoint except that they all have $v_K$ as an endpoint and two of them  have $z$ as an endpoint. We will
return such a vertex and corresponding set of 12 paths for each blue component.

Consider now an induced  subgraph  $J'$  of $J$.  If a 4-colouring of $J'$
could possibly  extend to  the kind of
colouring we desire then:
\begin{enumerate}
\item
Every vertex of $W \cap J'$ must be green
\item
no blue vertex can be adjacent to a red or yellow vertex,
\item
Every component of the  nongreen  subgraph of $J'$   must attach to at most   3 green vertices and exactly 3 if it is disjoint from $per(J')$,  as must every subgraph of the yellow graph  and every subgraph of the red graph.
\item
For every  component of the yellow  subgraph of $J'$  disjoint from $per(J')$
there must be either   a component of the red subgraph with the same  three green attachment vertices or a  red  component
intersecting per(J') with the same set of green
attachment vertices off $per(J')$  and no attachment vertices to which the yellow component is not attached, or all of the attachment
vertices of the component must be  in $per(J')$,
\item
 For every  component of the red  subgraph of $J'$  disjoint from $per(J')$
there must  either be a component of the yellow  subgraph with the same attachment vertices or a  yellow   component intersecting per(J')  with the same set of
attachment vertices off $per(J')$ and no attachment vertices to which the red component is not attached, or all of the attachment
vertices of the component must be  in $per(J')$,
\item
For every component $K$ of the blue subgraph disjoint from $per(J')$,  there must be a vertex $v$ of $K \cap Q$
such that letting $Z$ be the 3 green vertices to which $K$ attaches,  for every vertex $z$ of $Z$ there are four
paths as discussed above.
\end{enumerate}

We call such 3-colourings {\it valid}.

The intersection pattern for  a valid  3-colouring  consists of
\begin{enumerate}
\item
The colouring of $per(J')$
\item
The partition of  the red, yellow,and blue vertices   of $per(J') $ given by the components of the yellow, red and non-green   subgraphs of
$J'$

\item
The bipartite graph defined as follows. One side $A$  is the set of green vertices  which are either in $per(J')$ or are attachments of a component  of the non-green graph intersecting $per(J')$,
The other side consists of the  components of the yellow, red, yellow-red, and blue   graphs  which  intersect
$per(J')$  together with (i) for every triple of $A$  such  there is a component of yellow  (resp. red, yellow-red,blue) vertices disjoint  from $per(J')$
with these attachments,  a vertex corresponding to this component labelled yellow(resp. red, yellow-red,blue).
The attachment vertices in $per(J')$ are labelled with their names, the rest are labelled with numbers between $1$ and $3k$. The vertices corresponding to components intersecting $per(J')$ are labelled by
the set of vertices they contain, the rest are labelled yellow,red, yellow-red, or blue. We note that this graph has at most
$3|per(J')|$ vertices on the attachment side and at most  $4{3|per(J')| \choose 3}$ vertices on the other side.
\item
to  each blue component $K$ intersecting per(J'), we may associate a vertex  $v_K$ of $K \cap Q$.
We also have  of all possible labelled minors with labels in the green vertices to which $K$ attaches, $v_K$ if it exists,
ant the intersection of $K$ with $per(J')$ which have at most $k+8$ edges.
\end{enumerate}

For each node $t$ of the tree decomposition we  will  construct a table recording for
each intersection pattern for $G_t$,  the minimum number of green  vertices over all  valid 3-colourings
with the given intersection pattern along with pointers to the table entries  for the children of $t$ pointing out
 the  intersection patterns for  valid 3-colourings of the subgraphs cut off by these
children whose combination yields such a valid  3-colouring.  A value of $-1$ indicates no valid 3-colouring with the given intersection
pattern exists. It is not hard to see that we can  construct the table for a node given tables for its children,
in constant time. When we reach the root, we can therefore determine if there is a valid colouring,and the minimum
number of green nodes in such a colouring if one exists.

We now can use our pointers to construct such a colouring (and associate a vertex of $Q$ to each blue component) by
traversing the tree in post order. We omit the details.

\bibliographystyle{plain}
\bibliography{2drp}

\end{document}